\documentclass[11pt]{article}
\usepackage{graphicx,amssymb,amsmath,amsthm,tensor,empheq}

\usepackage{authblk}
\usepackage{mathtools}

\usepackage[all,cmtip]{xy}
\usepackage[usenames, dvipsnames]{color}
\usepackage[svgnames]{xcolor}
\usepackage[colorlinks,citecolor=RoyalBlue, urlcolor=RoyalBlue, linkcolor=RoyalBlue ]{hyperref}

\usepackage{comment}

\usepackage{url}

\usepackage{sseq}

\makeatletter

\usepackage{esint}

\newlength{\fighskip} \fighskip=2pt
\newlength{\figvskip} \figvskip=3pt

\newcommand{\be}{\begin{equation}}
\newcommand{\ee}{\end{equation}}
\newcommand{\bea}{\begin{eqnarray}}
\newcommand{\eea}{\end{eqnarray}}

\newcommand*{\wideboxed}[1]{\setlength{\fboxsep}{1ex}%
  \fbox{\m@th$\displaystyle#1$}}

\def\ubrace#1_#2{%
  \underbrace{#1}_{\hb@xt@\z@{\hss$\scriptstyle#2$\hss}}}

\newcommand\bdot{\mathbin{\mathpalette\bdot@{0.5}}}
\newcommand*\bdot@[2]{\vcenter{\hbox{\scalebox{#2}{$\m@th#1\bullet$}}}}

\newcommand{\hgf}{%
\,\tensor[_{2\kern-1.2pt}]{F}{_{\kern-0.8pt 1}}\kern-1.2pt}

\newcommand{\blangle}{\bigl\langle}
\newcommand{\brangle}{\bigr\rangle}
\newcommand{\dlangle}{\langle\kern-1.5pt\langle}
\newcommand{\drangle}{\rangle\kern-1.5pt\rangle}
\newcommand{\bdlangle}{\blangle\kern-3pt\blangle}
\newcommand{\bdrangle}{\brangle\kern-3pt\brangle}

\renewcommand{\ge}{\geqslant}

\DeclareMathOperator{\SO}{SO}

\DeclareMathOperator{\SU}{SU}

\def\ie{i.e.\ }

\def \H{\operatorname{H}}
\def\Z{{\mathbb{Z}}}
\def\TP{\mathrm{TP}}
\def\Sq{\mathrm{Sq}}
\def\B{\mathrm{B}}
\newcommand{\Spin}{{\rm Spin}}
\newcommand{\String}{{\rm String}}
\newcommand{\U}{{\rm U}}
\newcommand{\rF}{{\rm F}}
\newcommand{\C}{\mathbb{C}}
\newcommand{\R}{\mathbb{R}}
\newcommand{\cG}{ {\cal G} }

\newcommand{\eq}[1]{eq.~(\ref{#1})} 
 
\newcommand{\eqn}[1]{eqn.~(\ref{#1})} 
\newcommand{\Eqn}[1]{Eqn.~(\ref{#1})} 
\newcommand{\Sec}[1]{Sec.~\ref{#1}} 
\newcommand{\ii}{\hspace{1pt}\mathrm{i}\hspace{1pt}}
\newcommand{\dd}{\mathrm{d}}
\newcommand{\prt}{\partial} 
\newcommand{\cblue}[1]{\textcolor{black}{#1}}

\newtheorem{theorem}{Theorem}[section]

\newtheorem{remark}[theorem]{Remark}
\newtheorem{proposition}[theorem]{Proposition}

\usepackage[numbers,sort&compress]{natbib}

\newcommand{\Refe}[1]{Ref.~[\citenum{#1}]}

\newcommand{\Fig}[1]{Fig.~\ref{#1}}

\newcommand{\SM}{{\rm SM}}
\newcommand{\q}{{\rm q}}

\makeatother

\title{ \fontsize{24}{24}\selectfont{
Fermion Families and Pontryagin Class: 
} 
\\[2mm]
\fontsize{20}{20}\selectfont{Topological Field Theory
via Color Symmetry Extension
}}

\author[1]{\fontsize{14}{13}\selectfont
Zheyan Wan\footnote{wanzheyan@bimsa.cn}}

\author[2]{\fontsize{14}{13}\selectfont
Juven Wang\footnote{jw@lims.ac.uk \href{http://www.sns.ias.edu/~juven/}{http://idear.info}}}

\author[3,1]{\fontsize{14}{13}\selectfont
Shing-Tung Yau\footnote{styau@tsinghua.edu.cn}}

\affil[1]{\fontsize{10}{10}\selectfont
Beijing Institute of Mathematical Sciences and Applications, Beijing 101408, China}

\affil[2]{\fontsize{10}{10}\selectfont
London Institute for Mathematical Sciences, Royal Institution, W1S 4BS, UK}

\affil[3]{\fontsize{10}{10}\selectfont
Yau Mathematical Sciences Center, Tsinghua University, Beijing 100084, China}

\date{} 

\begin{document}

\setcounter{tocdepth}{2}

\maketitle                         
\begin{abstract}


We study 4-dimensional fermionic anomalies with discrete $\mathbb{Z}_n$ symmetry, classified by the 5d
 spin bordism group.
We show that only the anomaly from the group-cohomology subclass  $\H^5(\mathbb{Z}_n,\U(1))\cong \mathbb{Z}_n$ can be canceled by an anomalous $\mathbb{Z}_n$-symmetric 4d $\mathbb{Z}_n$-gauge topological quantum field theory (TQFT), while the beyond-group-cohomology anomaly involving 
generic $A_{\Z_n}p_1$ with Pontryagin class cannot be trivialized by any finite group extension (except $n=2,3$).
More generally, we prove that any cocycle $\alpha_d \in 
\H^d(\mathbb{Z}_n,\U(1))$ in odd spacetime dimension $d\ge3$ is trivialized by the symmetry extension
$
1 \to \mathbb{Z}_n \to \mathbb{Z}_{n^2} \to \mathbb{Z}_n \to 1,
$
and we construct explicitly the corresponding symmetric anomalous boundary TQFT.

As an application, to provide a nonperturbative global anomaly cancellation mechanism for an implication of the structure of the Standard Model (SM),
 we construct a 4d $\mathbb{Z}_N$-gauge TQFT that cancels the mixed discrete baryon-plus-lepton  $(\mathbf{B}+\mathbf{L})$-gauge-gravitational global anomaly of the generalized SM 
 with $N_f$ families  and $N_c$ colors, 
in the absence of $N_f$ families of ``sterile'' right-handed neutrinos $\nu_R$. 
In particular,
for $d=5$ and $n=3$, a 
$\mathrm{Spin}\times\mathbb{Z}_3$-symmetric 4d $\mathbb{Z}_3$-gauge TQFT can replace the 3 families of $\nu_R$ but still preserve the
3-family SM's $\Z_{6,{{\bf B} + {\bf L}}}^\rF$ symmetry. 
In general,
a 4d anomalous $\Spin \times_{\Z_2^\rF} \Z_{2 N_f,{{\bf B} + {\bf L}}}^\rF$ symmetric $\mathbb{Z}_N$-gauge 
 TQFT 
 can replace the $N_f$ families of $\nu_R$,
 via an appropriate
{\bf \emph{anomaly-trivialization symmetry extension}} construction
$$1 \to \mathbb{Z}_N\to \Spin\times \mathbb{Z}_{N N_f}\to 
\Spin \times_{\Z_2^\rF} \Z_{2 N_f}^\rF \to 1$$
of anomalous topological order
\cite{Wang2017locWWW1705.06728}. 
If $N$ and $N_f$ are minimal nonzero positive integers,
then we find minimal  extensions:
$$\left\{
\begin{array}{llll }
N=1,&N_f\ge 1,& 2 \nmid N_f,& 3 \nmid  N_f.\\ 
N=3,&N_f\ge 3,& 2 \nmid N_f,& 3 \mid  N_f.\\ 
N=4,&N_f\ge 2,& 2 \mid N_f,& 3 \nmid  N_f.\\ 
N=12,&N_f\ge 6, & 2 \mid N_f,& 3 \mid  N_f.
\end{array}\right.$$
The above symmetry
extension also constrains that if
$N_f$ is odd, then the minimal $N$ is odd.
Further, 
Witten anomaly constrains the SM 
such that 
if $N_f$ is odd, then $N_c$ must be odd.
If and only if $N_c$ and $N_f$ are odd,
the anomaly-trivialization symmetry extension construction can coincide with the baryon 
$\bf B$ to quark $\bf Q$ 
{\bf \emph{ color-center symmetry extension}}
$$
1 \to \Z_{N_c}\to
\Spin \times_{\Z_2^\rF} \Z_{2N_cN_f, {\bf Q} +N_c {\bf L}}\to 
\Spin \times_{\Z_2^\rF} \Z_{2 N_f,{\bf B +L}}^\rF\to 1.
$$
If and only if $N_c$ is odd,
then SM baryons are fermions.
So only 
with $N$, $N_c$, 
and $N_f$ all odd integers,
if further assuming $N=N_c$,
then 
$\Z_N$-gauge TQFT can coincide
with the $\SU(N_c)$ color gauge group center $Z(\SU(N_c))$. 
We prove that 3 families and 3 colors, $N=N_f=N_c=3$, is the unique minimal 
solution to have an anomalous
$\mathbb{Z}_N$-gauge TQFT 
matching the anomaly of $N_f$ of $\nu_R$.
We also prove  
the identity $A_{\mathbb Z_3}p_1=0 \mod 3$ on oriented manifolds, and $A_{\mathbb Z_2}p_1=0 \mod 2$ \emph{only} on oriented spin manifolds; for $n>3$, the analogous  $A_{\mathbb Z_n}p_1=0 \mod n$ is false in general.




\end{abstract}


\tableofcontents
\newpage

\section{Introduction}

\begin{figure}[!h]
\centering
\includegraphics[height=.5\columnwidth]{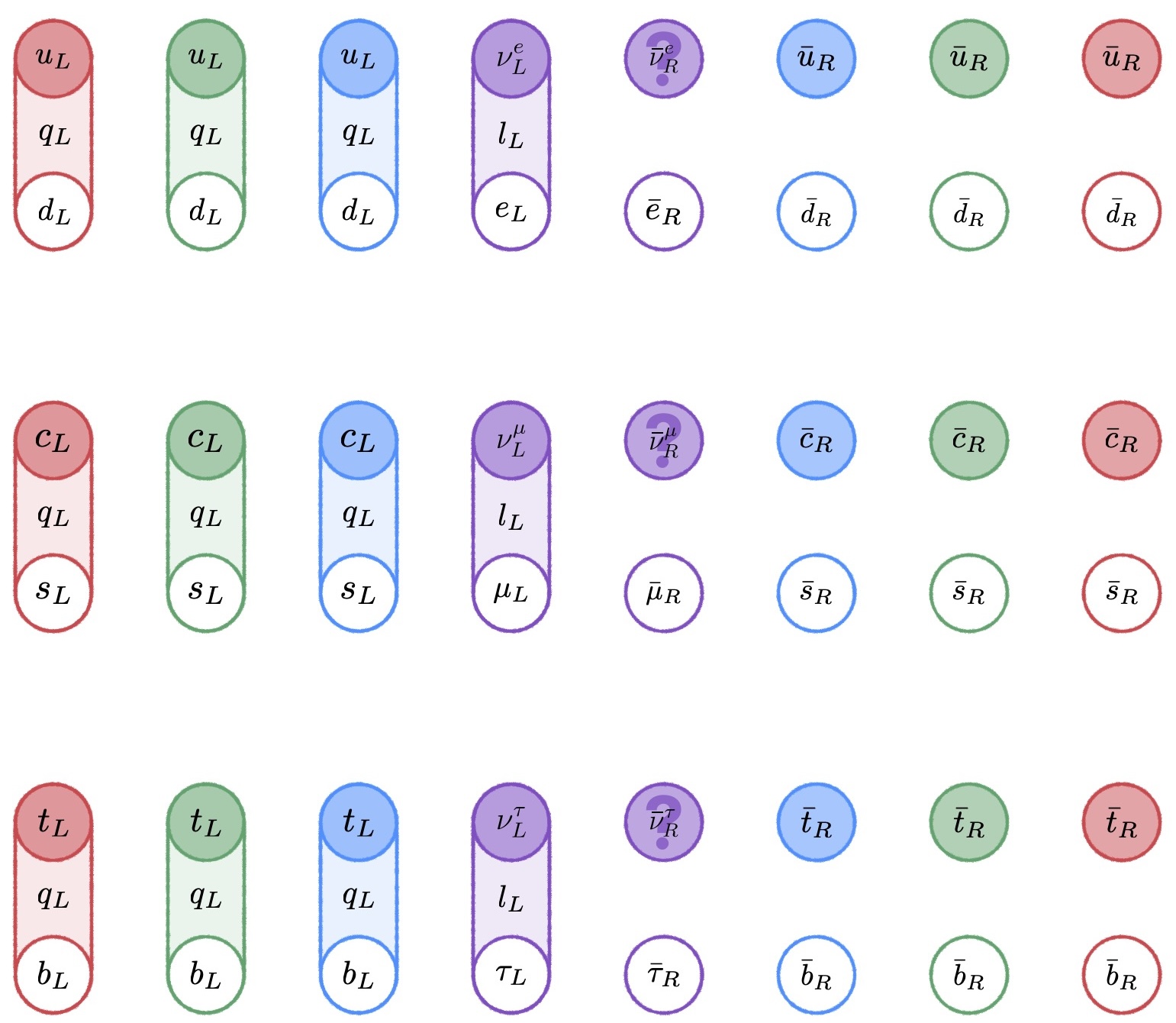}
\label{fig:3-family}
\caption{Three families ($N_f=3$) of quarks and leptons in the Standard Model (SM) of 3+1d spacetime dimensions, drawn in terms of left-handed Weyl fermions (namely each colored disk represents a  $\mathbf{2}^\C_L$ Weyl fermion representation in the 
Lorentz group Spin(1,3)).
There are 
15 Weyl fermions per family in the SM, possibly missing 1 Weyl fermion out of 16 per family,
due to the lack of ``sterile'' right-handed neutrinos $\bar{\nu}_R^e,
\bar{\nu}_R^\mu$, and $\bar{\nu}_R^\tau$ in question marks.
Following \Refe{Wang:2025oow2502.21319}'s  
proposal, we explore 
 a $N_c=3$-color symmetry extension of the discrete baryon-plus-lepton (${\bf B +L}$) symmetric gapped quantum topological order that can replace the $N_f=3$ families of massive sterile neutrinos.
More precisely, a $\Z_{N_c=3}$-gauge
fermionic 
$\Z_{2N_f=6,{\bf B +L}}^\rF$-symmetric anomalous gapped topological order
can replace the three
sterile neutrinos, with
$\Z_{2N_f=6,{\bf B +L}}^\rF = \Z_2^\rF \times \Z_{N_f=3,{\bf B +L}}$.
Namely, the nonperturbative global mixed gauge-gravitational anomaly of
$\Z_{2N_f=6,{\bf B +L}}^\rF$
symmetry of the 
45-Weyl-fermion SM (missing the 3 sterile neutrinos) can be canceled in a
full $\Z_{2N_f=6,{\bf B +L}}^\rF$-symmetry-preserving manner,  
without adding 
any massive sterile neutrinos, instead 
canceled by
including a massively gapped beyond-the-Standard-Model (BSM) \emph{topologically ordered dark sector} of the 3+1d symmetric anomalous gapped topological order with an underlying 
3+1d $\Z_{N_c=3}$-gauge
$\Z_{2N_f=6,{\bf B +L}}^\rF$-symmetric anomalous
topological quantum field theory at low energy.
The $\Z_{N_c}$ by itself is already anomaly-free and dynamically gauged.
The $\Z_{2N_f,{\bf B +L}}^\rF$ symmetry can be made anomaly-free for the combined system of the (3+1)d 45-Weyl-fermion SM 
and the (3+1)d  \emph{topologically ordered dark sector} together, thus 
$\Z_{2N_f,{\bf B +L}}^\rF$
also becomes 
dynamically gaugeable.
%
%
}
\end{figure}

One of the long-standing mysteries of theoretical physics is the origin of the three-family structure of fermion replication in the Standard Model (SM): quarks and leptons appear in exactly three families with identical gauge quantum numbers but differing masses and mixings,  observed in particle physics since the 1970s \cite{1977Harari}. While the SM itself places no restriction on the number of fermion families or generations, experimental observations indicate the family number $N_f=3$, a fact often referred to as the Family Puzzle or Generation Problem. Understanding whether this number $N_f=3$ is accidental or enforced by deeper consistency conditions remains an open theoretical puzzle 
in high-energy particle physics.

Experimentally, several independent complementary
lines of evidence establish that the number of fermion families in the SM is 
$N_f=3$ \cite{ParticleDataGroup:2024cfk}, which constrains different possible meanings of an additional generation: 

\begin{enumerate}

\item {\bf Quark-sector evidence and the completion of three observed quark families}:

The Kobayashi--Maskawa mechanism for CP violation requires at least three generations of quarks in the Standard Model \cite{Kobayashi:1973fv}. The subsequent discovery of the top quark by the CDF and D0 collaborations completed the third observed quark family \cite{CDF:1995wbb,D0:1995jca}. Thus the known quark spectrum realizes three observed quark families.

\item {\bf Invisible 
Z-width and the number of light active neutrinos}:

   Precision measurements of the invisible decay width of the $Z$ boson at LEP determine the number of light neutrino species, with masses below ($m_Z/2$), to be consistent with three \cite{OPAL:2000wza,ALEPH:2005ab}. Since each observed SM lepton family contains one left-handed neutrino, this provides direct evidence for three light active lepton families. This statement should be understood as a constraint on light neutrinos coupled to the $Z$, not by itself as a model-independent exclusion of a heavier fourth neutral lepton. A recent ATLAS measurement of the invisible $Z$ width at ($\sqrt{s}=13$) TeV is also consistent with the LEP determination and the SM prediction with three neutrino generations \cite{ATLAS:2023zrv}.

\item {\bf Electroweak precision and flavor constraints}:

   Additional sequential SM-like fermions are also constrained by electroweak precision observables, including $Z$-pole measurements, oblique corrections, and the consistency of flavor data such as CKM unitarity \cite{ALEPH:2005ab}. These measurements strongly restrict additional fermions with the gauge quantum numbers of a sequential SM family.

\item {\bf Higgs signal strengths and the exclusion of a minimal Higgs-coupled fourth chiral fermion generation}: 
   A crucial constraint on a heavy sequential fourth generation comes from Higgs physics. Here by a chiral fermion generation we mean a complete SM-like generation of quarks and leptons whose left- and right-handed Weyl fermions transform differently under the electroweak gauge group $(\SU(2)_L\times \U(1)_Y)$, so that their masses arise from Yukawa couplings to the SM Higgs rather than from gauge-invariant bare Dirac mass terms. If such additional fermions obtain their masses from the SM Higgs, their Yukawa couplings grow with their masses, and their loop contributions to Higgs production and decay are non-decoupling. In particular, additional heavy quarks would substantially enhance gluon-fusion Higgs production and modify loop-induced Higgs decay channels, such as ($h\to\gamma\gamma$), as well as the inferred Higgs branching ratios into electroweak gauge bosons. The observed Higgs signal strengths therefore strongly exclude the minimal perturbative one-Higgs-doublet Standard Model with a fourth sequential chiral fermion generation \cite{Eberhardt:2012sb,Lenz:2013iha}.

\item {\bf Cosmological support for three light relativistic neutrino species}: 

   Big-Bang Nucleosynthesis and Cosmic Microwave Background measurements constrain the effective number of relativistic species, 
   ($N_{\rm eff}$), 
   and are consistent with the Standard Model expectation from three light active neutrinos \cite{Steigman:2012ve,Fields:2014uja,Planck:2018vyg}. As with the invisible (Z)-width, these cosmological constraints primarily apply to light relativistic degrees of freedom.

\item {\bf Caveat on exotic non-Higgs mass-generation mechanisms}: 

   The preceding constraints establish the existence of three observed light SM fermion families and exclude a minimal Higgs-coupled fourth sequential chiral fermion generation. They do not, however, rule out every logically possible exotic scenario. For example, if a hypothetical additional full family did not obtain its mass from the SM Higgs and were instead gapped by a symmetry-preserving strong-dynamics mechanism such as Symmetric Mass Generation \cite{Wang:2022ucy2204.14271, Wang:2022ucy}, it would fall outside the assumptions of the minimal perturbative fourth-generation model constrained by Higgs non-decoupling.

\end{enumerate}

Theoretically, renewed attention to this problem has been prompted by proposals based on topological and nonperturbative global anomaly constraints, suggesting that the $N_f=3$ family structure may follow from fundamental mathematical consistency requirements rather than from model-dependent dynamics.
For example, 
Ref.~\cite{GarciaEtxebarriaMontero2018ajm1808.00009}
provides some interesting 
hints of the global anomalies from 
homotopy group or cobordism group constraints.
 There are two recent proposals based on topological constraints that have drawn our attention \cite{Wang:2023tbj2312.14928, 
Wang:2025oow2502.21319}.
\Refe{Wang:2023tbj2312.14928, 
Wang:2025oow2502.21319}
approach this problem in particle physics using tools from topology and topological quantum field theory (TQFT, which has \emph{no local point operators} but only \emph{extended operators}), which have been rapidly developed in recent years to describe topological quantum matter \cite{Wen2016ddy1610.03911}, thereby going beyond conventional model-building frameworks in particle physics
--- in other words, thinking outside the box of conventional particle approaches.

\begin{enumerate}
\item \Refe{Wang:2023tbj2312.14928} proposes that
when the family number $N_f$ is a multiple of 3,
\bea
N_f = 0 \mod 3, \text{ namely,} \quad N_f \in 3 \Z,
\eea
the multiple of 3 families of 16 Weyl fermions per family/generation in the SM, 
with total $(N_f=3) \times 16 =48$ Weyl fermions in the 3+1d spacetime dimensions,
are topologically constrained. 
This is due to 
\begin{enumerate}
\item 
Modular Invariance \cite{DiFrancesco:1997nkCFTConformalFieldTheory}:
The dimensional-reduced 1+1d theory from 3+1d has a chiral central charge
\bea
c_- = c_L - c_R = \frac{ N_f \times 16}{2}=  
\frac{N_f}{3} \times 48 = 0 \mod 24.
\eea
\item 
Hirzebruch signature \cite{milnor1974characteristic, Hirzebruch1966}: for a spacetime 4-manifold with a special orthogonal (SO) structure and purely bosonic gauge-invariant matter content, the signature 
$\sigma(M)$ of the manifold $M$ and its first Pontryagin class $p_1(TM)$ \cite{Pontrjagin1946, Pontrjagin1947, milnor1974characteristic}
of the tangent bundle $TM$
follow an integer-quantized relation
$
\sigma(M)=\frac{\langle p_1(TM),[M]\rangle}{3} \in \Z.
$. See more elaboration in 
\Sec{subsec:Pontryagin}.
\item Rokhlin's theorem \cite{rokhlin1952new}: for a 
spacetime 4-manifold with
a Spin structure (the fermion parity $\Z_2^\rF$ graded lift of the special orthogonal group SO, 
so $\Spin/\Z_2^\rF=\SO$), 
and fermionic gauge-invariant matter content, the signature 
$\sigma(M)$ of the manifold $M$ becomes
$\sigma(M)=\frac{\langle p_1(TM),[M]\rangle}{3} \in  16\Z.$ 
\item Cobordism mapping: 
The 48 Weyl fermions of the Standard Model, organized according to the bosonic SO and fermionic Spin structures, can be mapped to a trivial class in String cobordism (related to the framing anomaly-free \cite{Witten1988hfJonesQFT}), or more generally to a trivial class in the
$w_1$-$p_1$ cobordism (related to the 2-framing anomaly-free \cite{atiyah1990framings}). This mapping argument holds independently of any internal global symmetry or gauge structure of the SM.  
Thus, this argument \cite{Wang:2023tbj2312.14928}
may hold robustly, even if we destroy all the internal or gauge structure of the SM.

\end{enumerate}

Namely, the observation in \Refe{Wang:2023tbj2312.14928} is primarily a gravitational anomaly argument (in the dimensionally reduced 1+1d theory) or corresponds to a trivial cobordism class in a consistent quantum gravity theory \cite{McNamara1909.10355:2019rup}. 
It concerns a 
nonperturbative global gravitational anomaly argument instead of 
a perturbative local gravitational anomaly argument. 

\item \Refe{Wang:2025oow2502.21319}
proposes a unique interplay between the family and color numbers, with $N_f=N_c=3$.
This approach introduces an additional internal $\Z_3$ symmetry, naturally arising from the discrete Baryon-plus-Lepton ${\bf B +L}$ symmetry 
$\Z_{6,{\bf B +L}}^\rF$ in the
3+1d SM \cite{KorenProtonStability2204.01741, WangWanYou2204.08393}, including the fermion parity $\Z_2^\rF$ subgroup,
\bea
\Z_{6,{\bf B +L}}^\rF= \Z_2^\rF \times \Z_{3,{\bf B +L}}.
\eea

In the absence of the three right-handed sterile neutrinos $\nu_R$, the Standard Model exhibits a mixed 
${\bf (B +L)}$-gauge–gravitational nonperturbative global anomaly. The corresponding \emph{nonperturbative global anomaly index} for this 3+1d SM
as a 4d quantum field theory (QFT)
is $N_f = 3$ (up to a $\pm$ sign),
captured by a 5d bordism group
$\Omega_5^{\rm Spin\times\mathbb{Z}_3}$ class
that classifies the 
${\rm Spin\times\mathbb{Z}_3}$-structure 5d manifold up to cobordism relation:
$$
N_f = 3 \in \Z_9 =\Omega_5^{\rm Spin\times\mathbb{Z}_3}.
$$
The anomaly index $N_f = 3$
can be canceled via an appropriate
color center symmetry extension 
via
\bea 
\label{eq:Z3-extension}
1 \to \Z_{N_c=3}\to  \Z_{N_cN_f=9}\to  \Z_{N_f=3}\to 1.
\eea
Here the $\Z_{N_c=3}=Z(\SU(N_c=3))$ corresponds to the center of the color gauge group SU(3) of quantum chromodynamics (QCD).
Here the symmetry extension refers to
a particular \Refe{Wang2017locWWW1705.06728}'s
symmetry extension 
(e.g., group extension) 
construction of a {\bf \emph{symmetric gapped topological order  with a low-energy TQFT}}.\footnote{Note that the 
3+1d symmetric gapped topological order (with a low-energy 3+1d TQFT) 
of a 4+1d bulk invertible topological field theory
is a one-higher-dimensional version of the 2+1d symmetric surface topological order of the 3+1d bulk 
Symmetry-Protected Topological states (SPTs) studied in condensed matter (reviewed in \cite{Senthil1405.4015}).
Generating the symmetric topological order energy gap is sometimes known as {\bf Topological Mass Generation} that matches with a nontrivial anomaly index (\emph{without} employing any massive fermions), 
distinct from the {\bf Symmetric Mass Generation} that matches with no anomaly index in an anomaly-free system \cite{Wang:2022ucy2204.14271}.}
By \Refe{Wang2017locWWW1705.06728},
we trivialize the nontrivial anomaly index in $G = \Spin \times \Z_{N_f=3}$ by pulling it back to $G_{\rm Tot} = \Spin \times \Z_{N_cN_f=9}$ as a trivial anomaly class in 
$G_{\rm Tot}$. 

This means that the Standard Model without the $3\nu_R$
can still preserve the full $\Z_{6,{\bf B +L}}^\rF$ symmetry
by replacing $3\nu_R$ by a 
$\Z_{6,{\bf B +L}}^\rF$-symmetric
gapped topological order with a finite gauge $\Z_{N_c=3}$ TQFT at 
low-energy.
This demonstrates the uniqueness of 
\bea
N_f = N_c = 3.
\eea 
This corresponds to
 $N_f=3$ families with $N_c=3$ colors, which represent the number of quarks in a baryon {\bf B}. 
Note that
this 
$\Z_{6,{\bf B +L}}^\rF= \Z_2^\rF \times \Z_{3,{\bf B +L}}$ symmetry 
\cite{KorenProtonStability2204.01741, WangWanYou2204.08393}
is independent of the choice of the SM gauge group (see \cite{Tong2017oea1705.01853}
for an explanation of $G_{\SM_\q}$),
$$G_{\SM_\q} \equiv \frac{\SU(3) \times   \SU(2) \times \U(1)_{\tilde Y}}{\Z_\q},
\quad \quad 
\q=1,2,3,6.
$$ 
Consequently, the argument in \cite{Wang:2025oow2502.21319} that 
$N_f = N_c = 3$
is also independent of the choice of the SM gauge group
$G_{\SM_\q}$, which holds for any 
$\q=1,2,3,6$.

\end{enumerate}

In this work, on one hand, we follow the setup in \Refe{Wang:2025oow2502.21319}, to prove some of its observations differently and more mathematically; on the other hand,
we obtain some generalized theorems
and we derive some general topological constraints for the hidden topologically ordered sector
of the generic $N_f$-family $N_c$-color generalized SM.

We note that some previous works have also invoked potential nonperturbative global anomalies to constrain the 
$N_f=3$ families of the SM, but their methodology differs fundamentally from ours:
\begin{itemize}
    \item
    \Refe{Dobrescu:2001ae0102010}
uses the 6d homotopy group analysis, 
$\pi_6(\SU(2)) = \Z_{12}$, $\pi_6(\SU(3)) = \Z_6$ and $\pi_6(G_2) = \Z_3$,
to argue the nonperturbative global anomaly constraints from 6d to the 4d SM.
However, we find that the cobordism classification of nonperturbative global anomaly constraints shows $\Omega_7^{\Spin \times G_{{\rm SM}_{\rm q}}} = 0$ vanishes \cite{Davighi:2020kok2012.11693, HAHSIV},
thus this means
no 6d nonperturbative global anomalies
to constrain the 4d SM.
    \item
\Refe{GarciaEtxebarriaMontero2018ajm1808.00009}
introduces an additional $\Z_3$ symmetry with a $\Z_9$ nonperturbative global anomaly in 4d, motivated by baryon triality or proton hexality. 
However, this extra $\Z_3$ symmetry 
relies on the structure of the more sophisticated supersymmetric Standard Model.
Instead, 
\Refe{Wang:2025oow2502.21319}
and our present work
consider the simpler 
discrete 
$\Z_{3,{\bf B +L}}$ 
arising from the universal discrete 
${\bf B +L}$
symmetry of the non-supersymmetric 
SM \cite{KorenProtonStability2204.01741, WangWanYou2204.08393}.
\end{itemize}
In comparison, we believe that the $\Z_{3,{\bf B +L}}$ symmetry with a $\Z_9$ nonperturbative global anomaly \cite{Wang:2025oow2502.21319}
offers a more robust and trustworthy argument than the 3-family arguments in
the older literature \Refe{GarciaEtxebarriaMontero2018ajm1808.00009, Dobrescu:2001ae0102010}.

In practice, our present work studies fermionic anomalies associated with discrete $\Z_n$ symmetry in 3+1 dimensions \cite{1808.02881, 2506.19710}. Our focus is on distinguishing which anomaly classes can be canceled by anomalous $\Z_{n'}$-symmetric gauge TQFTs through symmetry extension constructions, and which cannot;
 for certain choices of $n$ and its corresponding $n'$. 
Similar questions about 3+1d symmetric anomalous fermionic TQFTs are explored in
 \cite{1808.02881, GuoJW1812.11959, Cordova1912.13069, Cheng:2024awi2411.05786,
Decoppet:2025eic2509.10603, 
Debray:2025kfg2510.24834, 
2512.25038, Debray:2026sqw-2602.12335}.\footnote{To clarify,
whenever we say ``fermionic'' in 
fermionic anomalies and 
 fermionic TQFTs, we mean that there are gauge-invariant fermions such that
 fermion parity $\Z_2^\rF$ is part of the internal global symmetry, namely the Spin group is part of the spacetime-internal symmetry group. 
 So fermionic TQFTs mean the TQFTs that has the Spin group symmetry such as spin TQFTs; in contrast with bosonic TQFTs that are non-spin TQFTs.}
 
Readers may also notice that on the physical mathematics side,  recent papers explore the relation between the $p_1$ structure (also the $w_1$-$p_1$ structure),
 gravitational Chern-Simons theory, and 
  topological field theories in
 \cite{2601.05518, Freed:2026fcn2603.11291},
 where  Ref.~\cite{2601.05518} further points out a relation to 
 dualizable tensor categories
\cite{Douglas:2013aea1312.7188}.

\subsection{Main results}

The main results of our present article are as follows:
\begin{enumerate}
\item
    We show that the group-cohomology subclass $\H^5(\Z_n,\U(1))\cong\Z_n$ anomaly can be canceled by 3+1d $\Z_n$-gauge TQFTs, while beyond-group-cohomology contributions involving the Pontryagin class $p_1$ generally cannot. Namely, 
    $$A_{\Z_n}(\beta_{(n,n)}A_{\Z_n})(\beta_{(n,n)}A_{\Z_n})
    $$ 
    can be trivialized via the symmetry extension  
    \bea \label{eq:Zn^2}
1\to\Z_n\to\Z_{n^2}\to\Z_n\to1,
\eea
    while
    $$A_{\Z_n}p_1$$ cannot be trivialized by any finite group extension except $n =2$ or $n=3$. Here, $A_{\Z_n}$ is the generator of $\H^1(\Z_n,\U(1))\cong \H^1(\Z_n,\Z_n)$ and $\beta_{(n,n)}:\H^*(-,\Z_n)\to\H^{*+1}(-,\Z_n)$ is the Bockstein homomorphism. See Appendices \ref{sec:trivialized} and \ref{sec:generic}.
\item
    More generally, for odd spacetime dimensions $d\ge3$, we prove that any cocycle $\alpha_d\in\H^d(\Z_n,\U(1))$ can be trivialized via the symmetry extension \eqref{eq:Zn^2},
and we explicitly construct the corresponding anomalous boundary 
symmetry-extended 
$(d-1)$d TQFTs via the explicit $(d-1)$-cochain
$\tilde{\beta}_{d-1}$ solution found in
Appendix \ref{sec:cochain}.
Following \Refe{Wang2017locWWW1705.06728}'s symmetry-extension approach,
we derive the explicit $(d-1)$-cochain $\tilde{\beta}_{d-1}$
that splits the $d$-cocycle
$\tilde{\alpha}_d=\delta\tilde{\beta}_{d-1}$
as a coboundary
in $\H^d(\Z_{n^2},\U(1))$
 by the above symmetry extension \eq{eq:Zn^2} for 
 any odd $d\ge3$ and any $n\ge2$.

\item
    As a physical application, for $d=5$ and $n=3$, we construct a
    $\Spin\times\Z_3$-symmetric 4d $\Z_3$-gauge TQFT that cancels the
    mixed discrete $(\mathbf{B}+\mathbf{L})$-gauge-gravitational anomaly
    of the Standard Model in the absence of 3 right-handed neutrinos,
    $\nu_{R,e}, \nu_{R,\mu}$,
    and $\nu_{R,\tau}$.
    More precisely, the statement $A_{\Z_3}p_1=0\mod3$ is used here as a
    mod-$3$ characteristic-class statement for the corresponding 5d
    anomaly invariant. It removes the possible Pontryagin-class
    obstruction in this $n=3$ branch, but it is not a claim that a generic
    $\U(1)$-descended mixed gravitational anomaly phase, or a generic
    beyond-group-cohomology anomaly $A_{\Z_n}p_1$, is trivial. After this
    mod-$3$ $p_1$ contribution vanishes, the remaining anomaly lies in the
    group-cohomology subclass generated by
    $A_{\Z_3}(\beta_{(3,3)}A_{\Z_3})(\beta_{(3,3)}A_{\Z_3})$. This
    remaining $\Z_3$ anomaly is trivialized by
    \bea
    1\to\Z_3\to\Z_9\to\Z_3\to1,
    \eea
    and is matched by the anomalous boundary $\Z_3$-gauge TQFT.

\item
    For a generic $N_c$-color and $N_f$-family SM, in the absence of
$N_f$ sterile right-handed neutrinos $\nu_R$, we separate the
{\bf anomaly-trivializing symmetry extension} from its possible {\bf color-center
symmetry extension}
interpretation. The anomaly is canceled by a 4d anomalous
$\Spin\times_{\Z_2^\rF}\Z_{2N_f,{\bf B}+{\bf L}}^\rF$-symmetric
$\Z_N$-gauge TQFT. 

\begin{enumerate}
    
\item 
{\bf Anomaly-trivializing symmetry extension}:
The relevant cyclic symmetry extension is
\bea 
\label{eq:NcNf-extension}
1 \to \Z_N\to\Z_{NN_f}\to\Z_{N_f}\to 1,
\eea
or, more precisely, including  spacetime symmetry and fermion parity,
\bea 
\label{eq:Spin-NcNf-extension}
\boxed{1 \to \Z_N\to
\Spin \times \Z_{NN_f}\to 
\Spin \times_{\Z_2^\rF} \Z_{2 N_f,{\bf B}+{\bf L}}^\rF\to 1}.
\eea
Here $N$ is the order of the finite gauge group $K=\Z_N$ in the anomalous
TQFT; it is determined by the anomaly and by $N_f$, and $N$ is \emph{not} assumed
a priori to be the color number $N_c$. In the physical interpretation of
the symmetry-extension construction \cite{Wang2017locWWW1705.06728}, this
normal subgroup $\Z_N$ is anomaly-free and is consistently dynamically
gauged. The minimal values of $N$ are determined by $N_f$ as stated in
Theorem \ref{main-theorem}.

\item
{\bf Color-center  baryon-to-quark symmetry extension}:
In the faithful baryon (${\bf B}$)
symmetry 
to the faithful
quark (${\bf Q}$)
symmetry
notation, 
this {\bf color-center baryon-to-quark symmetry
extension} is
\bea 
\label{eq:Spin-NcNf-extension-quark}
\boxed{1 \to \Z_{N_c}\to
\Spin \times_{\Z_2^\rF} \Z_{2N_cN_f, {\bf Q}+N_c{\bf L}}^\rF
\to 
\Spin \times_{\Z_2^\rF} \Z_{2 N_f,{\bf B}+{\bf L}}^\rF\to 1},
\eea
due to
their quantized numbers are related by
$${\bf Q}+N_c{\bf L}
=
N_c ({\bf B}+{\bf L}).$$
This color-center baryon-to-quark symmetry
extension \eq{eq:Spin-NcNf-extension-quark} exists as a fermionic symmetry extension
\emph{if and only if} $N_c$ is odd.

\item
The color-center symmetry-extension interpretation to match 
\eq{eq:Spin-NcNf-extension}
and
\eq{eq:Spin-NcNf-extension-quark},
is an \emph{additional} requirement: If the TQFT
gauge group is to coincide with the center of the $\SU(N_c)$ color gauge
group, then one must impose
\bea 
\label{eq:ZNc}
\Z_N=\Z_{N_c}=Z(\SU(N_c)).
\eea
Below we show that 
\eq{eq:Spin-NcNf-extension}
and
\eq{eq:Spin-NcNf-extension-quark}
are equivalent, 
\emph{if and only if} $N=N_c$ and
$N_f$ are odd.

The {\bf anomaly-trivializing symmetry extension} can be interpreted as
a {\bf 
color-center baryon-to-quark symmetry extension} precisely when the anomaly-determined minimal
order $N$ is set equal to the chosen color number $N_c$. At the stronger
spacetime-internal symmetry level relevant to the SM baryon-to-quark
extension, the clean faithful identification occurs in the odd
color-matched branch, namely when $N_f$ and $N_c$ are odd and $N=N_c$. In
that branch
\bea
\Spin \times_{\Z_2^\rF} \Z_{2N_cN_f,{\bf Q}+N_c{\bf L}}^\rF
=\Spin\times \Z_{N_cN_f,{\bf Q}+N_c{\bf L}},
\eea
so the extension also has the baryon-to-quark symmetry extension interpretation; the SM
baryons are fermions when $N_c$ is odd, and the $\Z_N$-gauge TQFT can coincide with the color center
$\Z_{N_c}$. 
Read the footnote for a proof on the
condition that $N=N_c$ and $N_f$ 
are odd.\footnote{Now we discuss the existence of the symmetry extensions \eqn{eq:Spin-NcNf-extension} and \eqn{eq:Spin-NcNf-extension-quark} and when they match.
\Eqn{eq:Spin-NcNf-extension} should be understood as an extension of fermionic symmetry groups, namely, the quotient map is required to preserve the distinguished fermion-parity subgroup $\mathbb Z_2^{\rF}$. Consider
\bea
1\to \mathbb Z_{N}\to \mathbb Z_2^{\rF}\times\mathbb Z_{NN_f}
\overset{\pi_a}{\longrightarrow}
\mathbb Z_{2N_f}^{\rF}
\to1.
\label{eq:general-product-extension}
\eea
A fermion-parity-preserving quotient map is necessarily of the form, where 
$\epsilon \in \{ 0, 1\} 
\cong \Z_2^{\rF}$
and $b \in \Z_{N N_f}$,
\bea
\pi_a(\epsilon,b)=\epsilon\,N_f+a\,b
\mod 2N_f,
\eea
for some integer $a \in \Z$. Such a map is well defined if and only if
\bea
\pi_a(0, N N_f) \equiv 0 
\mod 2 N_f,
\eea
which is equivalent to
\bea
aN \equiv 0 \mod 2,
\eea
and it is surjective if and only if
\bea
\gcd(a,N_f)=1.
\eea
Hence \eqn{eq:Spin-NcNf-extension} exists as a fermionic symmetry extension if and only if
\bea
N\ {\rm is\ even}\qquad\text{or}\qquad N_f\ {\rm is\ odd}.
\eea
Indeed, if $N$ is even, we can choose $a=1$; if $N_f$ is odd, we can choose $a=2$; if $N$ is odd and $N_f$ is even, such an $a$ does not exist.
On the other hand, \eqn{eq:Spin-NcNf-extension-quark},
\bea
1\to\mathbb Z_{N_c}\to\mathbb Z_{2N_cN_f}^{\rF}
\to
\mathbb Z_{2N_f}^{\rF}
\to1,
\eea
exists as a fermionic symmetry extension if and only if $N_c$ is odd, since the fermion parity in $\mathbb Z_{2N_cN_f}^{\rF}$ is represented by the element $N_cN_f$, whose 
fermion parity
image in $\mathbb Z_{2N_f}^{\rF}$ agrees with $N_f$ precisely when $N_cN_f\equiv N_f \mod 2N_f$, equivalently $N_c$ is odd.
Finally, taking $N=N_c$, when both $N=N_c$ and $N_f$ are odd,
\bea
\mathbb Z_2^{\rF}\times\mathbb Z_{N_cN_f}
\cong
\mathbb Z_{2N_cN_f}^{\rF},
\eea
by the Chinese remainder theorem, since $\gcd(2,N_cN_f)=1$. Therefore, in this case that $N=N_c$ and $N_f$ are odd, \eqn{eq:Spin-NcNf-extension} and \eqn{eq:Spin-NcNf-extension-quark} coincide as fermionic symmetry extensions.}

\item To have the Witten SU(2) anomaly-free \cite{Witten1982fp}
for the generalized SM, there is a constraint $$N_f(N_c+1) \in 2 \Z$$
(see \eq{eq:Witten}) such that
if $N_f$ is odd, then $N_c$ is odd.

Further, since $N=N_c$ and $N_f$ 
are already odd to match 
\eq{eq:Spin-NcNf-extension}
and
\eq{eq:Spin-NcNf-extension-quark},
the Witten SU(2) anomaly \cite{Witten1982fp}
indeed does not impose an additional constraint.

\item To have the 
baryon is a fermion, then $N_c$ must be odd.

\end{enumerate}

The combinations of all the five constraints from (a) to (e) above set that
$N=N_c$ and $N_f$ 
are odd, while the minimal nontrivial solution is unique as 
$$
N=N_c=N_f=3.
$$

\item We prove that $A_{\Z_n}p_1=0\mod n$ on oriented manifolds if and only if $n=3$, see Appendix \ref{sec:proof}.

\item  We prove a main theorem in \Sec{sec:NcNfSM} that we will outline here,
\begin{theorem}\label{main-theorem}
For a generalized Standard Model with $N_f$ families and $N_c$ colors,
let $N$ denote the order of a minimal finite cyclic extension $K=\Z_N$
that trivializes the anomaly of $N_f$ copies of a 4d charge-$1$ Weyl
fermion (equivalently, the anomaly of the generalized SM missing
$N_f$ sterile right-handed neutrinos $\nu_R$) with symmetry
$\Spin\times_{\Z_2^{\rF}}\Z_{2N_f,{{\bf B}+{\bf L}}}^{\rF}$ by the symmetry extension
\bea
1\to \Z_N\to \Spin\times \Z_{N N_f}\to
\Spin \times_{\Z_2^\rF} \Z_{2N_f,{{\bf B}+{\bf L}}}^\rF\to 1.
\eea
If $N$ and $N_f$ are minimal nonzero positive integers, then the minimal
extension orders are
\bea
\left\{
\begin{array}{llll}
N=1,&N_f\ge 1,& 2\nmid N_f,& 3\nmid  N_f,\\
N=3,&N_f\ge 3,& 2\nmid N_f,& 3\mid  N_f,\\
N=4,&N_f\ge 2,& 2\mid N_f,& 3\nmid  N_f,\\
N=12,&N_f\ge 6,& 2\mid N_f,& 3\mid  N_f.
\end{array}\right.
\eea
This is the anomaly-trivializing extension. Note that $N$ and $N_f$
are either both even or both odd.

At the level of the finite
kernel, it becomes a color-center extension exactly when one imposes the
additional physical identification $N=N_c$. In particular, if $N_f$ is
odd, the Witten anomaly of the generalized SM requires $N_c$ to be odd,
and the minimal anomaly-trivializing extension also has odd $N$, due 
to that
$N$ and $N_f$ have the same even-odd parity. 
In this
odd {\bf color-matched branch}, the symmetry extension coincides with the
SM-like baryon-to-quark color extension
\bea
1 \to \Z_{N_c}\to
\Spin \times_{\Z_2^\rF} \Z_{2N_cN_f, {\bf Q}+N_c{\bf L}}^\rF\to
\Spin \times_{\Z_2^\rF} \Z_{2 N_f,{\bf B}+{\bf L}}^\rF\to 1,
\eea
and the TQFT gauge group is the color center
$\Z_N=\Z_{N_c}=Z(\SU(N_c))$. With this extra color matching, baryons are
fermions and the unique minimal nontrivial solution is
\bea
N=N_c=N_f=3.
\eea
\end{theorem}

Below we sketch the proof of Theorem \ref{main-theorem}, the detailed proof can be found in the proof of Theorem \ref{main-theorem-rewrite}.
Assume that the anomaly of $N_f$ copies of a 4d charge-$1$ Weyl fermion
with symmetry
$\Spin\times_{\Z_2^{\rF}}\Z_{2N_f,{\bf B}+{\bf L}}^\rF$ is trivialized by
a minimal cyclic $\Z_N$-extension. For the generalized SM without the
$N_f$ sterile right-handed neutrinos $\nu_R$, the same anomaly is read with
an overall sign: if $n_{\nu_R}$ denotes the number of sterile
right-handed-neutrino species, then the missing-neutrino contribution is
$-N_f+n_{\nu_R}=-N_f$ when $n_{\nu_R}=0$. This sign does not affect the
minimal extension order. Write
$$
N_f=2^p 3^r s,\qquad p,r\ge0,\quad 2\nmid s,\quad 3\nmid s.
$$
The relevant anomaly has nontrivial primary components only at the primes
$2$ and $3$:
$$
\mathcal A(N_f)=
\left\{
\begin{array}{ll}
0, & p=0,\ r=0,\\
3^r\!\cdot\!\mathbb Z_{3^{r+1}}, & p=0,\ r\ge1,\\
2^p\!\cdot\!\mathbb Z_{2^{p+3}}, & p\ge1,\ r=0,\\
2^p\!\cdot\!\mathbb Z_{2^{p+3}}\ \oplus\ 3^r\!\cdot\!\mathbb Z_{3^{r+1}}, & p\ge1,\ r\ge1.
\end{array}\right.
$$
By \cite{2512.25038}, the $2$-primary component is trivialized by a
minimal $\mathbb Z_4$-extension. By Appendix \ref{sec:Z3r}, the
$3$-primary component is trivialized by a minimal $\mathbb Z_3$-extension;
this uses the mod-$3$ fact $A_{\Z_3}p_1=0\mod 3$, so that the remaining
$3$-primary anomaly lies in the group-cohomology subclass. Since the $2$-
and $3$-primary extensions are relatively prime, the simultaneous minimal
extension has order $12$ when both components are present. Hence
$$\left\{
\begin{array}{llll}
N=1,&p=0,&r=0,\\
N=3,&p=0,&r\ge1,\\
N=4,&p\ge1,&r=0,\\
N=12,&p\ge1,&r\ge1.
\end{array}\right.$$
Equivalently, this is the four-fold list of minimal nonzero values of $N$
stated in Theorem \ref{main-theorem}.

It remains to ask when this anomaly-trivializing extension is also a
color symmetry extension. At the level of the finite kernel, this happens
when the anomaly-determined $N$ is identified with the chosen color number
$N_c$. If $N_f$ is odd, then $p=0$, so the minimal extension order $N$ is
odd. The Witten anomaly of the generalized SM further requires $N_c$ to be
odd when $N_f$ is odd. In this odd color-matched branch, imposing
$N=N_c$ lets the kernel $K=\Z_N$ be read as $Z(\SU(N_c))$, and
$\Spin\times_{\Z_2^\rF}\Z_{2N_cN_f,{\bf Q}+N_c{\bf L}}^\rF$ reduces to
$\Spin\times\Z_{N_cN_f,{\bf Q}+N_c{\bf L}}$. Hence the
anomaly-trivializing extension coincides with the SM-like baryon-to-quark
color extension. The only minimal nontrivial odd branch is $N=3$, and
taking $N_f$ minimal gives
$$
N=N_c=N_f=3.
$$

\end{enumerate}
In summary, the results separate three points. 

First, finite
symmetry extension by an anomalous gauge TQFT generally cancels the
group-cohomology subclass of the discrete $\Z_n$ anomaly, not an arbitrary
beyond-group-cohomology $A_{\Z_n}p_1$ anomaly. 

Second, for the SM case
$N_f=3$, the mod-$3$ identity $A_{\Z_3}p_1=0\mod3$ removes the relevant
Pontryagin obstruction, so the remaining $\Z_3$ group-cohomology anomaly
can be canceled by the $\Z_3$-gauge TQFT from
$1\to\Z_3\to\Z_9\to\Z_3\to1$. 

Third, for a generic $N_f$-family and
$N_c$-color SM, the anomaly determines a minimal TQFT gauge group
$K=\Z_N$ with $N=1,3,4,12$ according to the $2$- and $3$-primary parts of
$N_f$; it becomes a color-center extension at the finite-kernel level only
after the additional identification $N=N_c$, and it coincides with the
SM-like baryon-to-quark color symmetry extension in the odd color-matched
branch. 

With this extra color matching and minimality, the unique
nontrivial solution is $N=N_c=N_f=3$. This framework provides a systematic
topological perspective on anomaly cancellation and offers a viewpoint on
the distinguished role of the $N_f=3$-family structure of the Standard
Model with the color number $N_c=3$, firstly observed in \Refe{Wang:2025oow2502.21319}.

\subsection{Symmetry Extension}
\label{subsec:SymmetryExtension}

Let us explain the symmetry extension method in \cite{Wang2017locWWW1705.06728}.
For 't Hooft anomalies  
of some global symmetry $G$ to be nonperturbative global anomalies, 
we can potentially apply the 
appropriate 
symmetry-extension trivialization method \cite{Wang2017locWWW1705.06728},
 making 
a nonperturbative global anomaly 
 in $G$
 becomes anomaly-free in an appropriate
 $G_{\rm Tot}$ 
 via an appropriate group extension
\bea \label{eq:extension}
1 \to K \to G_{\rm Tot} \overset{r}{\longrightarrow} 
G 
\to
1.
\eea
Namely, the precise mathematical 
check is, 
given a $G$, we search for
what an appropriate finite $K$ and
an appropriate extended $G_{\rm Tot}$ are,
such that
a nonperturbative global anomaly index 
\bea
\nu_G \in \TP_d^G
\eea 
in the Freed-Hopkins version
\cite{1604.06527} of cobordism group TP
becomes the trivial anomaly class 
\bea
(r^*\nu)_{G_{\rm Tot}} =0 \in \TP_d^{G_{\rm Tot}}
\eea
for the cobordism group TP of the pulled back 
$G_{\rm Tot}$.
Here $r$ is the reduction map from $G_{\rm Tot} \overset{r}{\longrightarrow} 
G$,
then the $r^*$ with a $*$ denote the 
pullback.
According to  
\cite{Wang2017locWWW1705.06728}, this provides 
a 
(3+1)d anomalous $G$-symmetric  
 $K$-gauge topological order construction
 whose low-energy theory is a (3+1)d 
 finite $K$-gauge TQFT, which is designed to carry the original
 nontrivial 't Hooft anomaly index
  in $G$, namely $\nu_G \in \TP_d(G)$. By $K$-gauge, we mean that $K$ is dynamically gauged with corresponding finite $K$ gauge fields,
  such as in
  \eq{eq:Spin-NcNf-extension} and \eq{eq:ZNc}, the $K=\Z_{N_c}=Z(\SU(N_c))$ is dynamically gauged.

Note that TP cobordism group
is the direct sum of a torsion bordism group in $d$ dimensions
and a free bordism group in $d+1$ dimensions
\bea
\TP_d(G) = (\Omega_d^G)_{\text{torsion}} \oplus (\Omega_{d+1}^G)_{\text{free}}.
\eea

When we focus on the nonperturbative global anomaly of 3+1d quantum field theory (QFT) with $G$ symmetry, we could use the classification of anomalies 
at the $4+1 = 5$d cobordism group
as
$$\TP_5(G) = (\Omega_5^G)_{\text{torsion}}$$ in the case that 
$(\Omega_{6}^G)_{\text{free}}=0$.
Then we need to check the anomaly index
a nonperturbative global anomaly index 
\bea
\nu_G \in \TP_5(G) = (\Omega_5^G)_{\text{torsion}},
\eea 
which
becomes the trivial anomaly class 
\bea
(r^*\nu)_{G_{\rm Tot}} =0 \in (\Omega_5^{G_{\rm Tot}})_{\text{torsion}}
\eea
for the bordism group of the pulled back 
$G_{\rm Tot}$.
In summary, when we refer to the symmetry extension,
or the symmetry extension trivialization of the 't Hooft anomaly,
what we really mean is exactly the check done in this subsection, \Sec{subsec:SymmetryExtension}.

\subsection{First Pontryagin class} \label{subsec:Pontryagin}

We briefly introduce Pontryagin classes \cite{Pontrjagin1946, Pontrjagin1947, milnor1974characteristic}, which are fundamental topological invariants of real vector bundles, analogous to the Chern classes for complex vector bundles. For a real vector bundle \(E\to B\) where 
the total space $E$ maps surjectively to the base space $B$, the total Pontryagin class is defined via the complexification \(E\otimes_{\R}\mathbb{C}\) by
\bea
p(E)=1+p_1(E)+p_2(E)+\dots ,\qquad 
p_k(E)=(-1)^k\,c_{2k}(E\otimes_{\R}\mathbb{C})\;\in \H^{4k}(B,\mathbb{Z}).
\eea
Here, $c_{2k}(E\otimes_{\R}\mathbb{C})$ are the Chern classes of the complex vector bundle \(E\otimes_{\R}\mathbb{C}\) and $\H^{4k}(B,\Z)$ is the singular cohomology of $B$, defined as the quotient group of the group of singular cocycles quotient by the group of singular coboundaries. 

In the study of index theory, gravitational instantons, and gravitational anomalies 
\cite{Eguchi:1976db, AlvarezGaume1983igWitten1984,
Alvarez-Gaume:1984zlq}, the first Pontryagin class $p_1$ plays a central role.
Geometrically, by the Chern-Weil theory, the first Pontryagin class 
$p_1$ is the integer-valued topological invariant whose image in real cohomology is represented by
(see appendices of \cite{Putrov:2023jqi2302.14862} for detailed expressions written in differential forms and the relation to the gravitational Chern-Simons 3-form)
\bea
p_1
(TM)
=-\frac{1}{8\pi^2}
\mathrm{Tr}(R \wedge R)
\in 
\H^4_{\mathrm{dR}}(M,\R)
\cong \H^4(M,\R)\cong \H^4(M,\Z)\otimes_\Z \R,
\eea
where 
$R$ is the curvature 2-form of the tangent bundle $TM$,
\cblue{the trace Tr is over the real Lie algebra valued 4-form}, and $\H^4_{\mathrm{dR}}(M,\R)$ is the de Rham cohomology of the spacetime manifold $M$. The de Rham cohomology $\H^*_{\mathrm{dR}}(M,\R)$
is defined as the quotient group of the group of closed differential forms quotient by the group of exact differential forms.
By the Chern-Weil theory, if $M_1$ and $M_2$ are the same smooth manifold with different metrics and $\alpha_1$ and $\alpha_2$ are the closed differential forms representing the first Pontryagin class of $TM$ with respect to the two different metrics, then $\alpha_1$ and $\alpha_2$ differ by an exact differential form. Therefore, $\alpha_1$ and $\alpha_2$ represent the same de Rham cohomology class and the first Pontryagin class is independent of the choice of metric. In particular, the first Pontryagin class is \cblue{invariant} under Wick rotation, regardless whether we choose the metric to be Euclidean or 
Lorentz/Minkowski 
signature, \cblue{and independent of the choice of the
real Lie algebra, Euclidean $so(4,\R)$ or Lorentz
$so(3,1,\R)$}.\footnote{
The real Euclidean rotational 
Lie algebra 
$$so(4, \R)  \equiv  so(4) \cong so(3) \oplus so(3)
\cong su(2) \oplus su(2)$$  
and the real Lorentz Lie algebra $$so(1,3, \R) \equiv so(1,3)
\cong sl(2,\C) 
$$ are both the real forms of the same complexified Lie algebra 
$$so(4,\C) 
\cong (su(2)\oplus su(2))\otimes_{\R} \C  \cong 
sl(2,\C) \oplus sl(2,\C),$$ 
the Euclidean $so(4, \R)$
and the Lorentz $so(1,3, \R)$
are related by the Wick rotation of their time coordinates $t_{\rm E} \mapsto \ii t$. 
So the Euclidean Lie algebra and the Lorentz Lie algebra have the same complexification.
{Some comments:
\begin{enumerate}
\item
The complexification of a real vector space $V$
is denoted as $V \otimes_{\mathbb{R}} \mathbb{C}$.
\item
If $V$ happens to be a 
complex vector space,
then the complexification of the underlying real vector space of a complex vector space $V$, namely $V \otimes_{\mathbb{R}} \mathbb{C}$, 
is the direct sum of two copies of $V$, namely
$V \otimes_{\mathbb{R}}  \mathbb{C}=V \oplus V$.
\item 
Note that $su(2)$ is a real vector space (not closed under the scalar multiplication by $\ii$),
$sl(2,\C)$ is a complex 
vector space (closed under the scalar multiplication by $\ii$). 
So $sl(2,\C) 
 \cong su(2) \otimes_{\R} \C 
 $ is the complexification of $su(2)$,  
and 
the complexification of the underlying real vector space of a complex vector space
$sl(2,\C)$
is
$sl(2,\C)\otimes_{\R}\C\cong sl(2,\C) \oplus sl(2,\C) 
 \cong
 (su(2) \otimes_{\R} \C) 
  \oplus
 (su(2) \otimes_{\R} \C)
 =(su(2) \oplus
 su(2) ) \otimes_{\R} \C $, although 
 $sl(2,\C) \not\cong su(2) \oplus
 su(2)$.
\end{enumerate}
}}

 For a closed oriented 4-manifold $M$, the $p_1$ integral over a closed oriented 4-manifold 
gives the Pontryagin number
\bea 
\label{eq:Pontryagin-number}
\langle p_1(TM),[M]\rangle 
=-\int_{M} \frac{1}{8\pi^2}
\mathrm{Tr}(R \wedge R)
\in \mathbb{Z},
\eea
which is intimately related to the signature \(\sigma(M)\) via the Hirzebruch signature theorem
\cite{Hirzebruch1966}: 
\bea
\sigma(M)=\frac13\langle p_1(TM),[M]\rangle.
\eea 
Here, $TM$ is the tangent bundle of $M$, $[M]$ is the fundamental class of $M$, and $\langle p_1(TM),[M]\rangle$ is the pairing which evaluates the cohomology class $p_1(TM)$ on the homology class $[M]$.

{Although Lorentzian signature does not admit real self-dual or anti-self-dual (Euclidean instanton) solutions due to the properties of the Hodge star operator, the topological charge (defined by the Chern number
$\langle c_2,[M]\rangle$
or the Pontryagin number
$\langle p_1,[M]\rangle$) remains well-defined and quantized. This topological charge quantity depends only on the topology of the bundle and represents an integral characteristic class, and is therefore quantized independently of the metric or signature.}

A crucial feature of the first Pontryagin class is that it is an oriented cobordism invariant. Two closed oriented \(n\)-manifolds $M_0$ and $M_1$ are oriented cobordant if there exists a compact oriented \((n+1)\)-manifold \(W\) whose boundary is the disjoint union \(M_0\sqcup \overline{M_1}\) (where the overbar denotes reversed orientation). A characteristic number is an oriented cobordism invariant if it depends only on the cobordism class; for \(n=4\), the Pontryagin number \(\langle p_1(TM),[M]\rangle\) is precisely such an invariant. 

In Appendix \ref{sec:proof-3}, we prove another property of the first Pontryagin class \cite{tomonaga1964mod,tomonaga1965pontryagin}: 
\bea
A_{\Z_n}p_1=0\mod n 
\eea
on oriented manifolds if and only if $n=3$,
where $A_{\Z_n}$ is the generator of $\H^1(\Z_n,\U(1))\cong\H^1(\Z_n,\Z_n)$. The cup product $A_{\Z_n}p_1$ is a mod $n$ cohomology class on a manifold $M$ when $A_{\Z_n}$ is pulled back to $M$. Throughout this article, $p_1$ always means $p_1(TM)$ for a manifold $M$.
The proof for the if part is based on the following facts. 
Let $P_3^1:\H^*(-,\Z_3)\to\H^{*+4}(-,\Z_3)$ be the mod 3 Steenrod reduced power. 
\begin{enumerate}

\item 
Based on the defining property of the mod 3 Steenrod reduced power,
\bea
P_3^1(x)=0\text{ if }\deg(x)<2.
\eea
    \item 
    This fact is nontrivial and a more general fact is proven in Appendix \ref{sec:proof-general}.
    \bea
P_3^1(x)=p_1\smile x=x\smile p_1
\eea
for any $\Z_3$-valued cohomology class $x$, where $\smile$ is the cup product.

\end{enumerate}
Since $\deg(A_{\Z_3})=1$, the above two facts imply $A_{\Z_3}p_1=0\mod3$.

The only if part can be proved by evaluating $A_{\Z_n}p_1$ on $S^1\times\mathbb{CP}^2$.
                             
\subsection{The Plan}
The plan of this article goes as follows: 

In \Sec{sec:SpinxU(1)}, to warm up,
we derive 
the 3+1d nonperturbative global anomaly of a Weyl fermion in 
$\Spin \times \Z_3$ symmetry (with 
$\Z_{3, {\bf B + L}}$ in mind)
from the reduction of the
perturbative local anomaly in $\Spin \times \U(1)$ symmetry.

In \Sec{sec:Spin^c},
we derive 
the 3+1d nonperturbative global anomaly of a Weyl fermion in
$\Spin \times_{\Z_2^\rF} {\Z_{6}^\rF}$
symmetry (with ${\Z_{6, {\bf B + L}}^\rF}$ in mind)
from the reduction of the
perturbative local anomaly in
$\Spin^c \equiv \Spin \times_{\Z_2^\rF} \U(1)$.

In \Sec{Sec:Symmetry-Extension},
we show the global anomaly trivialization
via the
{symmetry extension $1 \to \mathbb{Z}_{N_c=3}\to \mathbb{Z}_{N_cN_f=9}\to \mathbb{Z}_{N_f=3}\to 1$, and 
we explicitly construct 3+1d anomalous $\Z_3$-gauge TQFT as the low-energy theory of
topological order. 
Thus this topological order may be a hypothetical quantum dark matter candidate --- ``quantum'' in the sense that topological order is well-defined at the  0K temperature quantum limit}.

In \Sec{sec:NcNfSM}, we generalize to a generic $N_c$-color and
$N_f$-family SM. In the absence of $N_f$ sterile right-handed neutrinos,
we construct a 4d anomalous
$\Spin\times_{\Z_2^\rF}\Z_{2N_f,{\bf B}+{\bf L}}^\rF$-symmetric
$\Z_N$-gauge TQFT that cancels the mixed discrete
$({\bf B}+{\bf L})$-gauge-gravitational global anomaly via the symmetry
extension
$$
1 \to \Z_N\to
\Spin \times \Z_{NN_f}\to 
\Spin \times_{\Z_2^\rF} \Z_{2 N_f,{\bf B}+{\bf L}}^\rF\to 1.
$$
We determine the minimal extension order $N$ as a function of $N_f$:
$N=1,3,4,12$ according as $N_f$ is divisible by neither $2$ nor $3$, by
$3$ but not $2$, by $2$ but not $3$, or by both $2$ and $3$. This section
also explains that $N$ is not assumed to equal the color number $N_c$;
only in the color-matched odd case with $N=N_c$ does the $\Z_N$-gauge TQFT
coincide with the center $Z(\SU(N_c))$, leading to the unique minimal
solution $N=N_c=N_f=3$.

In addition, we will prove various mathematical theorems in the Appendices
that will be implemented in the main text.

In Appendix \ref{sec:trivialized},
{we prove that any group cocycle 
$\alpha_d \in \H^d(\Z_n,\U(1))$ is trivialized by the symmetry extension $1\to\Z_n\to\Z_{n^2}\to\Z_n\to 1$ 
for odd $d\ge3$ and any $n\ge2$}.

In Appendix \ref{sec:generic}, we prove that $A_{\Z_n}p_1$ cannot be trivialized by any finite group extension except for $n=2$ and $n=3$ where $A_{\Z_n}$ is defined as the generator of $\H^1(\Z_n,\U(1))$.

In Appendix \ref{sec:cochain},
following \Refe{Wang2017locWWW1705.06728}'s symmetry-extension approach,
we derive the explicit $(d-1)$-cochain $\tilde{\beta}_{d-1}$
that splits the $d$-cocycle
$\tilde{\alpha}_d=\delta\tilde{\beta}_{d-1}$
as a coboundary
in $\H^d(\Z_{n^2},\U(1))$
 by the symmetry extension $1\to\Z_n\to\Z_{n^2}\to\Z_n\to1$ for 
 any odd $d\ge3$ and any $n\ge2$.

In Appendix \ref{sec:dd-d-1d},
we construct the path integral of
{$d$d-bulk/$(d-1)$d-boundary coupled invertible topological field theory/symmetric anomalous gapped TQFT
by the symmetry extension $1\to\Z_n\to\Z_{n^2}\to\Z_n\to 1$ for any odd $d\ge3$ and any $n\ge2$}.

In Appendix \ref{sec:anomaly-nmid23},
we derive a {3+1d nonperturbative global anomaly formula of Weyl fermion in 
$\Spin \times {\Z_{n}}$ symmetry for integer $n$ with $2 \nmid n$ and $3 \nmid n$}.

In Appendix \ref{sec:Z3r},
we derive a  {3+1d nonperturbative global anomaly formula of Weyl fermion
$\Spin \times \Z_{3^r} = \Spin \times_{\Z_2^\rF} {\Z_{2 \cdot 3^r}^\rF}$ symmetry}.

In Appendix \ref{sec:Z2p},
we derive a  {3+1d nonperturbative global anomaly formula of Weyl fermion
$\Spin \times_{\Z_2^\rF} {\Z_{2 \cdot 2^p}^\rF}$ symmetry}.

In Appendix \ref{sec:proof}, we prove that $A_{\Z_n}p_1=0\mod n$ on oriented manifolds if and only if $n=3$.

\section{Perturbative Local Anomaly  to
Nonperturbative Global Anomaly}

In the 3+1d Standard Model (SM), 
each family contains 
15 Weyl fermions in the absence 
of 
the 16th Weyl fermion sterile right-handed neutrino $\nu_R$.
This SM 
suffers from the perturbative local
mixed-gauge-gravitational anomalies 
\cite{Eguchi:1976db, AlvarezGaume1983igWitten1984,
Alvarez-Gaume:1984zlq}
between the lepton number ${\bf L}$ symmetry and gravitational background fields, in 3+1d (or simply 4d) spacetime. 
Namely these anomalies are computable
via perturbative triangle Feynman diagrams
$\U(1)_{\bf L}^3$ and $\U(1)_{\bf L}$-gravity-gravity, 
$$
\includegraphics[height=.15\columnwidth]{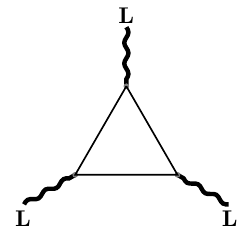}
\quad \quad \quad 
\includegraphics[height=.15\columnwidth]{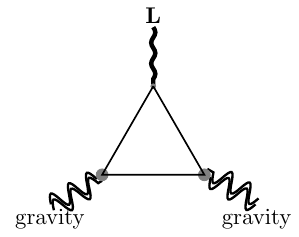}
$$
with the anomaly index coefficient 
$$-N_f + n_{\nu_R},$$
counting the difference between the family or generation number $N_f$ (typically $N_f=3$) and the total right-hand neutrino number $n_{\nu_R}$. See recent related expositions about this anomaly index $-N_f + n_{\nu_R}$ for examples
in  \cite{JW2006.16996, JW2008.06499, JW2012.15860, WangWanYou2112.14765, WangWanYou2204.08393, Putrov:2023jqi2302.14862, Wang:2024auy2501.00607}.
However, because of
the analogous Adler-Bell-Jackiw anomalies
\cite{Adler1969gkABJ, Bell1969tsABJ}
via the SM electroweak gauge instanton
\cite{BelavinBPST1975, tHooft1976ripPRL, JackiwRebbi:1976pf, 
CallanDashenGross:1976je},
instead of thinking of the classical lepton number ${\bf L}$ symmetry,
only the baryon number plus or minus lepton number ${\bf B} \pm {\bf L}$ symmetries are physically meaningful quantum mechanical symmetries of the SM \cite{KorenProtonStability2204.01741, WangWanYou2204.08393}:

\begin{itemize}
\item 
For the gauge-invariant baryons, 
a full 
faithful 
combined symmetry of ${\bf B} - {\bf L}$ and ${\bf B} + {\bf L}$ with the
Lorentz spacetime Spin group symmetry is
\cite{WangWanYou2204.08393}
\bea \label{eq:SpinU1Z2NF}
\Spin \times_{\Z_2^\rF} \U(1)_{{\bf B} - {\bf L}} \times_{\Z_2^\rF}\Z^\rF_{2 N_f, {\bf B} + {\bf L}}.
\eea
\item 
For the free quarks, a full 
faithful 
combined symmetry of ${\bf Q} - N_c {\bf L}$ and ${\bf Q} + N_c {\bf L}$
with the
Lorentz spacetime Spin group symmetry is 
\bea
\Spin \times_{\Z_2^\rF} \U(1)_{{\bf Q} - N_c {\bf L}} \times_{\Z_2^\rF}\Z^\rF_{2 N_c N_f, {\bf Q} + N_c {\bf L}},
\eea
but it is unfaithful for the gauge-invariant baryons.
Here in the conventional SM, the family number is $N_f=3$, 
and the color number is $N_c=3$.
\end{itemize}

For the conventional SM with $N_f=3$,
below we determine the $\Spin \times_{\Z_2^\rF} \Z_{6,{{\bf B} + {\bf L}}}^\rF =
\Spin \times \Z_{3,{{\bf B} + {\bf L}}}$ mixed gauge-gravitational anomaly for the
right-handed ``sterile'' neutrino $\nu_R$ (sterile to 
the SM gauge force but not sterile to ${\bf B} \pm {\bf L}$ gauge field).
In fact, we shall treat the anti-particle of right-handed neutrino $\bar{\nu}_R$ as the left-handed particle, whose quantum numbers (here discrete charges of abelian global symmetries) are given by:
\bea
\begin{tabular}{| c  | c |  c |  c |  c | c | c |  c |  c |  c | }
\hline
  &   $\U(1)_{{\bf B}-{\bf L}}$ & 
  $\Z_{6,{{\bf B} + {\bf L}}}^\rF$ & $\Z_{3,{{\bf B} + {\bf L}}}$
  & $\Z_2^\rF$ & 
  $\Z_{18, {\bf Q} + 3 {\bf L}}^\rF$ & $\Z_{9, {\bf Q} + 3 {\bf L}}$
  \\
\hline
\hline
$\bar{\nu}_R$ &   1 & $-1$ & $-1$ & $1$ & $-3$  & $-3$\\
\hline
\end{tabular}.
\eea
So the charge 
$Q_{\Z_{6,{{\bf B} + {\bf L}}}^\rF} = Q_{\Z_{3,{{\bf B} + {\bf L}}}} \mod 3$,
and 
$Q_{\Z_{18,{{\bf B} + {\bf L}}}^\rF} = Q_{\Z_{9,{{\bf B} + {\bf L}}}} \mod 9$.\footnote{{We
label $n_6 \in \Z_6 = \Z_6^\rF \supset  \Z_2^\rF$ in terms of a doublet $(n_2^\rF, n_3) \in  
\Z_2^\rF \times \Z_3$,
such that the bosons have $n_2^\rF=0$
and the fermions have $n_2^\rF=1$. 
In addition, without loss of generality, 
we assign the charge $q=1 \in \Z_6^\rF$ fermion
to the $(n_2^\rF, n_3)=(1,1) \in  \Z_2^\rF \times \Z_3 $. This constrains
the map as $n_6 = 3 n_2^\rF - 2 n_3$,
so $n_6 = n_3 \mod 3$.
\\
{We label $n_{18} \in \Z_{18} = \Z_{18}^\rF \supset  \Z_2^\rF$ in terms of a doublet 
$(n_2^\rF, n_9) \in \Z_2^\rF \times  \Z_9$,
such that the bosons have $n_2^\rF=0$
and the fermions have $n_2^\rF=1$. 
In addition, without loss of generality, 
we assign the charge $q=1 \in \Z_{18}^\rF$ fermion
to the $( n_2^\rF, n_9)=(1,1) \in \Z_2^\rF \times  \Z_9$. This constrains
the map as $n_{18} = 9 n_2^\rF - 8 n_9$, 
so $n_{18} =  n_9 \mod 9$. 
}
}}

In Subsection \ref{sec:SpinxU(1)}, 
we start with a perturbative local anomaly of U(1) charge $q=1$ left-handed Weyl fermion 
in $\Spin \times \U(1)$ to
derive 
the nonperturbative global anomaly 
of $\Z_{3, {\bf B + L}}$ charge $q=1$ Weyl fermion 
in 
$\Spin \times \Z_{3, {\bf B + L}}$.

In Subsection \ref{sec:Spin^c}, 
we will make a comparison to a {perturbative local anomaly in $\Spin^c \equiv \Spin \times_{\Z_2^\rF} \U(1)$ and
a nonperturbative global anomaly in 
$\Spin \times_{\Z_2^\rF} {\Z_{6,{\bf B + L}}^\rF}$}.

\subsection{$\Spin \times \U(1)$ to
$\Spin \times \Z_3$}

\label{sec:SpinxU(1)}

The perturbative local anomaly of U(1) charge $q=1$ left-handed Weyl fermion 
of $\Spin \times \U(1)$ symmetry
in 3+1d or 4d
is captured by a 5d invertible field theory (iTFT) with the anomaly index $k=1$ \cite{AlvarezGaume1983igWitten1984,
Putrov:2023jqi2302.14862},
as an invertible U(1)-valued partition function:
\begin{equation}
   \exp( \ii k  \int_{M^5} A \frac{c_1^2}{6}-A \frac{p_1}{24}),
\end{equation} 
with the first Chern class $c_1$
and the first Pontryagin class $p_1$.
Now we redefine the U(1) gauge field $A$ as a $\Z_3$ cohomology class gauge field $ A_{\Z_3} \in\H^1(\B\Z_3,\Z_3)= \Z_3$
with the following replacement:
\bea
A &\mapsto& \frac{2 \pi}{3}   A_{\Z_3}.\cr 
c_1 = \frac{\dd A}{2 \pi} &\mapsto& 
\frac{\dd   A_{\Z_3}}{3} \equiv \beta_{(3,3) } A_{\Z_3}. 
\eea
The $\beta_{(n,m)}: \H^*(-,\Z_m) \mapsto 
\H^{*+1}(-,\Z_n) $ 
is the Bockstein 
 homomorphism associated with the extension
 $\Z_n \stackrel{\cdot m}{\to} \Z_{nm} \to \Z_m$. Thus we get the 5d topological invariant of the $\Spin \times \Z_3$ that captures the 4d anomaly as:
\bea\label{eq:Spin-Z3}
   &&\exp \big( \ii 2 \pi k  \int_{M^5} 
   ( \frac{1}{18} 
   {A_{\Z_3}} (\beta_{(3,3)} {A_{\Z_3}}) (\beta_{(3,3)} {A_{\Z_3}})
   -\frac{1}{3 \cdot 24} {A_{\Z_3}} p_1 ) \big) \cr
   &=&\exp \big( \ii \frac{2 \pi}{9} k  \int_{M^5} ( \frac{1}{2} 
   {A_{\Z_3}} (\beta_{(3,3)} {A_{\Z_3}}) (\beta_{(3,3)} {A_{\Z_3}})
   -\frac{1}{8} {A_{\Z_3}} p_1 ) \big)\cr
   &=&\exp \big( \ii \frac{2 \pi}{9} k  \int_{M^5} ( -4
   {A_{\Z_3}} (\beta_{(3,3)} {A_{\Z_3}}) (\beta_{(3,3)} {A_{\Z_3}})
   +3\cdot {} \frac{A_{\Z_3}p_1}{3} ) \big)\cr
   &=&\exp \big( \ii \frac{2 \pi}{9} (-4k)  \int_{M^5} ( 
   {A_{\Z_3}} (\beta_{(3,3)} {A_{\Z_3}}) (\beta_{(3,3)} {A_{\Z_3}})
   -3\cdot {} \frac{A_{\Z_3}p_1}{3} ) \big).
\eea 
Here $\beta_{(3,3)}:\H^1(-,\Z_3)\to \H^2(-,\Z_3)$ is the Bockstein homomorphism.
Here ${} \frac{A_{\Z_3}p_1}{3}$ is a mod 3 class that involves Pontryagin class because $A_{\Z_3}p_1=0\mod3$ \cite{tomonaga1964mod,tomonaga1965pontryagin} (see Appendix \ref{sec:proof} for the proof),
while ${A_{\Z_3}} (\beta_{(3,3)} {A_{\Z_3}}) (\beta_{(3,3)} {A_{\Z_3}})$
is a mod 3 class.

In \eqref{eq:Spin-Z3}, the first equality rewrites the coefficients $\frac{1}{18}=\frac{1}{9}\cdot\frac{1}{2}$ and $-\frac{1}{3\cdot 24}=\frac{1}{9}\cdot(-\frac{1}{8})$, 
since the anomaly of 4d Weyl fermion with symmetry $\Spin\times\Z_3$ contains only 3-torsion \cite{Kapustin1406.7329, 2016arXiv160406527F, GarciaEtxebarriaMontero2018ajm1808.00009, 1808.02881, GuoJW1812.11959, 2506.19710}, we can regard $2$ and $8$ as invertible in $\Z_9$.\\
The second equality uses the fact that $1=-8\mod9$ to obtain $\frac{1}{2}=-4\mod9$ and $-\frac{1}{8}=1\mod9$ and uses the fact that $A_{\Z_3}p_1=0\mod3$ \cite{tomonaga1964mod,tomonaga1965pontryagin} to rewrite $A_{\Z_3}p_1=3\cdot \frac{A_{\Z_3}p_1}{3}$.\\
The third equality uses the fact that $3 = 4 \cdot 3 \mod 9$ to rewrite $3\cdot \frac{A_{\Z_3}p_1}{3}=4\cdot 3\cdot \frac{A_{\Z_3}p_1}{3}\mod9$ and factors out the common factor $-4k$ of the two terms.

Thus the 4d fermionic anomaly 
has the anomaly index $k \in \Z_{9}$, agreeing with the bordism group classification by
$\Omega_5^{{\rm Spin} \times \Z_{3}}=\Z_9$
\cite{Kapustin1406.7329, 2016arXiv160406527F, GarciaEtxebarriaMontero2018ajm1808.00009, 1808.02881, GuoJW1812.11959, 2506.19710}.

When we have three right-handed neutrinos (3$\nu_R$),
we need to consider $k=3$ instead of $k=1$, so \eq{eq:Spin-Z3},
with $4k=-12 = -3 \mod 9$,
becomes
\bea \label{eq:alpha5}
&&
 \exp \big( \ii \frac{2 \pi}{3} (-1)  \int_{M^5} ( 
   {A_{\Z_3}} (\beta_{(3,3)} {A_{\Z_3}}) (\beta_{(3,3)} {A_{\Z_3}})
 ) \big),
\eea  
which is the generator of the bosonic group cohomology 5d iTFT 
from ${\rm H}^5(\mathbb{Z}_n,\U(1))\cong\mathbb{Z}_n$
 \cite{Wang1405.7689} with $n=3$.
Later \Sec{Sec:Symmetry-Extension} shows that
for this specific case with $N_f=3$,
the symmetry-extension \eq{eq:Z3-extension} can be used to construct the 3+1d $\Z_{N_c=3}$-gauge topological order with a low-energy 3+1d fermionic $\Z_{N_c=3}$-gauge TQFT 
(as a hypothetical sector of 3+1d dark matter).

\subsection{$\Spin^c \equiv \Spin \times_{\Z_2^\rF} \U(1)$ to
$\Spin \times_{\Z_2^\rF} {\Z_{6}^\rF}$}

\label{sec:Spin^c}

In this section, 
we start with a perturbative local anomaly of U(1) charge $q=1$ left-handed Weyl fermion 
in $\Spin^c \equiv \Spin \times_{\Z_2^\rF} \U(1)$ to
derive the
the nonperturbative global anomaly 
of ${\Z_{6,{\bf B + L}}^\rF}$ charge $q=1$ Weyl fermion 
in 
$\Spin \times_{\Z_2^\rF} {\Z_{6,{\bf B + L}}^\rF}$.

First, we compare the $\Spin^c$ gauge field and the U(1) gauge field.
{For $\Spin^c$, 
 the $\U(1) \supset {\Z_2^\rF}$ contains the fermion parity as a normal subgroup.

For the original $\U(1)$ with $c_1(\U(1))$, 
the gauge bundle constraint is $w_2(TM)= 2 c_1 \mod 2$.
In the original $\U(1)$, fermions have odd charges under $\U(1)$,
while bosons have even charges under $\U(1)$.
Call the original U(1) gauge field $A$,
then $c_1=\frac{\dd A}{2 \pi} \in \frac{1}{2}\Z$.\\

For the new $\U(1)'=\frac{\U(1)}{\Z_2^\rF}$ with $c_1(\U(1)')$,
the gauge bundle constraint is $w_2(TM)= c_1' = 2 c_1 \mod 2$.
Call the new $\U(1)'$ gauge field $A'$,
then $c_1'=\frac{\dd A'}{2 \pi}=\frac{\dd (2A)}{2 \pi} = 2 c_1\in 2\frac{1}{2}\Z = \Z$.\\

To explain why $A' = 2 A$ or $ c_1' = 2 c_1$, we look at the Wilson line operator
$\text{$\exp(\ii q' \oint A')$ and $\exp(\ii q \oint A)$.}$
The original $\U(1)$ has charge transformation $\exp(\ii q \theta)$ with $\theta \in [0, 2 \pi)$,
while the new $\U(1)'$ has charge transformation $\exp(\ii q' \theta')$ with $\theta' \in [0, 2 \pi)$.
But the $\U(1)'=\frac{\U(1)}{\Z_2^\rF}$, so the $\theta=\pi$ in the old $\U(1)$ 
is identified as $\theta'=2\pi$ as a trivial zero
in the new $\U(1)'$.
In the original $\U(1)$, the $q \in \Z$ to be compatible with $\theta \in [0, 2 \pi)$.
In the new $\U(1)'$, the original $q$ is still allowed to have $2\Z$ to be compatible with $\theta \in [0, \pi)$;
but the new $q'=\frac{1}{2} q \in \Z$
and the new $\theta'= 2 \theta  \in [0, 2 \pi)$ are scaled accordingly.
Since the new $q'=\frac{1}{2} q \in \Z$, we show the new $A'=2 A$.}

The perturbative local anomaly of charge $q=1$ left-handed Weyl fermion 
of $\Spin^c$ symmetry
in 3+1d or 4d
is captured by a 5d invertible field theory (iTFT) with the anomaly index $k'=1$
\cite{Putrov:2023jqi2302.14862}:
\begin{equation}
   \exp( \ii k'  \int_{M^5} A' \frac{(2c_1)^2}{48}-A' \frac{p_1}{48}).
\end{equation} 
Now we redefine the U(1) gauge field $A'$ as a $\Z_3$ gauge field $ A'_{\Z_3} \in\H^1(\B\Z_3,\Z_3)= \Z_3$
with the following replacement:
\bea
A' &\mapsto& \frac{2 \pi}{3}   A'_{\Z_3}.\cr 
2c_1=c_1' = \frac{\dd A'}{2 \pi} &\mapsto& 
\frac{\dd   A'_{\Z_3}}{3} \equiv \beta_{(3,3) } A'_{\Z_3}. 
\eea
The $\beta_{(n,m)}: \H^*(-,\Z_m) \mapsto 
\H^{*+1}(-,\Z_n) $ 
is the Bockstein 
 homomorphism associated with the extension
 $\Z_n \stackrel{\cdot m}{\to} \Z_{nm} \to \Z_m$. Thus we get the 5d topological invariant of the $\Spin \times \Z_3$ as:
\bea\label{eq:Spin-Z6-Z2}
   &&\exp \big( \ii 2 \pi k'  \int_{M^5} 
   ( \frac{1}{144} 
   {A'_{\Z_3}} (\beta_{(3,3)} {A'_{\Z_3}}) (\beta_{(3,3)} {A'_{\Z_3}})
   -\frac{1}{3 \cdot 48} {A'_{\Z_3}} p_1 ) \big) \cr
   &=&\exp \big( \ii \frac{2 \pi}{9} k'  \int_{M^5} ( \frac{1}{16} 
   {A'_{\Z_3}} (\beta_{(3,3)} {A'_{\Z_3}}) (\beta_{(3,3)} {A'_{\Z_3}})
   -\frac{1}{16} {A'_{\Z_3}} p_1 ) \big)\cr
   &=&\exp \big( \ii \frac{2 \pi}{9} k'  \int_{M^5} ( 
   4{A'_{\Z_3}} (\beta_{(3,3)} {A'_{\Z_3}}) (\beta_{(3,3)} {A'_{\Z_3}})
   -4\cdot 3\cdot {} \frac{A'_{\Z_3}p_1}{3} ) \big)\cr
   &=&\exp \big( \ii \frac{2 \pi}{9} (4k')  \int_{M^5} ( 
   {A'_{\Z_3}} (\beta_{(3,3)} {A'_{\Z_3}}) (\beta_{(3,3)} {A'_{\Z_3}})
   - 3\cdot {} \frac{A'_{\Z_3}p_1}{3} ) \big).
\eea  
Here $\beta_{(3,3)}:\H^1(-,\Z_3)\to \H^2(-,\Z_3)$ is the Bockstein homomorphism.
Here ${} \frac{A'_{\Z_3}p_1}{3}$ is a mod 3 class that involves Pontryagin class because $A'_{\Z_3}p_1=0\mod3$ \cite{tomonaga1964mod,tomonaga1965pontryagin} (see Appendix \ref{sec:proof} for the proof),
while ${A'_{\Z_3}} (\beta_{(3,3)} {A'_{\Z_3}}) (\beta_{(3,3)} {A'_{\Z_3}})$
is a mod 3 class.

In \eqref{eq:Spin-Z6-Z2}, the first equality rewrites the coefficients $\frac{1}{144}=\frac{1}{9}\cdot\frac{1}{16}$ and $-\frac{1}{3\cdot48}=\frac{1}{9}\cdot(-\frac{1}{16})$ since the anomaly of 4d Weyl fermion with symmetry $\Spin\times_{\Z_2^{\rF}}\Z_6$ contains only 3-torsion \cite{Kapustin1406.7329, 2016arXiv160406527F, GarciaEtxebarriaMontero2018ajm1808.00009, 1808.02881, GuoJW1812.11959, 2506.19710}, we can regard $16$ as invertible in $\Z_9$.\\
The second equality uses the fact that $1=64\mod9$ to obtain $\frac{1}{16}=4\mod9$ and uses the fact that $A'_{\Z_3}p_1=0\mod3$ \cite{tomonaga1964mod,tomonaga1965pontryagin} to rewrite $A'_{\Z_3}p_1=3\cdot \frac{A'_{\Z_3}p_1}{3}$.\\
The third equality factors out the common factor $4k'$ of the two terms.

Eq.\eqref{eq:Spin-Z3} and \eqref{eq:Spin-Z6-Z2} give the same 5d topological invariant since 
\bea 
A'_{\Z_3}=-A_{\Z_3} \mod 3, \text{ and }
k' = k.
\eea
Here, $A'_{\Z_3}=2A_{\Z_3}=-A_{\Z_3}\mod 3$, because $c_1'=2c_1$ and $A'=2A$.\\
{Here, $k' = k$, because we derive from the same perturbative anomaly from the same charge $q=1$ Weyl fermion for both $\Spin \times \U(1)$ and
$\Spin^c \equiv \Spin \times_{\Z_2^\rF} \U(1)$ symmetries.}

Thus, the conclusion here in \Sec{sec:Spin^c} follows the same as \Sec{sec:SpinxU(1)}.

\section{Symmetry Extension $1 \to \mathbb{Z}_{N_c=3}\to \mathbb{Z}_{N_cN_f=9}\to \mathbb{Z}_{N_f=3}\to 1$, Anomaly Trivialization, and 3+1d Anomalous $\Z_3$-Gauge Topologically Ordered Dark Sector}

\label{Sec:Symmetry-Extension}

The fermionic anomaly 
in 3+1d or 4d spacetime
with
$\Z_n$ symmetry is classified by
fermionic Spin bordism group $\Omega_5^{{\rm Spin} \times \Z_{n}}$
\cite{Kapustin1406.7329, 2016arXiv160406527F, GarciaEtxebarriaMontero2018ajm1808.00009, 1808.02881, GuoJW1812.11959,
2506.19710}
which is isomorphic to 
bosonic SO bordism $\Omega_5^{\rm SO}(\B\Z_{n})$ 
up to 2-torsion term when $2 \nmid n$,
namely
\bea
&&\Omega_5^{\rm Spin}(\B\Z_n) \cong \tilde{\Omega}_5^{\rm SO}(\B\Z_n), \quad 2 \nmid n \\
&& \Omega_5^{\rm Spin}(\B\Z_{3^r \cdot s}  )
\cong \tilde{\Omega}_5^{\rm SO}(\B\Z_{3^r \cdot s}  )
\cong \Z_{3^{r+1}}\oplus \Z_{3^{r-1}}\oplus 
\Z_s \oplus 
\Z_s, \quad 2\nmid s, \quad 3\nmid s.
\eea
Here $\tilde{\Omega}_5^{\rm SO}(\B G):=\Omega_5^{\rm SO}(\B G)/\Omega_5^{\rm SO}$ is the reduced bordism group, modding out the $\Omega_5^{\rm SO}=\Omega_5^{\rm SO}(pt)$.

In Appendix \ref{sec:trivialized}, we prove that only the group cohomology subclass 
($\H^5(\Z_n,\U(1))\cong \Z_n$) 
anomaly\footnote{For $n=3$ as $\H^5(\Z_3,\U(1))\cong \Z_3$,
that generator is the anomaly of three right-handed neutrinos $3\nu_R$, which has the anomaly index $3 \in 
\Omega_5^{\rm Spin \times \Z_3} \cong \Omega_5^{\rm Spin}(B\Z_3) \cong \Z_9
$.
That is also the anomaly of the SM missing the $3\nu_R$, up to a negative sign $-1$ for the anomaly index.} 
can be canceled by 
anomalous $G$-symmetric $K=\Z_n$-gauge 4d TQFTs, via the 
appropriate symmetry-extension construction \cite{Wang2017locWWW1705.06728} of
\bea
1\to K \to G_{\rm Tot}\to G \to 1.
\eea
as 
\bea \label{eq:Zn-extension}
1\to\Z_n\to\Z_{n^2}\to\Z_n\to 1.
\eea
On the other hand, the beyond-group-cohomology subclass anomaly 
that involves 
$A p_1$ (the first cohomology class $A$
and the first Pontryagin class $p_1$)
allows no such symmetric TQFTs.

More generally and mathematically, 
in this work, for odd $d\ge3$ and any $n\ge2$,
we prove that any group cohomology cocycle 
\bea
\alpha_d \in \H^d(\Z_n,\U(1)) \cong \Z_n
\eea 
is trivialized by the group extension 
as \eq{eq:Zn-extension}'s $1\to\Z_n\to\Z_{n^2}\to\Z_n\to 1$ \cite{Wang2017locWWW1705.06728}.

In Appendix \ref{sec:cochain}, we find an explicit $(d-1)$-cochain $\beta_{d-1}$ 
that splits the $d$-cocycle $\alpha_d$ by that extension for odd $d\ge3$ and any $n\ge2$.
Namely, $\alpha_d=\delta \beta_{d-1}$ holds when pulling back
the quotient $\Z_n$ to the total $\Z_{n^2}$ group,
from the cocycle $\alpha_d$ in $\H^d(\Z_n,\U(1))$ to the coboundary 
\bea
\tilde{\alpha}_d=\delta \tilde{\beta}_{d-1}\in
\H^d(\Z_{n^2},\U(1)).
\eea

\subsection{Explicit construction of 3+1d anomalous 
TQFT by the symmetry extension $1\to\Z_3\to\Z_{9}\to\Z_3\to 1$}

As an application, for $d=5$ and $n=3$,
we prove \Refe{Wang:2025oow2502.21319}'s statement that 
the symmetry-extension via \eq{eq:Z3-extension}'s
$$
1 \to \Z_{N_c=3}\to\Z_{N_cN_f=9}\to\Z_{N_f=3}\to 1
$$
can construct a $G=\Spin \times \Z_{N_f=3}$-symmetric $K=\Z_{N_c=3}$-gauge 4d low-energy TQFT of gapped anomalous topologically ordered dark matter
via canceling the missing $N_f=3$ right-handed neutrinos $\nu_R$'s
$\Z_{6,\bf B + L}^{\rm F}$- or $\Z_{3,\bf B + L}$-gauge-gravitational anomaly in the 4d SM.
This proves a claim in \Refe{Wang:2025oow2502.21319}.\footnote{Ref.~\cite{1801.05416}'s Section 7 constructs low dimensional coupled bulk–boundary TQFTs via symmetry extension; examples include a $(2+1)$d bulk with a $(1+1)$d boundary and a $(3+1)$d bulk with a $(2+1)$d boundary. Our approach is similar to
Ref.~\cite{1801.05416}'s Sec.~7.}

More explicitly, for the 
$\alpha_5$ given in \eq{eq:alpha5},
\bea
\alpha_5=
 \exp \big( \ii \frac{2 \pi}{3}  \int_{M^5} ( 
   {A_{\Z_3}} (\beta_{(3,3)} {A_{\Z_3}}) (\beta_{(3,3)} {A_{\Z_3}})
 ) \big),
\eea  
we have $\beta_4$ with 
$\alpha_5= \delta \beta_4$,
obtained in Appendix \ref{sec:cochain}, which suggests a construction of the 5d iTFT on the bulk 5-manifold $M^5$
and the 4d noninvertible TQFT 
on the 4d boundary $M^4= \prt M^5$ 
with dynamical 
1-cochain gauge field $a_{\Z_3} \in C^1(\B\Z_3,\Z_3)$
and 2-cochain
(dual) gauge field $b_{\Z_3} \in C^2(\B\Z_3,\Z_3)$, such that the full 5d/4d coupled path integral is given by
\bea \label{eq:5d-4d}
&&\exp \big( \ii \frac{2 \pi}{3}  \int_{M^5} ( 
   {A_{\Z_3}} (\beta_{(3,3)} {A_{\Z_3}}) (\beta_{(3,3)} {A_{\Z_3}})
 ) \big) \cdot 
 \cr
&& \cdot \sum_{
 \substack{
 a_{\Z_3} \in C^1(\B\Z_3,\Z_3) \\
 b_{\Z_3} \in C^2(\B\Z_3,\Z_3)
 }}\exp \big( \ii \frac{2 \pi}{3}  \int_{M^4=\prt M^5} ( 
 b_{\Z_3} \dd a_{\Z_3}
 - b_{\Z_3} \beta_{(3,3)} {A_{\Z_3}}
 - a_{\Z_3} {A_{\Z_3}}
 \beta_{(3,3)} {A_{\Z_3}}
  ) \big) \cr
  &=&\exp \big( \ii \frac{2 \pi}{3}  \int_{M^5} ( 
   {A_{\Z_3}} (\beta_{(3,3)} {A_{\Z_3}}) (\beta_{(3,3)} {A_{\Z_3}})
 ) \big) \cdot
 \cr
&& \cdot\sum_{
 \substack{
 a_{\Z_3} \in C^1(\B\Z_3,\Z_3) \\
 b_{\Z_3} \in C^2(\B\Z_3,\Z_3)
 }}\exp \big( \ii \frac{2 \pi}{3}  \int_{M^4=\prt M^5}(  
  a_{\Z_3}(\dd  b_{\Z_3}- {A_{\Z_3}} \beta_{(3,3)} {A_{\Z_3}})-b_{\Z_3}\beta_{(3,3)} {A_{\Z_3}}
  ) \big).
\eea
This 5d/4d coupled path integral analogously matches the discrete cocycle forms or cochain forms (e.g., \cite{Wang1405.7689})
derived in Appendix \ref{sec:cochain}, as the 5-cocycle
$$\alpha_5(g_1,g_2,g_3,g_4,g_5)=\zeta_n^{g_1 [\frac{g_2+g_3}{n}] [\frac{g_4+g_5}{n}]}$$
and the 4-cochain
$$
\tilde{\beta}_4(h_1,h_2,h_3,h_4)=\zeta_n^{g_1k_2[\frac{g_3+g_4}{n}]}
$$
at $n=3$, where
$\zeta_n$ is an $n$-th root of unity
such as $\zeta_n = \exp(\frac{2 \pi \ii}{n})$,
with variables 
$g \in \Z_n$ and $k \in \Z_n$.

{The 5d bulk partition function on a 5d manifold with a 4d boundary is not gauge-invariant,
but the full 5d/4d coupled path integral \eq{eq:5d-4d} is gauge-invariant under the following gauge transformation:}
\bea\label{eq:gauge-transformation}
&& A_{\Z_3} \mapsto A_{\Z_3} + \dd \lambda_{0,\Z_3}, \cr
&& a_{\Z_3} \mapsto a_{\Z_3}  + \dd \mu_{0,\Z_3}, \cr
&& b_{\Z_3}  \mapsto b_{\Z_3} +\lambda_{0,\Z_3}\beta_{(3,3)}A_{\Z_3}  + \dd \mu_{1,\Z_3}, \cr
&& A_{\Z_3} \in\H^1(\B\Z_3,\Z_3)= \Z_3,\cr
&& a_{\Z_3} \in C^1(\B\Z_3,\Z_3),\cr
&&  b_{\Z_3} \in C^2(\B\Z_3,\Z_3),\cr
&&\lambda_{0,\Z_3} \in C^0(\B\Z_3,\Z_3),\cr
&&\mu_{0,\Z_3} \in C^0(\B\Z_3,\Z_3),\cr
&&\mu_{1,\Z_3} \in C^1(\B\Z_3,\Z_3).
\eea
Here, $C^k(\B\Z_3,\Z_3)$ is the group of $\Z_3$-valued $k$-cochains of the classifying space $\B\Z_3$. The cohomology group $\H^k(\B\Z_3,\Z_3)$ is defined as the quotient group $Z^k(\B\Z_3,\Z_3)/B^k(\B\Z_3,\Z_3)$ where $Z^k(\B\Z_3,\Z_3)$ is the group of $\Z_3$-valued $k$-cocycles of the classifying space $\B\Z_3$ and $B^k(\B\Z_3,\Z_3)$ is the group of $\Z_3$-valued $k$-coboundaries of the classifying space $\B\Z_3$.

Below we check that \eqref{eq:5d-4d} is gauge-invariant under \eqref{eq:gauge-transformation}.
Because $ \beta_{(3,3)} \dd \lambda_{0,\Z_3}=0$, $ \beta_{(3,3)} {A_{\Z_3}}$ and $\dd  b_{\Z_3}- {A_{\Z_3}} \beta_{(3,3)} {A_{\Z_3}}$ are gauge-invariant under the gauge transformation \eqref{eq:gauge-transformation}, hence \eqref{eq:5d-4d} transforms under \eqref{eq:gauge-transformation} as
 \bea
&& \exp \big( \ii \frac{2 \pi}{3}  \int_{M^5} ( 
   {A_{\Z_3}} (\beta_{(3,3)} {A_{\Z_3}}) (\beta_{(3,3)} {A_{\Z_3}})
 ) \big) \cdot 
 \cr
&& \cdot\sum_{
 \substack{
 a_{\Z_3} \in C^1(\B\Z_3,\Z_3) \\
 b_{\Z_3} \in C^2(\B\Z_3,\Z_3)
 }}\exp \big( \ii \frac{2 \pi}{3}  \int_{M^4=\prt M^5}(  
 a_{\Z_3}(\dd  b_{\Z_3}- {A_{\Z_3}} \beta_{(3,3)} {A_{\Z_3}})-b_{\Z_3}\beta_{(3,3)} {A_{\Z_3}}
  ) \big)\cr
  &\mapsto&
  \exp \big( \ii \frac{2 \pi}{3}  \int_{M^5} ( 
   ({A_{\Z_3}}+\dd \lambda_{0,\Z_3}) (\beta_{(3,3)} {A_{\Z_3}}) (\beta_{(3,3)} {A_{\Z_3}})
 ) \big) \cdot 
 \cr
&& \cdot\sum_{
 \substack{
 a_{\Z_3} \in C^1(\B\Z_3,\Z_3) \\
 b_{\Z_3} \in C^2(\B\Z_3,\Z_3)
 }}\exp \big( \ii \frac{2 \pi}{3}  \int_{M^4=\prt M^5}(  
 (a_{\Z_3}+\dd\mu_{0,\Z_3})(\dd  b_{\Z_3}- {A_{\Z_3}} \beta_{(3,3)} {A_{\Z_3}})\cr
 &&-(b_{\Z_3}+\lambda_{0,\Z_3}\beta_{(3,3)}A_{\Z_3}  + \dd \mu_{1,\Z_3})\beta_{(3,3)} {A_{\Z_3}}
  ) \big).
 \eea
Since by the Stokes theorem, we have
\bea
\int_{M^4=\prt M^5}(\dd\mu_{0,\Z_3})(\dd  b_{\Z_3}- {A_{\Z_3}} \beta_{(3,3)} {A_{\Z_3}})=0,
\eea
\bea
\int_{M^4=\prt M^5}(\dd \mu_{1,\Z_3})\beta_{(3,3)} {A_{\Z_3}}=0,
\eea
and 
\bea
\int_{M^4=\prt M^5}\lambda_{0,\Z_3}\beta_{(3,3)}A_{\Z_3}\beta_{(3,3)}A_{\Z_3}=\int_{M^5}(\dd\lambda_{0,\Z_3})\beta_{(3,3)}A_{\Z_3}\beta_{(3,3)}A_{\Z_3},
\eea
\eqref{eq:5d-4d} is gauge-invariant under \eqref{eq:gauge-transformation}.
 
 In Appendix \ref{sec:dd-d-1d}, we construct the $d$d-bulk/$(d-1)$d-boundary coupled invertible topological field theory/symmetric anomalous gapped TQFT
by the symmetry extension $1\to\Z_n\to\Z_{n^2}\to\Z_n\to 1$ for any odd $d\ge3$ and any $n\ge2$ explicitly.

Many more 3+1d anomalous fermionic
TQFTs, 
which carry a mixed gauge-gravitational
nonperturbative global anomaly
of the $G =\mathrm{Spin} \times_{\mathbb{Z}_2^{\rm F}} \Z_{2m}$ 
or $G= \mathrm{Spin} \times \mathbb{Z}_n$
symmetry,
can be found in:
Cheng-Wang-Yang's 3+1d anomalous 
fermionic $\Z_4$-gauge theory \cite{Cheng:2024awi2411.05786} (see also the bosonic analogous discussion in \cite{Yang:2023gvi2303.00719}),
the recent work of
D{\'e}coppet-Yu \cite{Decoppet:2025eic2509.10603} and Debray-Ye-Yu 
\cite{Debray:2025kfg2510.24834},
and Wan-Wang \cite{2512.25038}.
General obstructions and 
constraints on the existence of these anomalous symmetric (3+1)d TQFTs are
discussed in Cordova-Ohmori \cite{Cordova1912.13069}.
General properties of the anomalies
of
$G =\mathrm{Spin} \times_{\mathbb{Z}_2^{\rm F}} \mathbb{Z}_{2m}^{\rm F}$ or 
$\mathrm{Spin} \times \mathbb{Z}_n$
are discussed in Hsieh \cite{1808.02881}
and Wan \cite{2506.19710}, and other related nonperturbative global anomalies
are discussed in Brennan-Intriligator
\cite{2312.04756Brennan:2023vsa}. 

These TQFTs can have 
beyond-the-Standard-Model 
(BSM) applications 
\cite{Wang:2025oow2502.21319, JW2006.16996, JW2012.15860}
for canceling 
SM's nonperturbative global anomalies  
\cite{GarciaEtxebarriaMontero2018ajm1808.00009, 1808.02881, DavighiGripaiosLohitsiri2019rcd1910.11277, WW2019fxh1910.14668}.
For future directions, it 
will be interesting to explore how other nonperturbative global anomalies
can constrain other QFT-coupling-to-TQFT
systems, with other potential BSM applications in mind.

\section{Conclusion and Discussions: General $N_f$ family and General $N_c$ color Standard Model:
Topologically Ordered Dark Sector
via symmetry extension $1 \to \Z_{N_c}\to\Z_{N_cN_f}\to\Z_{N_f}\to 1$}
\label{sec:NcNfSM}

Consider the following $N_f$ family and 
$N_c$ color version of the Standard Model (SM), which
is a 4d chiral gauge theory with Yang-Mills spin-1 gauge fields of
the  Lie algebra 
\bea \label{eq:SMLieAlgebra}
\cG_{\rm SM} \equiv su(N_c) \times  su(2) \times 
u(1)_{Y}
\eea
coupling to $N_f$ families of 15 or 16 Weyl fermions (spin-$\frac{1}{2}$ Weyl spinor 
is in the ${\bf 2}_L^\C$ representation {of} the spacetime symmetry Spin(1,3),
written as a left-handed 15- or 16-plet $\psi_L$)
in the following $\cG_{\rm SM}$ representation
\cite{Bar:2001qk0105258, Shrock:1995bp9512430}
\begin{multline}
    \label{eq:SMrep}
({\psi_L})_{\rm I} =
( \bar{d}_R \oplus {l}_L  \oplus q_L  \oplus \bar{u}_R \oplus   \bar{e}_R  
)_{\rm I}
\oplus
n_{\nu_{{\rm I},R}} {\bar{\nu}_{{\rm I},R}}
\\
\sim  \big((\overline{\bf N}_c,{\bf 1})_{- (1-r) h} \oplus ({\bf 1},{\bf 2})_{-N_c h}   \oplus ({\bf N}_c,{\bf 2})_{h} \oplus (\overline{\bf N}_c,{\bf 1})_{- (1+r) h} \oplus ({\bf 1},{\bf 1})_{2 N_c h} \big)_{\rm I} \oplus n_{\nu_{{\rm I},R}} {({\bf 1},{\bf 1})_{0}}
\\
\sim 
\big((\overline{\bf N}_c,{\bf 1})_{N_c -1} \oplus ({\bf 1},{\bf 2})_{-N_c }  
\oplus
({\bf N}_c,{\bf 2})_{1} \oplus (\overline{\bf N}_c,{\bf 1})_{- (N_c + 1) } \oplus ({\bf 1},{\bf 1})_{2 N_c } \big)_{\rm I}
\oplus n_{\nu_{{\rm I},R}} {({\bf 1},{\bf 1})_{0}}
\end{multline} 
for each family,
while $h$ is an overall normalization
and $r$ is the splitting parameter that can be solved by $u(1)_Y^3$ cubic anomaly cancellation to find $r=\pm N_c$ where
$r=N_c$ gives the correct choice of $u$ and $d$ quark charges.
Here our generic 
$u(1)_{Y}$ hypercharges are solved by the following anomaly-cancellation conditions 
\bea
u(1)_{Y}\text{-}su(N_c)^2 &:&
2 Y_{q_L} + Y_{\bar{u}_R}
+ Y_{\bar{d}_R} =0, \cr
u(1)_{Y}\text{-}su(2)^2
&:& N_c Y_{q_L} +  Y_{l_L}  =0,  \cr    
u(1)_{Y}\text{-}(\text{gravity})^2
&:& 2 N_c Y_{q_L} + N_c Y_{\bar{u}_R}
+ N_c Y_{\bar{d}_R} +  2 Y_{l_L}+
Y_{\bar{e}_R} + Y_{\bar{\nu}_R}=0,\cr
u(1)_{Y}^3 &:& 2 N_c Y_{q_L}^3 + N_c Y_{\bar{u}_R}^3
+ N_c Y_{\bar{d}_R}^3 +  2 Y_{l_L}^3+
Y_{\bar{e}_R}^3 + Y_{\bar{\nu}_R}^3=0,\cr
({\bf B-L})\text{-}u(1)_{Y}^2 &:& 
(2 Y_{q_L}^2 -Y_{\bar{u}_R}^2
-Y_{\bar{d}_R}^2)
-(2 Y_{l_L}^2-Y_{\bar{e}_R}^2)=0,
\eea
and the solution when $Y_{\nu_{{\rm I},R}}=0$ is 
given by
\bea \label{eq:hypercharge}
(Y_{\bar{d}_R}, Y_{{l}_L},  Y_{q_L},  Y_{\bar{u}_R} ,    Y_{\bar{e}_R}, Y_{\nu_{{\rm I},R}} ) 
= 
h \times ( N_c-1 , -N_c, 1, - (N_c+1) , 2 N_c , 0)
\eea
At $h=1$, $N_c=3$, we get
\bea
 (Y_{\bar{d}_R}, Y_{{l}_L},  Y_{q_L},  Y_{\bar{u}_R} ,    Y_{\bar{e}_R}, Y_{\nu_{{\rm I},R}} ) 
= 
( 2 , -3, 1, - 4 , 6 , 0).
\eea
So $N_c=3$ typically goes as
\begin{multline}
({\psi_L})_{\rm I} =
( \bar{d}_R \oplus {l}_L  \oplus q_L  \oplus \bar{u}_R \oplus   \bar{e}_R  
)_{\rm I}
\oplus
n_{\nu_{{\rm I},R}} {\bar{\nu}_{{\rm I},R}}
\sim 
\big((\overline{\bf N}_c,{\bf 1})_{2} \oplus ({\bf 1},{\bf 2})_{-3}  
\oplus
({\bf N}_c,{\bf 2})_{1} \oplus (\overline{\bf N}_c,{\bf 1})_{-4} \oplus ({\bf 1},{\bf 1})_{6} \big)_{\rm I}
\oplus n_{\nu_{{\rm I},R}} {({\bf 1},{\bf 1})_{0}}
\end{multline}

The total number of Weyl fermions in $N_f$ family for the whole multiplet of 
\eq{eq:SMrep} is
\bea
N_f (4 N_c + 3) + \sum_{{\rm I}}  n_{\nu_{{\rm I},R}}.
\eea
For $N_f = N_c = 3$, this total number becomes 
$3 \cdot 15 + \sum_{{\rm I}} n_{\nu_{{\rm I},R}}$.

Now the Witten SU(2) anomaly-free \cite{Witten1982fp} demands that the total number of 
{\bf 2} dimensional representations of SU(2) 
Weyl fermions need to be an even integer:
\bea \label{eq:Witten}
&&N_f(N_c+1) \in 2 \Z
\text{ for Witten SU(2) anomaly-free,}
\cr
&& \text{ so either } 
  \left\{
\begin{array}{l}
N_f \in \Z_{\text{odd}}, N_c \in \Z_{\text{odd}},  {\text{ thus baryon is a fermion}}.\\
N_f \in \Z_{\text{even}}, N_c \in \Z,
 {\text{ thus baryon can be a fermion ($N_c \in \Z_{\text{odd}}$) or a boson ($N_c \in \Z_{\text{even}}$)}}.
\end{array}
    \right.
\eea

Again, we treat the anti-particle of right-handed neutrino $\bar{\nu}_R$ as the left-handed particle, whose quantum numbers (here discrete charges of abelian global symmetries) are given by:
\bea
\begin{tabular}{| c  | c | c |  c |    c | c | c |  c |  c |  }
\hline
  &  $\U(1)_{\bf L}$ & $\U(1)_{{\bf B}-{\bf L}}$ & 
  $\Z_{2 N_f ,{{\bf B} + {\bf L}}}^\rF$ 
  & $\Z_2^\rF$ & 
  $\Z_{2 N_c N_f, {\bf Q} + 3 {\bf L}}^\rF$ 
  \\
\hline
\hline
$\bar{\nu}_R$ &  $-1$ & 1 & $-1$  & $1$ & $-N_c$ \\
\hline
\end{tabular}
\eea
Moreover, only when $N_f$ and 2 are coprime, namely their greatest common divisor is $\gcd(N_f,2)=1$,
such as $N_f=3,5,7,\dots$, then we further have  $\Z_{2 N_f}^\rF=
\Z_{2}^\rF \times \Z_{N_f}$, such that
$\bar{\nu}_R$ has a well-defined $\Z_{N_f ,{{\bf B} + {\bf L}}}$ charge $-1$:
\bea
\begin{tabular}{| c  | c |  }
\hline
  &   $
  \Z_{2N_f ,{{\bf B} + {\bf L}}}^\rF 
=\Z_{2}^\rF \times 
\Z_{N_f ,{{\bf B} + {\bf L}}}$
  \\
\hline
\hline
$\bar{\nu}_R$ &  $-1 \sim 1 \cdot -1$ \\
\hline
\end{tabular}, \quad 
\gcd(N_f,2)=1.
\eea
Furthermore, only when $N_c N_f$ and 2 are coprime, 
namely their greatest common divisor is $\gcd(N_c N_f,2)=1$,
then we further have  $\Z_{2 N_cN_f}^\rF=
\Z_{2}^\rF \times \Z_{N_c N_f}$, such that
$\bar{\nu}_R$ has a well-defined $\Z_{N_cN_f,
{{\bf Q} + N_c{\bf L}}=
N_c({{\bf B} + {\bf L}})
}$ charge $-N_c$:
\bea
\begin{tabular}{| c  | c |  }
\hline
  &   $
  \Z_{2 N_c N_f ,{{\bf Q} + N_c{\bf L}}}^\rF 
=\Z_{2}^\rF \times 
\Z_{N_c N_f ,{{\bf Q} + N_c{\bf L}}}$
  \\
\hline
\hline
$\bar{\nu}_R$ &  
$-N_c \sim 1 \cdot 
-N_c$ \\
\hline
\end{tabular}
, \quad  
\gcd(N_c N_f,2)=1.
\eea
For a generic $N_f$-family SM missing some
$\bar{\nu}_R$, we have to determine its anomaly index in
$\Omega_5^{{\rm Spin} \times_{\Z_2^\rF} \Z_{2N_f ,{{\bf B} + {\bf L}}}^\rF}$, which we require the following bordism group classification
(let $N_f={2^{p}  \cdot  3^r \cdot s}$)
\bea\label{eq:bordism-group}
&&\Omega_5^{{\rm Spin} \times_{\Z_2^\rF} {\Z_{2^{p+1}  \cdot  3^r \cdot s
}}} \cong
\Omega_5^{{\rm Spin} \times_{\Z_2^\rF} {\Z_{2^{p+1} }}} 
\oplus \tilde{\Omega}_5^{\rm SO}(\B\Z_{3^r \cdot s}  )
 \cr
&&=
\Z_{2^{p+3}}\oplus \Z_{2^{p-1}}\oplus
\Z_{3^{r+1}}\oplus \Z_{3^{r-1}}\oplus 
\Z_s \oplus 
\Z_s,\quad p\ge1,\quad r\ge1,\quad 2\nmid s,\quad 3\nmid s.  
\eea
Here $\tilde{\Omega}_5^{\rm SO}(\B G):=\Omega_5^{\rm SO}(\B G)/\Omega_5^{\rm SO}$ is the reduced bordism group, modding out the $\Omega_5^{\rm SO}=\Omega_5^{\rm SO}(pt)$.

Now we ask two general questions relevant for high-energy phenomenology for
$G=\Spin \times_{\Z_2^\rF} \Z_{2 N_f,{{\bf B} + {\bf L}}}$ symmetry:\\
\begin{enumerate}
\item
For 
a single charge $q=-1 \in \Z_{2 N_f,{{\bf B} + {\bf L}}} \subset \Spin \times_{\Z_2^\rF} \Z_{2 N_f,{{\bf B} + {\bf L}}}$-symmetry
$\bar{\nu}_R$, can there exist a symmetric-gapped 4d TQFT matching the
$\bar{\nu}_R$'s symmetry
and anomaly in the full
$\Spin \times_{\Z_2^\rF} \Z_{2 N_f,{{\bf B} + {\bf L}}}$?

The answer to this question is the same as asking in the case of a single charge $q=1 \in \Z_{2 N_f} \subset \Spin \times_{\Z_2^\rF} \Z_{2 N_f}$-symmetry  Weyl fermion, up to a $-1$ sign of the chosen basis. The answer is no, there exists no such symmetric-gapped 4d TQFT
matching $q=1$ or $-1$ Weyl fermion's anomaly in general.

\item
For $N_f$ copies of
charge $q=-1 \in \Z_{2 N_f,{{\bf B} + {\bf L}}} \subset \Spin \times_{\Z_2^\rF} \Z_{2 N_f,{{\bf B} + {\bf L}}}$-symmetry
$\bar{\nu}_R$, can there exist a symmetric-gapped 4d TQFT matching the
$\bar{\nu}_R$'s symmetry
and anomaly in the full
$\Spin \times_{\Z_2^\rF} \Z_{2 N_f,{{\bf B} + {\bf L}}}$?

The answer to this question is the same as asking in the case of $N_f$ copies of charge $q=1 \in \Z_{2 N_f} \subset \Spin \times_{\Z_2^\rF} \Z_{2 N_f}$-symmetry  Weyl fermions, up to a $-1$ sign of the chosen basis.
The answer is in general yes, there exists such symmetric-gapped 4d TQFT
matching $N_f$ copies of $q=1$ or $-1$ Weyl fermion's anomaly in general. 

But there is a refined question:
Is this $N_f$ copies of Weyl fermion
anomaly within a group cohomology (GC) class or 
beyond a group cohomology (BGC) class?

\begin{enumerate}

\item  For group cohomology (GC) class,
\Refe{Wang2017locWWW1705.06728} shows that there always exists a symmetric anomalous gapped boundary with a finite abelian gauge group 
as the low-energy TQFT (here in 4d) to cancel the GC class

\item For beyond a group cohomology (BGC) class,
\Refe{Wang2017locWWW1705.06728} cannot show that a symmetric anomalous gapped boundary with a finite abelian gauge group exists or not. But we are able to determine what are the 
minimal finite group $K$ symmetry extension
that can trivialize the anomaly in $G$ via 
pulling back through
$1\to K \to G_{\rm Tot}\to G \to 1$
to the anomaly-free in $G_{\rm Tot}$.

\end{enumerate}

\end{enumerate}

\subsection{Proof of a Theorem}

Due to Witten's SU(2) anomaly-free constraint in
\eq{eq:Witten},
we summarize the results in two cases,
\eq{eq:Nfodd} and \eq{eq:Nfeven}.
Here we determine the {\bf minimal finite abelian $K$-gauge group extension} for 
the generalized SM with $N_f$ family number (in the column)
and $N_c$ color number (in the row). The symmetry-extension trivialization via $K$ means that we can replace $N_f$ copies of $\bar{\nu}_R$ by $K$-gauge
symmetric-gapped 4d TQFT (namely with a gauge group $K$).
The color index $N_c$ does \emph{not} directly affect the minimal 
$K$-gauge TQFT. But for a certain appropriate $N_c$,
when $K=\Z_{N_c}$, there is an interesting interplay between 
$N_c$ color and $N_f$ family.
\begin{enumerate}
\item The $N_f \in \Z_{\text{odd}}$ and $N_c \in \Z_{\text{odd}}$ case gives rise to the following relation in a table:
{\fontsize{9}{10}\selectfont
\bea \label{eq:Nfodd}
\begin{tabular}{| c  | c | c | c | c | c |  c |  c |  c |  c |}
\hline
$K$-group extension 
 &   $N_f=1$ &  
  \textcolor{white}{$N_f=2$}
 &   $N_f=3$ &
\textcolor{white}{$N_f=4$} 
 &   $N_f=5$ &   
\textcolor{white}{$N_f=6$} 
 &   $N_f=7$
 &   
\textcolor{white}{$N_f=8$} 
 &   $N_f=9$
  \\
\hline
\hline
  &   &    &    & &    &    & & &\\
$N_c=3$ &  No   &    &   \cblue{$\Z_{3=N_c}$} GC & & Trivial & & Trivial  & & {$\Z_{3}$} GC (9 $\bar{\nu}_{R}$) \\ 
  &    &    &  &  & & & &&\\
$N_c=5$ & No & & $\Z_3$ GC  & & Trivial & & Trivial  &&  \\
  &   & & & &&& &&  \\
$N_c=7$ & No & & $\Z_3$ GC &  &
Trivial & & Trivial && \\
  &   & & & &&& &&  \\
$N_c=9$ & No & & $\Z_3$ GC &  &
Trivial & & Trivial && 
{$\Z_{9}$} GC (3 $\bar{\nu}_{R}$)  \\
\hline
\end{tabular}.
\eea
}

\begin{enumerate}

\item
For $N_f=1$, there is no anomaly. We consider the SM missing $N_f=1$ $\bar{\nu}_{R}$, thus ``No'' means ``No anomaly'' (even for a single Weyl fermion) and ``No anomalous TQFT is required.''

\item
For $N_f=3$, there is a $\Omega_5^{{\rm Spin} \times_{\Z_2^\rF} {\Z_{6}^\rF}}=
\Omega_5^{{\rm Spin} \times {\Z_{3}}}=\Z_9$ class anomaly. We consider the SM missing $N_f=3$ $\bar{\nu}_{R}$, where ``$\Z_3$ GC'' means $K=\Z_3$-gauge TQFT
can match the group cohomology $\Z_3$ subclass anomaly. When $N_c=3$, we have
a $K=\Z_{N_c=3}$-gauge TQFT that the $K$-extension matches an  intriguing $\Z_{N_c=3}$-color extension,
because surprisingly ${N_c} = {N_f=3}$ in this case.

\item
For $N_f=5$, $N_f=7$, or other
$N_f$ such that $2 \nmid N_f$
and $3 \nmid N_f$, there is a $\Omega_5^{{\rm Spin} \times_{\Z_2^\rF} {\Z_{2 N_f}}}=
\Omega_5^{{\rm Spin} \times {\Z_{N_f}}}=\Z_{N_f} 
\oplus \Z_{N_f}$ class anomaly.
Although a generic number of
$\bar{\nu}_{R}$ can contribute
an anomaly, 
when we consider the SM missing
$N_f$ $\bar{\nu}_{R}$,
the total anomaly class is trivial, thus we write ``Trivial'' in this case and there is also no need to have any 4d TQFT to cancel the anomaly.

\item
For $N_f=9$, there is a $\Omega_5^{{\rm Spin} \times_{\Z_2^\rF} {\Z_{18}^\rF}}=
\Omega_5^{{\rm Spin} \times {\Z_{9}}}=
\Z_{27} \oplus \Z_3$ class anomaly. 

We consider the SM missing $N_f=9$ $\bar{\nu}_{R}$, where ``$\Z_3$ GC'' means $K=\Z_3$-gauge TQFT
can match the group cohomology $\Z_3$ subclass anomaly. When $N_c=3$, we have
a $K=\Z_{N_c=3}$-gauge TQFT that the $K$-extension matches a  $\Z_{N_c=3}$-color extension,
but ${N_c=3} 
\neq {N_f=9}$ in this case.

Instead if we consider the ($N_f=9$)-SM missing $3$ $\bar{\nu}_{R}$, there ``$\Z_9$ GC'' means $K=\Z_9$-gauge TQFT
can match the group cohomology $\Z_9$ subclass anomaly. When $N_c=9$, we have
a $K=\Z_{N_c=9}$-gauge TQFT that the $K$-extension matches a  $\Z_{N_c=9}$-color extension,
although ${N_c} 
={N_f=9}$ in this case, we need to have 6 extra $\bar{\nu}_{R}$ added into the SM.

Thus we show that 
${N_c} = {N_f=3}$ case is more natural in terms of the $\Z_{N_c}$-color extension. 

For $N_f=9$, or more generally higher 3-power cases 
$$
N_f=3^r, r \geq 2,
$$
for the ($N_f=3^r$)-SM missing $N_f=3^r$ copies of sterile neutrinos $\bar{\nu}_{R}$, 
we again have ``$\Z_3$ GC'' means $K=\Z_3$-gauge TQFT
can match the group cohomology $\Z_3$ subclass anomaly. 
When $N_c=3$, we have
a $K=\Z_{N_c=3}$-gauge TQFT that the $K$-extension matches a $\Z_{N_c=3}$-color extension,
but ${N_c =3} \neq {N_f=3^r}$, with $r \geq 2$ in this case.

Here we need to quote our results derived in Appendices \ref{sec:trivialized}
and \ref{sec:Z3r}, the $k=3^r$ anomaly
of the 4d Weyl fermion for $k=3 \in 
\Z_{3^{r+1}} 
\subset 
 \Omega_5^{\rm Spin}(\B\Z_{3^r }  )
\cong
\Z_{3^{r+1}}\oplus \Z_{3^{r-1}}$
\eq{eq:E7}  with $\Spin\times\Z_{3^r}$ symmetry can be trivialized by a $\Z_3$ extension. 

\end{enumerate}

\item
The $N_f \in \Z_{\text{even}}$ and $N_c \in \Z$
case gives rise to the following $K$-extension:
{\fontsize{10}{12}\selectfont
\bea \label{eq:Nfeven}
\begin{tabular}{| c  | c | c | c | c | c |  c |  c |}
\hline
$K$-group extension  &  
\textcolor{white}{$N_f=1$}
&   $N_f=2$ &   
\textcolor{white}{$N_f=3$}
&   $N_f=4$  & 
\textcolor{white}{$N_f=5$}
&   $N_f=6$ &  
\textcolor{white}{$N_f=7$}
  \\
\hline
\hline
$N_c=2$ &  & 
{$\Z_{4}$} BGC &  & {$\Z_{4}$} BGC & & {$\Z_{12}$} BGC &\\
$N_c=3$ &   & 
{$\Z_{4}$} BGC &    & {$\Z_{4}$} BGC   &   & {$\Z_{12}$} BGC &  \\ 
$N_c=4$ &    &  \cblue{$\Z_{4=N_c}$} BGC   &  &  \cblue{$\Z_{4=N_c}$} BGC & & {$\Z_{12}$} BGC & \\
$N_c=5$ &   & {$\Z_{4}$} BGC & & {$\Z_{4}$} BGC & & {$\Z_{12}$} BGC &  \\
$N_c=6$ &  & {$\Z_{4}$} BGC & & {$\Z_{4}$} BGC & & {$\Z_{12}$} BGC & \\
$N_c=7$ & & {$\Z_{4}$} BGC & & {$\Z_{4}$} BGC & & {$\Z_{12}$} BGC &   \\
$\dots$ &  & $\dots$ & & $\dots$ & & $\dots$ &  \\
$N_c=12$ &   & {$\Z_{4}$} BGC  &    &  {$\Z_{4}$} BGC   &    &  \cblue{$\Z_{12=N_c}$} BGC &\\
\hline
\end{tabular}.
\eea
}

\begin{enumerate}

\item
For $N_f=2$, there is a $\Omega_5^{{\rm Spin} \times_{\Z_2^\rF} {\Z_{4}^\rF}}=
 \Z_{16}$ class anomaly. We consider the SM missing $N_f=2$ $\bar{\nu}_{R}$, there ``$\Z_4$ BGC'' means $K=\Z_4$-gauge TQFT
can match the beyond-the-group-cohomology (BGC) $\Z_8$ subclass anomaly. When $N_c=4$, we have
a $K=\Z_{N_c=4}$-gauge TQFT that the $K$-extension matches a $\Z_{N_c=4}$-color extension, but ${N_c=4} \neq {N_f=2}$ in this case.

\item
For $N_f=4$, there is a $\Omega_5^{{\rm Spin} \times_{\Z_2^\rF} {\Z_{8}^\rF}}=
 \Z_{32} \oplus \Z_2$ class anomaly.
 We consider the SM missing $N_f=4$ $\bar{\nu}_{R}$, there ``$\Z_4$ BGC'' means $K=\Z_4$-gauge TQFT
can match the beyond-the-group-cohomology (BGC) $\Z_8$ subclass anomaly.
When $N_c=4$, we have
a $K=\Z_{N_c=4}$-gauge TQFT that the $K$-extension matches a $\Z_{N_c=4}$-color extension, so ${N_c} = {N_f=4}$ in this case.

\item
For $N_f=6$, there is a $\Omega_5^{{\rm Spin} \times_{\Z_2^\rF} {\Z_{12}^\rF}}=
 \Z_{16} \oplus \Z_9$ class anomaly.
We consider the SM missing $N_f=6$ $\bar{\nu}_{R}$, there ``$\Z_{12}$ BGC'' means $\Z_4$-gauge TQFT
can match the beyond-the-group-cohomology (BGC) $\Z_8 \in \Z_{16} $ subclass anomaly
and an additional $\Z_3$-gauge TQFT
can match the group-cohomology (GC) $\Z_3 \in \Z_{9} $ subclass anomaly. 
When $N_c={12}$, we have
a $K=\Z_{N_c=12}$-gauge TQFT that the $K$-extension matches a $\Z_{N_c=12}$-color extension, but ${N_c=12} \neq {N_f=6}$ in this case.

\item
For $N_f \in \Z_{\text{even}}$, we can go through similar discussions like the above.

\end{enumerate}

\end{enumerate}

To conclude, for the general $N_f$ family and general $N_c$ color Standard Model (SM),
we first determine the minimal cyclic gauge group $K=\Z_N$ that can replace the
$N_f$ missing sterile neutrinos $\bar{\nu}_{R}$ by a 4d anomalous
$\Spin\times_{\Z_2^\rF}\Z_{2 N_f,{{\bf B}+{\bf L}}}^\rF$-symmetric
$\Z_N$-gauge TQFT. The corresponding symmetry extension is
\bea
1 \to \Z_N\to \Spin\times \Z_{NN_f}\to 
\Spin \times_{\Z_2^\rF} \Z_{2 N_f,{{\bf B}+{\bf L}}}^\rF
\to 1.
\eea
The minimal extension order is controlled by the $2$- and $3$-primary parts
of $N_f$:
\bea
  \left\{
\begin{array}{l}
N=1,\; 2\nmid N_f,\; 3\nmid N_f, {\text{ no nontrivial extension is required}}.\\
N=3,\; 2\nmid N_f,\; 3\mid N_f, {\text{ the group-cohomology branch}}.\\
N=4,\; 2\mid N_f,\; 3\nmid N_f, {\text{ the $2$-primary beyond-group-cohomology branch}}.\\
N=12,\; 2\mid N_f,\; 3\mid N_f, {\text{ both the $2$- and $3$-primary branches}}.
\end{array}
    \right.
\eea
The color number $N_c$ enters only after this minimal TQFT gauge group is
compared with the color center. If $N_f$ is odd, then Witten anomaly
cancellation requires $N_c$ to be odd, and the minimal extension above also
has odd $N$. When one further imposes the color-center identification
$N=N_c$, the $\Z_N$-gauge TQFT can coincide with $Z(\SU(N_c))$ and the
baryon-to-quark symmetry extension. We shall now prove the precise theorem,
which is a restatement of Theorem \ref{main-theorem}.

\begin{theorem}\label{main-theorem-rewrite}
Let $N_f=2^p3^r s$ with $p,r\ge0$, $2\nmid s$, and $3\nmid s$.
Assume that the anomaly of $N_f$ copies of the 4d charge $q=1$ Weyl
fermion with symmetry $\Spin\times_{\Z_2^{\rF}}\Z_{2N_f,{{\bf B}+{\bf L}}}^{\rF}$ is
trivialized by a minimal cyclic extension $K=\Z_N$. Then
\bea
N=\left\{
\begin{array}{ll}
1, & p=0,\ r=0,\\
3, & p=0,\ r\ge1,\\
4, & p\ge1,\ r=0,\\
12, & p\ge1,\ r\ge1.
\end{array}\right.
\eea
Equivalently, the minimal nonzero possibilities are
\bea
\left\{
\begin{array}{llll}
N=1,&N_f\ge 1,& 2\nmid N_f,& 3\nmid  N_f,\\
N=3,&N_f\ge 3,& 2\nmid N_f,& 3\mid  N_f,\\
N=4,&N_f\ge 2,& 2\mid N_f,& 3\nmid  N_f,\\
N=12,&N_f\ge 6,& 2\mid N_f,& 3\mid  N_f.
\end{array}\right.
\eea
If $N_f$ is odd, then the symmetry extension requires the minimal $N$ to
be odd, hence the only anomalous extension branch is $N=3$. If, in
addition, the Witten-anomaly/baryon-fermion condition requires $N_c$ to
be odd and the TQFT gauge group is identified with the color center by
$N=N_c$, then the unique minimal solution is
\bea
N=N_c=N_f=3.
\eea
\end{theorem}

\begin{proof}
Under our assumption, the anomaly of $N_f$ copies of the 4d charge
$q=1$ Weyl fermion with symmetry
$\Spin\times_{\Z_2^{\rF}}\Z_{2N_f,{{\bf B}+{\bf L}}}^\rF$ is trivialized by
a minimal cyclic extension $K=\Z_N$. Let
$N_f=2^p\cdot 3^r\cdot s$, where $p\ge0$, $r\ge0$, $2\nmid s$, and
$3\nmid s$. The relevant anomaly index has nontrivial primary components
only at the primes $2$ and $3$:
\bea
\mathcal A(N_f)=
\left\{
\begin{array}{ll}
0, & p=0,\ r=0,\\
3^r\cdot \Z_{3^{r+1}}, & p=0,\ r\ge1,\\
2^p\cdot \Z_{2^{p+3}}, & p\ge1,\ r=0,\\
2^p\cdot \Z_{2^{p+3}}\oplus 3^r\cdot \Z_{3^{r+1}}, & p\ge1,\ r\ge1.
\end{array}\right.
\eea
Here the $2$-primary component is the anomaly inherited from the charge
$q=1$ Weyl fermion with symmetry
$\Spin\times_{\Z_2^{\rF}}\Z_{2^{p+1}}^\rF$, while the $3$-primary component
is the corresponding anomaly with $\Spin\times\Z_{3^r}$ symmetry.

By the results in \cite{2512.25038}, the nontrivial $2$-primary component
is trivialized by a minimal $\Z_4$-extension. In Appendix \ref{sec:Z3r},
we prove that the nontrivial $3$-primary component is trivialized by a
minimal $\Z_3$-extension; the input $A_{\Z_3}p_1=0\mod 3$ removes the
possible beyond-group-cohomology obstruction in this branch. If neither
component is present, no nontrivial extension is required and $N=1$. If
both components are present, the two minimal extensions have coprime orders
$4$ and $3$, so the minimal simultaneous cyclic extension has order
$12$. This gives
\bea
N=\left\{
\begin{array}{ll}
1, & p=0,\ r=0,\\
3, & p=0,\ r\ge1,\\
4, & p\ge1,\ r=0,\\
12, & p\ge1,\ r\ge1.
\end{array}\right.
\eea
Equivalently, this is the four-fold list of minimal nonzero extension
orders stated in the theorem.

It remains to impose the additional color-center interpretation. If
$N_f$ is odd, then $p=0$, so the minimal extension order $N$ is odd. The
Witten anomaly of the generalized SM also requires $N_c$ to be odd when
$N_f$ is odd. If one further identifies the TQFT gauge group with the
color center by setting $N=N_c$, then among the anomalous odd branches the
minimal possibility is $N=N_c=3$, and the minimal compatible family number
is $N_f=3$. Thus the unique minimal color-matched solution is
\bea
N=N_c=N_f=3.
\eea
\end{proof}

In summary, a 4d anomalous
$\Spin \times_{\Z_2^\rF} \Z_{2 N_f,{{\bf B}+{\bf L}}}^\rF$-symmetric
$\Z_N$-gauge TQFT can replace the missing $N_f$ copies of sterile
neutrinos $\bar{\nu}_{R}$, with minimal $N=1,3,4,12$ according to the
prime factors of $N_f$ described above. If we further restrict to the
color-matched generalized SM in which baryons are fermions and the TQFT
gauge group is identified with the color center, $N=N_c$, then the Witten
anomaly and the minimal odd extension branch select the unique minimal
case
\bea
N=N_c=N_f=3.
\eea
In this case the 4d $\Z_N$-gauge fermionic TQFT can be interpreted as the
anomalous topological order replacing the $N_f$ sterile right-handed
neutrinos.

\section{Future Directions}

Here are some potential future research directions:

\begin{enumerate}

\item 
    We remark that our scenario shall be \emph{different} from the Dark Dimension \cite{Montero:2025hye2512.09052} scenario involving:
\begin{itemize}
\item
Higgs mechanism in the ${{\bf B} - {\bf L}}$ gauge field sector --- note that the Higgs mechanism involves the
symmetry-breaking mechanism.
\item 3 right-handed neutrinos propagate in the 5th Dark Dimension.
\end{itemize}
In our case and in parallel work \cite{Wang:2025oow2502.21319, 2512.25038}, we emphasize:
\begin{itemize}
\item
Symmetry-extension construction of anomalous topological order is beyond the Anderson-Higgs mechanism,
different from the symmetry-breaking mechanism.

\item Massive 4+1d Dirac fermion with a relative $\pm 1$ sign flips of the mass  can give rise to 4+1d invertible topological field theory (iTFT)/ Symmetry-Protected Topological states (SPTs)
in the 4+1d bulk. However, the anomalous topological order lives on 3+1d, which
 can be attached to the boundary of 4+1d bulk; or simply attached to the 3+1d SM without the need of the 4+1d bulk at all.

\end{itemize}

\item  Membranes can carry $\Z_3$ statistics associated with the Pontryagin class, which has been demonstrated in 
higher 
spacetime dimensions  \cite{Feng:2025mdg-2509.14314}.
The $\Z_3$ statistics is related to the 
Pontryagin class mod 3 that we explored here. This structure generalizes the bosonic statistics
and fermionic statistics of particle excitations, and also 
generalizes 
the anyon particle/string/membrane statistics, including Aharonov-Bohm or multi-loop braiding statistics
\cite{WangLevin1403.7437, Jiang1404.1062, Wang1404.7854, 1602.05951, Putrov2016qdo1612.09298PWY, Wang2019diz1901.11537}. It will be illuminating to know how 
Pontryagin $\Z_3$ statistics apply to our BSM topological order/TQFT context.

\item The SM contains both
$\U(1)_{{\bf B} - {\bf L}}$
and 
$\Z^\rF_{2 N_f, {\bf B} + {\bf L}}$ within
$$
\Spin \times_{\Z_2^\rF} \U(1)_{{\bf B} - {\bf L}} \times_{\Z_2^\rF}\Z^\rF_{2 N_f, {\bf B} + {\bf L}}$$
according to
\eq{eq:SpinU1Z2NF} 
\cite{KorenProtonStability2204.01741, WangWanYou2204.08393}.
In this article, 
we have mainly focused
on the
$\Z^\rF_{2 N_f, {\bf B} + {\bf L}}$ part of the anomaly cancellation between SM and BSM TQFT/topologically ordered sectors.
What is the situation of 
$\U(1)_{{\bf B} - {\bf L}}$'s
anomaly cancellation scenario between SM and BSM 
sectors? (See the relevant discussions in \cite{Wang:2024auy2501.00607}.)

\begin{itemize}

\item First, although
$\U(1)_{{\bf B} - {\bf L}}$ is preserved in the SM, the $\U(1)_{{\bf B} - {\bf L}}$ shall not be fully dynamically gauged 
(namely, it is not observed as a dynamical continuous gauge symmetry),
because our vacuum lacks that gauged $\U(1)_{{\bf B} - {\bf L}}$ photon propagating modes.

\item Second, the anomaly associated with
a genuine $\U(1)$ (including $\U(1)_{{\bf B} - {\bf L}}$) symmetry for a 4d
QFT is the
{\bf\emph{perturbative 
local anomaly}}, which cannot be canceled by any symmetry-preserving TQFT/ topological order.
This pessimism 
is 
in contrast with the case of  {\bf\emph{nonperturbative 
global anomaly}} cancellation, for which such a symmetry-preserving TQFT/topological order 
is optimistically possible, subject to certain criteria \cite{Cordova1912.13069, Cheng:2024awi2411.05786}.

\item However, the discrete subgroup of $\U(1)_{{\bf B} - {\bf L}}$, 
say 
$\Z^\rF_{2m,{\bf B} - {\bf L}}$, 
can be consistently dynamically gauged
if it is anomaly-free.
The  
$\Z^\rF_{2m,{\bf B} - {\bf L}}$ subgroup of
$\U(1)_{{\bf B} - {\bf L}}$ does not necessarily require additional Higgs for the symmetry-breaking,
but the $\Z^\rF_{2m,{\bf B} - {\bf L}}$
can be (1) realized by the remaining symmetry of some gauge-invariant 
multi-fermion interactions,
$$
\psi \psi \dots \psi \psi
$$
or
(2) itself the genuine UV completion
--- namely, 
there is no 
$\U(1)_{{\bf B} - {\bf L}}$ at UV but it only emerges as this continuous symmetry at IR of the SM energy scale. Hence it is natural to consider
the SM + BSM anomaly cancellation within the symmetry
$$\Spin \times_{\Z_2^\rF} \Z^\rF_{2m,{\bf B} - {\bf L}}\times_{\Z_2^\rF}
\Z^\rF_{2 N_f, {\bf B} + {\bf L}}
$$
such that both
$\Z^\rF_{2m,{\bf B} - {\bf L}}$
and 
$\Z^\rF_{2 N_f, {\bf B} + {\bf L}}$
can be dynamically gauged in the full 
SM + BSM anomaly-free system. This discrete
${{\bf B} \pm {\bf L}}$ anomaly cancellation
topic is previously explored, for example, in 
\cite{JW2006.16996, JW2008.06499, JW2012.15860, WangWanYou2112.14765, WangWanYou2204.08393, Putrov:2023jqi2302.14862, 
Cheng:2024awi2411.05786, Wang:2024auy2501.00607,
Wang:2025oow2502.21319,
2512.25038}. It will be informative to systematically explore how the discrete
${{\bf B} \pm {\bf L}}$ anomaly cancellation combined constraints 
can predict any new 
BSM physics consistent
with phenomenology and experiments.

\end{itemize}

\end{enumerate}

\section{Acknowledgment}

JW thanks Dan Freed, Pavel Putrov, Constantin Teleman, Weicheng Ye, and Matthew Yu for the discussions
on the related topics, and thanks Seth Koren for helpful comments on the manuscript.
JW thanks Roman Rybiansky
for his suggestion on improving 
Fig.~\ref{fig:3-family}.
ZW is supported by the NSFC Grant No. 12405001.
JW is supported by LIMS and Ben Delo Fellowship.
JW would like to thank the Isaac Newton Institute for Mathematical Sciences, Cambridge, for support and hospitality during the programme Diving Deeper into Defects: On the Intersection of Field Theory, Quantum Matter, and Mathematics, where work on this paper was undertaken. This work was supported by EPSRC grant EP/Z000580/1.
JW also thanks Simons Foundation Collaboration on Global Categorical Symmetries Annual Meetings in 2024 and 2025, where this work is discussed and performed during the 
meetings.

\appendix

\section{Any cocycle 
$\alpha_d \in \H^d(\Z_n,\U(1))$ is trivialized by the symmetry extension $1\to\Z_n\to\Z_{n^2}\to\Z_n\to 1$ 
for odd $d\ge3$ and any $n\ge2$}

\label{sec:trivialized}

In this appendix, we show that any cocycle $\alpha\in \H^d(\Z_n,\U(1))$ is trivialized by the symmetry extension 
$1\to\Z_n\to\Z_{n^2}\to\Z_n\to1$ for odd $d\ge3$ and any $n\ge2$.

We consider the Lyndon-Hochschild-Serre (LHS) spectral
sequence
\bea\label{eq:SSS}
E_2^{p,q}=\H^p(\Z_n,\H^q(\Z_n,\U(1)))\Rightarrow \H^{p+q}(\Z_{n^2},\U(1))
\eea
associated with the extension $1\to\Z_n\to\Z_{n^2}\to\Z_n\to1$.

Since 
\bea
\H^d(\Z_n,\U(1))=\left\{\begin{array}{lll}\Z_n&d\text{ odd}\\0&d\text{ even}>0\\\U(1)&d=0\end{array}\right.
\eea
and
\bea
\H^d(\Z_n,\Z_n)=\Z_n\;\;\;\forall d\ge0,
\eea
the $E_2$ page of the LHS spectral
sequence \eqref{eq:SSS} is shown in \Fig{fig:SSS}.

\begin{figure}[!h]
\centering
\begin{sseq}[grid=none,labelstep=1,entrysize=1.5cm]{0...7}{0...6}
\ssdrop{\U(1)}
\ssmoveto 1 0 
\ssdrop{\Z_n}
\ssmoveto 2 0
\ssdrop{0}
\ssmoveto 3 0 
\ssdrop{\Z_n}
\ssmoveto 4 0
\ssdrop{0}
\ssmoveto 5 0 
\ssdrop{\Z_n}
\ssmoveto 6 0
\ssdrop{0}
\ssmoveto 7 0 
\ssdrop{\Z_n}

\ssmoveto 0 1
\ssdrop{\Z_n}
\ssmoveto 1 1
\ssdrop{\Z_n}
\ssmoveto 2 1
\ssdrop{\Z_n}
\ssmoveto 3 1
\ssdrop{\Z_n}
\ssmoveto 4 1
\ssdrop{\Z_n}
\ssmoveto 5 1
\ssdrop{\Z_n}
\ssmoveto 6 1
\ssdrop{\Z_n}
\ssmoveto 7 1
\ssdrop{\Z_n}

\ssmoveto 0 2
\ssdrop{0}
\ssmoveto 1 2
\ssdrop{0}
\ssmoveto 2 2
\ssdrop{0}
\ssmoveto 3 2
\ssdrop{0}
\ssmoveto 4 2
\ssdrop{0}
\ssmoveto 5 2
\ssdrop{0}
\ssmoveto 6 2
\ssdrop{0}
\ssmoveto 7 2
\ssdrop{0}

\ssmoveto 0 3
\ssdrop{\Z_n}
\ssmoveto 1 3
\ssdrop{\Z_n}
\ssmoveto 2 3
\ssdrop{\Z_n}
\ssmoveto 3 3
\ssdrop{\Z_n}
\ssmoveto 4 3
\ssdrop{\Z_n}
\ssmoveto 5 3
\ssdrop{\Z_n}
\ssmoveto 6 3
\ssdrop{\Z_n}
\ssmoveto 7 3
\ssdrop{\Z_n}

\ssmoveto 0 4
\ssdrop{0}
\ssmoveto 1 4
\ssdrop{0}
\ssmoveto 2 4
\ssdrop{0}
\ssmoveto 3 4
\ssdrop{0}
\ssmoveto 4 4
\ssdrop{0}
\ssmoveto 5 4
\ssdrop{0}
\ssmoveto 6 4
\ssdrop{0}
\ssmoveto 7 4
\ssdrop{0}

\ssmoveto 0 5
\ssdrop{\Z_n}
\ssmoveto 1 5
\ssdrop{\Z_n}
\ssmoveto 2 5
\ssdrop{\Z_n}
\ssmoveto 3 5
\ssdrop{\Z_n}
\ssmoveto 4 5
\ssdrop{\Z_n}
\ssmoveto 5 5
\ssdrop{\Z_n}
\ssmoveto 6 5
\ssdrop{\Z_n}
\ssmoveto 7 5
\ssdrop{\Z_n}

\ssmoveto 0 6
\ssdrop{0}
\ssmoveto 1 6
\ssdrop{0}
\ssmoveto 2 6
\ssdrop{0}
\ssmoveto 3 6
\ssdrop{0}
\ssmoveto 4 6
\ssdrop{0}
\ssmoveto 5 6
\ssdrop{0}
\ssmoveto 6 6
\ssdrop{0}
\ssmoveto 7 6
\ssdrop{0}

\ssmoveto 1 1
\ssarrow[color=blue] 2 {-1}

\ssmoveto 3 1
\ssarrow[color=blue] 2 {-1}

\ssmoveto 5 1
\ssarrow[color=blue] 2 {-1}

\ssmoveto 1 3
\ssarrow[color=blue] 3 {-2}

\end{sseq}

\caption{The $E_2$ page of the LHS spectral sequence \eqref{eq:SSS}. The differentials will be explained later.}
\label{fig:SSS}
\end{figure}

The differentials in the LHS spectral sequence \eqref{eq:SSS} are 
\bea
d_r^{p,q}:E_r^{p,q}\to E_r^{p+r,q-r+1}\text{ for }r\ge2,
\eea
and the pages $E_r$ are defined inductively from $E_2$ by
\bea
E_{r+1}^{p,q}=\frac{\text{Ker }d_r^{p,q}}{\text{Im }d_r^{p-r,q+r-1}}.
\eea
The differentials $d_r$ vanish and the pages $E_r$ stabilize for sufficiently large $r\ge N$. The page $E_N$ is denoted $E_{\infty}$.

The homomorphism $\H^d(\Z_n,\U(1))\to \H^d(\Z_{n^2},\U(1))$ induced from the extension $1\to\Z_n\to\Z_{n^2}\to\Z_n\to1$ is the composition
\bea\label{eq:composition}
E_2^{d,0}=\H^d(\Z_n,\U(1))\twoheadrightarrow E_{\infty}^{d,0}\hookrightarrow \H^d(\Z_{n^2},\U(1)).
\eea

For $p+q=d$, there is a filtration 
\bea
F^{-1}=0\subset F^0\subset F^1\subset\cdots\subset F^d=\H^d(\Z_{n^2},\U(1))
\eea
of $\H^d(\Z_{n^2},\U(1))$
with 
\bea
F^q/F^{q-1}=E_{\infty}^{d-q,q}.
\eea

We will show that $E_{\infty}^{d,0}=0$ for odd $d\ge3$ and any $n\ge2$, hence by \eqref{eq:composition}, any cocycle $\alpha\in \H^d(\Z_n,\U(1))$ is trivialized by the extension $1\to\Z_n\to\Z_{n^2}\to\Z_n\to1$ for odd $d\ge3$ and any $n\ge2$.

Since $\H^2(\Z_{n^2},\U(1))=0$, the group $E_2^{1,1}=\Z_n$ is eliminated by the differential
\[
d_2^{1,1}:E_2^{1,1}\longrightarrow E_2^{3,0},
\]
so it does not survive to the $E_3$--page. Hence $E_{\infty}^{3,0}=E_3^{3,0}=0$.

Since $\H^4(\Z_{n^2},\U(1))=0$, the groups $E_2^{1,3}=\Z_n$ and $E_2^{3,1}=\Z_n$ are removed on some page by differentials. On the other hand, because $\H^3(\Z_{n^2},\U(1))=\Z_{n^2}$ and $E_2^{0,3}=E_2^{2,1}=\Z_n$, and since $E_2^{3,0}=\Z_n$ is eliminated by the differential $d_2^{1,1}:E_2^{1,1}\to E_2^{3,0}$, the groups $E_2^{0,3}=\Z_n$ and $E_2^{2,1}=\Z_n$ survive to the $E_{\infty}$--page.

Therefore $E_2^{3,1}=\Z_n$ cannot be the target of any nonzero differential; instead it supports the differential
\[
d_2^{3,1}:E_2^{3,1}\longrightarrow E_2^{5,0},
\]
and so does not persist to $E_{\infty}$. Hence $E_{\infty}^{5,0}=E_3^{5,0}=0$.

Similarly, $E_2^{1,3}=E_3^{1,3}=\Z_n$ cannot be the target of any nonzero differential; it is the source of the differential
\[
d_3:E_3^{1,3}\longrightarrow E_3^{4,1},
\]
and therefore $E_{\infty}^{4,1}=E_4^{4,1}=0$.

Since $\H^6(\Z_{n^2},\U(1))=0$, the group $E_2^{5,1}=\Z_n$ is eliminated by some differential. On the other hand, because $\H^5(\Z_{n^2},\U(1))=\Z_{n^2}$ and $E_2^{0,5}=E_2^{2,3}=\Z_n$, while $E_2^{4,1}=\Z_n$ and $E_2^{5,0}$ are removed by the differentials $d_3$ and $d_2$ respectively, the groups $E_2^{0,5}=\Z_n$ and $E_2^{2,3}=\Z_n$ survive to $E_{\infty}$. Hence $E_2^{5,1}=\Z_n$ cannot be the target of a nonzero differential; instead it supports
\[
d_2^{5,1}:E_2^{5,1}\longrightarrow E_2^{7,0},
\]
and therefore $E_{\infty}^{7,0}=E_3^{7,0}=0$.

In fact, one can show in general that $E_2^{d,0}$ is eliminated by the differential
\[
d_2^{d-2,1}:E_2^{d-2,1}\longrightarrow E_2^{d,0},
\]
so $E_{\infty}^{d,0}=E_3^{d,0}=0$ for every odd $d\ge3$ and any $n\ge2$.

The cohomology ring $\H^*(\Z_n,\Z_n)$ is generated by $x\in \H^1(\Z_n,\Z_n)$ and $y\in \H^2(\Z_n,\Z_n)$ and they satisfy the relation $x^2=0$ for odd $n$ and $x^2=\frac{n}{2}y$ for even $n$ \cite[Example 3.41]{Hatcher2002}.
In particular, for $n=2$, $x^2=y$ and $\H^*(\Z_2,\Z_2)$ is generated by $x\in \H^1(\Z_2,\Z_2)$.
The generator of $E_2^{1,1}=\Z_n$ is $x$ and the generator of $E_2^{2,1}=\Z_n$ is $y$.
Since we have shown that $d_2(x)$ is non-trivial and $d_2(y)=0$, the differentials are derivations, and $\smile y: \H^m(\Z_n,\U(1))\to \H^{m+2}(\Z_n,\U(1))$ is an isomorphism for odd $m$\footnote{This is because $\H^m(\Z_n,\U(1))=\H^{m+1}(\Z_n,\Z)$, 
while $\H^*(\Z_n,\Z)$ is periodic of period 2, and $\H^2(\Z_n,\Z)=\H^2(\Z_n,\Z_n)$ by the universal coefficient theorem.}, for $n=2$, 
$d_2(x^k)=kd_2(x)x^{k-1}$ is non-trivial for odd $k$,
and for $n>2$,
$d_2(xy^k)=d_2(x)y^k$ 
is non-trivial for all $k\ge0$. Hence $E_{\infty}^{d,0}=E_3^{d,0}=0$ for odd $d\ge3$ and any $n\ge2$.

\section{$A_{\Z_n}p_1$ cannot be trivialized by any finite group extension except for $n=2$ and $n=3$}\label{sec:generic}
In this appendix, we show that 
$A_{\Z_n}p_1$ cannot be trivialized by any finite group extension except for $n=2$ and $n=3$ where $A_{\Z_n}$ is the generator of $\H^1(\Z_n,\U(1))$, here
$A_{\Z_n}p_1$ is defined on any dimensional manifold $M$
with $\Spin\times_{\Z_2^{\rF}}\Z_{2n}$ structure and $p_1=p_1(TM)$.

For any group extension
\bea
1\to K\to H\to G=\Z_n\to 1,
\eea
we have a similar LHS spectral sequence
\bea
E_2^{p,q}=\H^p(\Z_n,\H^q(K,\U(1)))\Rightarrow\H^{p+q}(H,\U(1)).
\eea
For degree reasons, there are no differentials from or to $E_2^{1,0}=\H^1(\Z_n,\U(1))$, so $E_2^{1,0}=\H^1(\Z_n,\U(1))$ survives to the $E_{\infty}$ page. Therefore, $A_{\Z_n}\in\H^1(\Z_n,\U(1))$ cannot be trivialized by any group extension. 

\subsection{$\Spin\times_{\Z_2^{\rF}}\Z_{2n}$ structure and symmetry extension trivialization}

Now we consider the availability of the symmetry extension trivialization
of $A_{\Z_n}p_1$
on the manifolds with $\Spin\times_{\Z_2^{\rF}}\Z_{2n}$ structure. Let us explain the exceptional cases $n=2$ and $n=3$:
\begin{itemize}
    \item 

For $n=3$, we have $A_{\Z_3}p_1=0\mod 3$ \cite{tomonaga1964mod,tomonaga1965pontryagin} (see Appendix \ref{sec:proof} for the proof) where $p_1$ is the first Pontryagin class. For other $n$, $A_{\Z_n}p_1\not\equiv0\mod n$.

\item
For $n=2$ and $\Spin\times_{\Z_2^{\rF}}\Z_4$ structure, we consider the symmetry extension
\bea\label{eq:extension-Spin-Z4}
1\to\Z_2^{\rF}\to\Spin\times\Z_4\xrightarrow{f}\Spin\times_{\Z_2^{\rF}}\Z_4\to1.
\eea
Since $\H^4(\B\Spin,\Z)$ is generated by $\lambda=\frac{p_1}{2}$, $f^*(p_1)=2\lambda$. Also, $f^*(A_{\Z_2})=2A_{\Z_4}$. Hence $f^*(A_{\Z_2}p_1)=2A_{\Z_4}\cdot2\lambda=0\mod4$.
Hence, $A_{\Z_2}p_1$ can be trivialized by the finite group extension \eqref{eq:extension-Spin-Z4}.

\end{itemize}

Next, for $\Spin\times_{\Z_2^{\rF}}\Z_{2n}$ structure, let us explain the remaining cases.
\begin{itemize}
    \item 

For odd $n>3$, we note that $\Spin\times_{\Z_2^{\rF}}\Z_{2n}=\Spin\times\Z_n$.
Because $\pi_1(\Spin)=0$, there are no finite covers of $\Spin$.
Since $\H^4(\B\Spin,\Z)$ is generated by $\lambda=\frac{p_1}{2}$, $A_{\Z_n}p_1=A_{\Z_n}\cdot2\lambda$. For odd $n>3$,
because $2 \neq 0 \mod n$,
so
$A_{\Z_n}p_1=A_{\Z_n}\cdot2\lambda$ cannot be trivialized by any finite group extension.

\item 
For even $n>2$ and $\Spin\times_{\Z_2^{\rF}}\Z_{2n}$ structure, 
we consider the symmetry extension
\bea\label{eq:extension-Spin-Z2n}
1\to\Z_2^{\rF}\to\Spin\times\Z_{2n}\xrightarrow{g}\Spin\times_{\Z_2^{\rF}}\Z_{2n}\to1.
\eea
Since $\H^4(\B\Spin,\Z)$ is generated by $\lambda=\frac{p_1}{2}$, $g^*(p_1)=2\lambda$. Also, $g^*(A_{\Z_n})=2A_{\Z_{2n}}$. Hence $g^*(A_{\Z_n}p_1)=2A_{\Z_{2n}}\cdot 2\lambda$. For even $n>2$, because $2\cdot2 \neq 0 \mod 2n$, so $A_{\Z_n}p_1$ cannot be trivialized by the finite group extension \eqref{eq:extension-Spin-Z2n}. Therefore, $A_{\Z_n}p_1$ cannot be trivialized by any finite group extension.

\end{itemize}


Therefore, except for $n=2$ and $n=3$, $A_{\Z_n}p_1$ cannot be trivialized by any finite group extension.\footnote{However, $p_1$ or $\frac{p_1}{2}$ 
is trivialized by a higher group extension
\bea
1\to \B^2\Z\to\String\to\Spin\to1.
\eea
Here $p_1$ is integer
$\Z$ valued for SO manifolds,
while
$\frac{p_1}{2}$ is integer $\Z$ valued 
for Spin manifolds because
$p_1=w_2^2 =0 \mod 2$.
} 

\subsection{$\SO\times \Z_{n}$ structure
 and symmetry extension trivialization}

Now we consider the availability of the symmetry extension trivialization
of $A_{\Z_n}p_1$
on the manifolds with $\SO\times \Z_{n}$ structure, instead of the $\Spin\times_{\Z_2^{\rF}}\Z_{2n}$ structure.

Let us explain the two exceptional cases, $n=3$ and $n=2$:
\begin{itemize}
\item
For $n=3$, on manifolds with the $\SO \times \Z_3$ structure,
we have $A_{\Z_3}p_1=0\mod 3$ \cite{tomonaga1964mod,tomonaga1965pontryagin} (see Appendix \ref{sec:proof} for the proof) where $p_1$ is the first Pontryagin class. For other $n$, $A_{\Z_n}p_1\not\equiv0\mod n$.
Hence, except for $n=3$, $A_{\Z_n}p_1$ can only be trivialized by trivializing $p_1\mod n$.
\item
For $n=2$, on manifolds with the $\SO \times \Z_2$ structure, since $p_1=w_2^2\mod2$ where $w_2$ is the second Stiefel-Whitney class, $p_1\mod2$ can be trivialized by trivializing $w_2$ via the symmetry extension, because
$w_2 =0$ on the pullback spin manifold,
\bea
1\to\Z_2^{\rF}\to\Spin\xrightarrow{f}\SO\to1.
\eea
So $A_{\Z_2}p_1= 
A_{\Z_2} w_2^2\mod2$ can be trivialized by a fermion parity $\Z_2^{\rF}$ group extension on the pullback manifold with $\Spin \times \Z_2$ structure.

\item For 
odd $n$ and even $n>2$:
Because $\pi_1(\SO)=\Z_2$, the only connected finite covers of $\SO$ are $\Spin$ and $\SO$.
Since $\H^4(\B\Spin,\Z)$ is generated by $\lambda=\frac{p_1}{2}$, $f^*(p_1\mod n)=2\lambda\mod n$. For odd $n$ and even $n>2$,
because $2 \neq 0 \mod n$,
so
$p_1\mod n$ cannot be trivialized by any finite group extension. 
\end{itemize}

Therefore, except for $n=2$ and $n=3$, $A_{\Z_n}p_1$ cannot be trivialized by any finite group extension.

\section{Explicit $(d-1)$-cochain $\tilde{\beta}_{d-1}$
that splits the $d$-cocycle
$\tilde{\alpha}_d=\delta\tilde{\beta}_{d-1}$
as a coboundary
in $\H^d(\Z_{n^2},\U(1))$
 by the symmetry extension $1\to\Z_n\to\Z_{n^2}\to\Z_n\to1$ for 
 any odd $d\ge3$ and any $n\ge2$}

\label{sec:cochain}
 
In this appendix, we find an explicit $(d-1)$-cochain $\tilde{\beta}_{d-1}$ that splits the $d$-cocycle 
$\tilde{\alpha}_d=\delta\tilde{\beta}_{d-1}$ as a coboundary
in $\H^d(\Z_{n^2},\U(1))$
by the symmetry extension $1\to\Z_n\to\Z_{n^2}\to\Z_n\to1$ for any odd $d\ge3$ and any $n\ge2$.
Our strategy follows \Refe{Wang2017locWWW1705.06728}'s symmetry-extension approach to trivialize a cocycle in terms of coboundary, especially in \Refe{Wang2017locWWW1705.06728}'s Appendices.

\subsection{$d= 3$ and any $n\ge2$: 
Find $\tilde{\beta}_2$
such that $\tilde{\alpha}_3=\delta\tilde{\beta}_2$}

Explicitly, any 3-cocycle $\alpha_3\in \H^3(\Z_n,\U(1))$ has the form \cite{Wang1405.7689, 1703.03266}
\bea
\alpha_3(g_1,g_2,g_3)=\zeta_n^{g_1 [\frac{g_2+g_3}{n}]}
\eea
where $g_i\in\Z_n$ for $i=1,2,3$, $\zeta_n$ is an $n$-th root of unity
(for example, $\zeta_n = \exp(\frac{2 \pi \ii}{n})$), and $[\frac{p}{q}]$ denotes the integer part of $\frac{p}{q}$.

We can find an explicit 2-cochain $\tilde{\beta}_2\in C^2(\Z_{n^2},\U(1))$ such that $\tilde{\alpha}_3=\delta\tilde{\beta}_2$ where $\tilde{\alpha}_3=f^*\alpha_3\in \H^3(\Z_{n^2},\U(1))$ and $f$ is the map in the extension $1\to\Z_n\to\Z_{n^2}\xrightarrow{f}\Z_n\to1$.
Namely,
\bea
\tilde{\alpha}_3(h_1,h_2,h_3)=\delta\tilde{\beta}_2(h_1,h_2,h_3)=\frac{\tilde{\beta}_2(h_2,h_3)\tilde{\beta}_2(h_1,h_2h_3)}{\tilde{\beta}_2(h_1h_2,h_3)\tilde{\beta}_2(h_1,h_2)}
\eea
where $h_i=(g_i,k_i)\in\Z_{n^2}$ for $i=1,2,3$ and 
\bea\label{eq:composition-law}
(g_1,k_1)\cdot(g_2,k_2)=(g_1+g_2,k_1+k_2+[\frac{g_1+g_2}{n}]).
\eea
Here, $\phi(g_1,g_2)=[\frac{g_1+g_2}{n}]$ is the 2-cocycle in $\H^2(\Z_n,\Z_n)$ classifying the extension $1\to\Z_n\to\Z_{n^2}\to\Z_n\to1$.
Note that for $n=2$, $[\frac{g_1+g_2}{2}]=g_1g_2\mod 2$, so the composition law \eqref{eq:composition-law} agrees with that in \cite{Wang2017locWWW1705.06728} for $n=2$.

Explicitly, the 2-cochain $\tilde{\beta}_2\in C^2(\Z_{n^2},\U(1))$ is given by 
\cite{Wang2017locWWW1705.06728}
\bea\label{eq:2-cochain}
\tilde{\beta}_2(h_1,h_2)=\zeta_n^{g_1k_2}.
\eea
Then
\bea
\delta\tilde{\beta}_2(h_1,h_2,h_3)=\zeta_n^{g_2k_3+g_1(k_2+k_3+[\frac{g_2+g_3}{n}])-(g_1+g_2)k_3-g_1k_2}=\zeta_n^{g_1[\frac{g_2+g_3}{n}]}=\alpha_3(g_1,g_2,g_3)=\tilde{\alpha}_3(h_1,h_2,h_3).
\eea
Therefore, $\tilde{\alpha}_3=\delta\tilde{\beta}_2$.

{This 2-cochain $\tilde{\beta}_2$ can also be constructed from the LHS spectral
sequence method, see Appendix \ref{sec:trivialized} and \cite{Wang2017locWWW1705.06728}. In Appendix \ref{sec:trivialized}, we prove that $\H^3(\Z_n,\U(1))=\Z_n$ is eliminated by the differential $d_2^{1,1}:E_2^{1,1}\to E_2^{3,0}$ where 
\bea
E_2^{1,1}=\H^1s\Z_n,\H^1(\Z_n,\U(1)))=\H^1(\Z_n,\Z_n)=\Z_n
\eea
and 
\bea
E_2^{3,0}=\H^3(\Z_n,\H^0(\Z_n,\U(1)))=\H^3(\Z_n,\U(1))=\Z_n.
\eea
Since any 1-cocycle $\alpha_1\in\H^1(\Z_n,\U(1))$ has the form \cite{Wang1405.7689, 1703.03266} 
\bea
\alpha_1(g)=\zeta_n^{g}
\eea
where $g\in\Z_n$, the
2-cochain \eqref{eq:2-cochain} can be expressed as 
\bea
\tilde{\beta}_2(h_1,h_2)=(\alpha_1(k_2))^{g_1}.
\eea
}

\subsection{$d= 5$ and any $n\ge2$: 
Find $\tilde{\beta}_4$
such that $\tilde{\alpha}_5=\delta\tilde{\beta}_4$}

Explicitly, any 5-cocycle $\alpha_5\in \H^5(\Z_n,\U(1))$ has the form \cite{Wang1405.7689, 1703.03266}
\bea
\alpha_5(g_1,g_2,g_3,g_4,g_5)=\zeta_n^{g_1 [\frac{g_2+g_3}{n}] [\frac{g_4+g_5}{n}]}
\eea
where $g_i\in\Z_n$ for $i=1,2,3,4,5$, $\zeta_n$ is an $n$-th root of unity, and $[\frac{p}{q}]$ denotes the integer part of $\frac{p}{q}$.

We can find an explicit 4-cochain $\tilde{\beta}_4\in C^4(\Z_{n^2},\U(1))$ such that $\tilde{\alpha}_5=\delta\tilde{\beta}_4$ where $\tilde{\alpha}_5=f^*\alpha_5\in \H^5(\Z_{n^2},\U(1))$ and $f$ is the map in the extension $1\to\Z_n\to\Z_{n^2}\xrightarrow{f}\Z_n\to1$.
Namely,
\bea
\tilde{\alpha}_5(h_1,h_2,h_3,h_4,h_5)=\delta\tilde{\beta}_4(h_1,h_2,h_3,h_4,h_5)=\frac{\tilde{\beta}_4(h_2,h_3,h_4,h_5)\tilde{\beta}_4(h_1,h_2h_3,h_4,h_5)\tilde{\beta}_4(h_1,h_2,h_3,h_4h_5)}{\tilde{\beta}_4(h_1h_2,h_3,h_4,h_5)\tilde{\beta}_4(h_1,h_2,h_3h_4,h_5)\tilde{\beta}_4(h_1,h_2,h_3,h_4)}
\eea
where $h_i=(g_i,k_i)\in\Z_{n^2}$ for $i=1,2,3,4,5$ and 
\bea
(g_1,k_1)\cdot(g_2,k_2)=(g_1+g_2,k_1+k_2+[\frac{g_1+g_2}{n}]).
\eea
Here, 
\bea
\phi(g_1,g_2)=[\frac{g_1+g_2}{n}]
\eea 
is the 2-cocycle in $\H^2(\Z_n,\Z_n)$ classifying the extension $1\to\Z_n\to\Z_{n^2}\to\Z_n\to1$.
Note that for $n=2$, $[\frac{g_1+g_2}{2}]=g_1g_2\mod 2$.

Explicitly, the 4-cochain $\tilde{\beta}_4\in C^4(\Z_{n^2},\U(1))$ is given by
\bea
\tilde{\beta}_4(h_1,h_2,h_3,h_4)=\zeta_n^{g_1k_2[\frac{g_3+g_4}{n}]}.
\eea
In fact, if we write $\tilde{\beta}_4=\zeta_n^{\tilde{\gamma}_4}$ and $\tilde{\beta}_2=\zeta_n^{\tilde{\gamma}_2}$, then $\tilde{\gamma}_4=\tilde{\gamma}_2\smile \phi$.
Then
\bea
\delta\tilde{\beta}_4(h_1,h_2,h_3,h_4,h_5)=\zeta_n^X.
\eea
We check that 
\bea
X&=&g_2k_3[\frac{g_4+g_5}{n}]+g_1(k_2+k_3+[\frac{g_2+g_3}{n}])[\frac{g_4+g_5}{n}]+g_1k_2[\frac{g_3+(g_4+g_5)\mod n}{n}]\cr
&&-(g_1+g_2)k_3[\frac{g_4+g_5}{n}]-g_1k_2[\frac{(g_3+g_4)\mod n+g_5}{n}]-g_1k_2[\frac{g_3+g_4}{n}]\cr
&=&g_1[\frac{g_2+g_3}{n}][\frac{g_4+g_5}{n}].
\eea
Here, we have used the fact that
\bea
[\frac{g_4+g_5}{n}]+[\frac{g_3+(g_4+g_5)\mod n}{n}]=[\frac{g_3+g_4+g_5}{n}]=[\frac{(g_3+g_4)\mod n+g_5}{n}]+[\frac{g_3+g_4}{n}].
\eea
This is in fact the cocycle condition for $\phi$.
Hence 
\bea
\delta\tilde{\beta}_4(h_1,h_2,h_3,h_4,h_5)=\zeta_n^X=\zeta_n^{g_1[\frac{g_2+g_3}{n}][\frac{g_4+g_5}{n}]}=\alpha_5(g_1,g_2,g_3,g_4,g_5)=\tilde{\alpha}_5(h_1,h_2,h_3,h_4,h_5).
\eea
Therefore, $\tilde{\alpha}_5=\delta\tilde{\beta}_4$.

\subsection{Any odd $d\ge3$ and any $n\ge2$:
Find $\tilde{\beta}_{d-1}$
such that $\tilde{\alpha}_d=\delta\tilde{\beta}_{d-1}$}

Explicitly, for odd $d\ge3$, any $d$-cocycle $\alpha_d\in \H^d(\Z_n,\U(1))$ has the form \cite{1703.03266}
\bea
\alpha_d(g_1,g_2,\dots,g_d)=\zeta_n^{ g_1 [\frac{g_2+g_3}{n}] [\frac{g_4+g_5}{n}] \cdots [\frac{g_{d-1}+g_d}{n}]}
\eea
where $g_i\in\Z_n$ for $i=1,2,\dots,d$, $\zeta_n$ is an $n$-th root of unity, and $[\frac{p}{q}]$ denotes the integer part of $\frac{p}{q}$.

We can find an explicit $(d-1)$-cochain $\tilde{\beta}_{d-1}\in C^{d-1}(\Z_{n^2},\U(1))$ such that $\tilde{\alpha}_d=\delta\tilde{\beta}_{d-1}$ where $\tilde{\alpha}_d=f^*\alpha_d\in \H^d(\Z_{n^2},\U(1))$ and $f$ is the map in the extension $1\to\Z_n\to\Z_{n^2}\xrightarrow{f}\Z_n\to1$.
Namely,
\bea
\tilde{\alpha}_d(h_1,h_2,\dots,h_d)&=&\delta\tilde{\beta}_{d-1}(h_1,h_2,\dots,h_d)\cr
&=&\frac{\tilde{\beta}_{d-1}(h_2,h_3,\dots,h_d)\tilde{\beta}_{d-1}(h_1,h_2h_3,h_4,\dots,h_d)\cdots\tilde{\beta}_{d-1}(h_1,h_2,\dots,h_{d-2},h_{d-1}h_d)}{\tilde{\beta}_{d-1}(h_1h_2,h_3,h_4,\dots,h_d)\tilde{\beta}_{d-1}(h_1,h_2,h_3h_4,h_5,\dots,h_d)\cdots\tilde{\beta}_{d-1}(h_1,h_2,\dots,h_{d-1})}
\eea
where $h_i=(g_i,k_i)\in\Z_{n^2}$ for $i=1,2,\dots,d$ and 
\bea
(g_1,k_1)\cdot(g_2,k_2)=(g_1+g_2,k_1+k_2+[\frac{g_1+g_2}{n}]).
\eea
Here, $\phi(g_1,g_2)=[\frac{g_1+g_2}{n}]$ is the 2-cocycle in $\H^2(\Z_n,\Z_n)$ classifying the extension $1\to\Z_n\to\Z_{n^2}\to\Z_n\to1$.
Note that for $n=2$, $[\frac{g_1+g_2}{2}]=g_1g_2\mod 2$.

Explicitly, the $(d-1)$-cochain $\tilde{\beta}_{d-1}\in C^{d-1}(\Z_{n^2},\U(1))$ is given by
\bea
\tilde{\beta}_{d-1}(h_1,h_2,\dots,h_{d-1})=\zeta_n^{g_1k_2[\frac{g_3+g_4}{n}]\cdots[\frac{g_{d-2}+g_{d-1}}{n}]}.
\eea
If we write $\tilde{\beta}_{d-1}=\zeta_n^{\tilde{\gamma}_{d-1}}$ and $\tilde{\beta}_2=\zeta_n^{\tilde{\gamma}_2}$, then
\bea
\tilde{\gamma}_{d-1}=\tilde{\gamma}_2\smile \phi^{\frac{d-3}{2}}.
\eea
If we write $\alpha_d=\zeta_n^{\epsilon_d}$ and $\alpha_3=\zeta_n^{\epsilon_3}$, then
\bea
\epsilon_d=\epsilon_3\smile \phi^{\frac{d-3}{2}}.
\eea
We have shown that $\delta \tilde{\gamma}_2(h_1,h_2,h_3)=\epsilon_3(g_1,g_2,g_3)$ and we have $\delta\phi=0$.
Hence
\bea
\delta\tilde{\beta}_{d-1}(h_1,h_2,\dots,h_d)
&=&\zeta_n^{\delta\tilde{\gamma}_{d-1}}(h_1,h_2,\dots,h_d)\cr
&=&\zeta_n^{\delta(\tilde{\gamma}_2\smile \phi^{\frac{d-3}{2}})}(h_1,h_2,\dots,h_d)\cr
&=&\zeta_n^{\epsilon_3\smile \phi^{\frac{d-3}{2}}}(g_1,g_2,\dots,g_d)\cr
&=&\zeta_n^{\epsilon_d}(g_1,g_2,\dots,g_d)\cr
&=&\alpha_d(g_1,g_2,\dots,g_d)\cr
&=&\tilde{\alpha}_d(h_1,h_2,\dots,h_d).
\eea
Therefore, $\tilde{\alpha}_d=\delta\tilde{\beta}_{d-1}$.

\section{$d$d-bulk/$(d-1)$d-boundary coupled invertible topological field theory /symmetric anomalous gapped TQFT
by the symmetry extension $1\to\Z_n\to\Z_{n^2}\to\Z_n\to 1$ for any odd $d\ge3$ and any $n\ge2$}

\label{sec:dd-d-1d}

In Appendix \ref{sec:trivialized}, we prove that the group cohomology class 
($\H^d(\Z_n,\U(1))\cong \Z_n$) 
can be canceled by 
anomalous $G$-symmetric $K=\Z_n$-gauge $(d-1)$d TQFTs for odd $d\ge3$ and any $n\ge2$, via the 
appropriate symmetry-extension construction \cite{Wang2017locWWW1705.06728} of
\bea
1\to K \to G_{\rm Tot}\to G \to 1.
\eea
as 
\bea \label{eq:Zn-extension-general}
1\to\Z_n\to\Z_{n^2}\to\Z_n\to 1.
\eea

More generally and mathematically, 
in this work, for odd $d\ge3$ and any $n\ge2$,
we prove that any group cohomology cocycle 
\bea
\alpha_d \in \H^d(\Z_n,\U(1)) \cong \Z_n
\eea 
is trivialized by the group extension 
as \eq{eq:Zn-extension-general}'s $1\to\Z_n\to\Z_{n^2}\to\Z_n\to 1$ \cite{Wang2017locWWW1705.06728}.

In Appendix \ref{sec:cochain}, we find an explicit $(d-1)$-cochain $\beta_{d-1}$ 
that splits the $d$-cocycle $\alpha_d$ by that extension for odd $d\ge3$ and any $n\ge2$.
Namely, $\alpha_d=\delta \beta_{d-1}$ holds when pulling back
the quotient $\Z_n$ to the total $\Z_{n^2}$ group,
from the cocycle $\alpha_d$ in $\H^d(\Z_n,\U(1))$ to the coboundary 
\bea
\text{$\tilde{\alpha}_d=\delta \tilde{\beta}_{d-1}$ in
$\H^d(\Z_{n^2},\U(1))$.}
\eea

More explicitly, for the 
$\alpha_d$ given by
\bea
\alpha_d=
 \exp \big( \ii \frac{2 \pi}{n}  \int_{M^d} ( 
   {A_{\Z_n}} (\beta_{(n,n)} {A_{\Z_n}})^{\frac{d-1}{2}}
 ) \big),
\eea  
we have $\beta_{d-1}$ with 
$\alpha_d= \delta \beta_{d-1}$,
obtained in Appendix \ref{sec:cochain}, which suggests a construction of the $d$d iTFT on the bulk $d$-manifold $M^d$
and the $(d-1)$d noninvertible TQFT 
on the $(d-1)$d boundary $M^{d-1}= \prt M^d$ 
with dynamical 
1-cochain gauge field $a_{\Z_n} \in C^1(\B\Z_n,\Z_n)$
and $(d-3)$-cochain
(dual) gauge field $b_{\Z_n} \in C^{d-3}(\B\Z_n,\Z_n)$, such that the full $d$d/$(d-1)$d coupled path integral is given by
\bea \label{eq:dd-d-1d}
&&\exp \big( \ii \frac{2 \pi}{n}  \int_{M^d} ( 
   {A_{\Z_n}} (\beta_{(n,n)} {A_{\Z_n}})^{\frac{d-1}{2}}
 ) \big) \cdot 
 \cr
&& \cdot \sum_{
 \substack{
 a_{\Z_n} \in C^1(\B\Z_n,\Z_n) \\
 b_{\Z_n} \in C^{d-3}(\B\Z_n,\Z_n)
 }}\exp \big( \ii \frac{2 \pi}{n}  \int_{M^{d-1}=\prt M^d} ( 
 b_{\Z_n} \dd a_{\Z_n}
 - b_{\Z_n} \beta_{(n,n)} {A_{\Z_n}}
 - a_{\Z_n} {A_{\Z_n}}
 (\beta_{(n,n)} {A_{\Z_n}})^{\frac{d-3}{2}}
  ) \big) \cr
  &=&\exp \big( \ii \frac{2 \pi}{n}  \int_{M^d} ( 
   {A_{\Z_n}} (\beta_{(n,n)} {A_{\Z_n}})^{\frac{d-1}{2}}
 ) \big) \cdot
 \cr
&& \cdot\sum_{
 \substack{
 a_{\Z_n} \in C^1(\B\Z_n,\Z_n) \\
 b_{\Z_n} \in C^{d-3}(\B\Z_n,\Z_n)
 }}\exp \big( \ii \frac{2 \pi}{n}  \int_{M^{d-1}=\prt M^d}(  
  a_{\Z_n}(\dd  b_{\Z_n}- {A_{\Z_n}} (\beta_{(n,n)} {A_{\Z_n}})^{\frac{d-3}{2}})-b_{\Z_n}\beta_{(n,n)} {A_{\Z_n}}
  ) \big).
\eea
This $d$d/$(d-1)$d coupled path integral analogously matches the discrete cocycle forms or cochain forms (e.g., \cite{Wang1405.7689})
derived in Appendix \ref{sec:cochain}, as the $d$-cocycle
$$\alpha_d(g_1,g_2,\dots,g_d)=\zeta_n^{g_1 [\frac{g_2+g_3}{n}]\cdots [\frac{g_{d-1}+g_d}{n}]}$$
and the $(d-1)$-cochain
$$
\tilde{\beta}_{d-1}(h_1,h_2,\dots,h_{d-1})=\zeta_n^{g_1k_2[\frac{g_3+g_4}{n}]\cdots[\frac{g_{d-2}+g_{d-1}}{n}]},
$$
where
$\zeta_n$ is an $n$-th root of unity
such as $\zeta_n = \exp(\frac{2 \pi \ii}{n})$,
with variables 
$g \in \Z_n$ and $k \in \Z_n$.

The $d$d bulk partition function on a $d$d manifold with a $(d-1)$d boundary is not gauge-invariant,
but the full $d$d/$(d-1)$d coupled path integral \eq{eq:dd-d-1d} is gauge-invariant under
\bea\label{eq:gauge-transformation-general}
&& A_{\Z_n} \mapsto A_{\Z_n} + \dd \lambda_{0,\Z_n}, \cr
&& a_{\Z_n} \mapsto a_{\Z_n}  + \dd \mu_{0,\Z_n}, \cr
&& b_{\Z_n}  \mapsto b_{\Z_n} +\lambda_{0,\Z_n}(\beta_{(n,n)}A_{\Z_n})^{\frac{d-3}{2}}  + \dd \mu_{d-4,\Z_n}, \cr
&& A_{\Z_n} \in\H^1(\B\Z_n,\Z_n)= \Z_n,\cr
&& a_{\Z_n} \in C^1(\B\Z_n,\Z_n),\cr
&&  b_{\Z_n} \in C^{d-3}(\B\Z_n,\Z_n),\cr
&&\lambda_{0,\Z_n} \in C^0(\B\Z_n,\Z_n),\cr
&&\mu_{0,\Z_n} \in C^0(\B\Z_n,\Z_n),\cr
&&\mu_{d-4,\Z_n} \in C^{d-4}(\B\Z_n,\Z_n).
\eea
Here, $C^k(\B\Z_n,\Z_n)$ is the group of $\Z_n$-valued $k$-cochains of the classifying space $\B\Z_n$. The cohomology group $\H^k(\B\Z_n,\Z_n)$ is defined as the quotient group $Z^k(\B\Z_n,\Z_n)/B^k(\B\Z_n,\Z_n)$ where $Z^k(\B\Z_n,\Z_n)$ is the group of $\Z_n$-valued $k$-cocycles of the classifying space $\B\Z_n$ and $B^k(\B\Z_n,\Z_n)$ is the group of $\Z_n$-valued $k$-coboundaries of the classifying space $\B\Z_n$.

Below we check that \eqref{eq:dd-d-1d} is gauge-invariant under \eqref{eq:gauge-transformation-general}.
Because $ \beta_{(n,n)} \dd \lambda_{0,\Z_n}=0$, $ \beta_{(n,n)} {A_{\Z_n}}$ and $\dd  b_{\Z_n}- {A_{\Z_n}} (\beta_{(n,n)} {A_{\Z_n}})^{\frac{d-3}{2}}$ are gauge-invariant under the gauge transformation \eqref{eq:gauge-transformation-general}, hence \eqref{eq:dd-d-1d} transforms under \eqref{eq:gauge-transformation-general} as
 \bea
&& \exp \big( \ii \frac{2 \pi}{n}  \int_{M^d} ( 
   {A_{\Z_n}} (\beta_{(n,n)} {A_{\Z_n}})^{\frac{d-1}{2}}
 ) \big) \cdot 
 \cr
&& \cdot\sum_{
 \substack{
 a_{\Z_n} \in C^1(\B\Z_n,\Z_n) \\
 b_{\Z_n} \in C^{d-3}(\B\Z_n,\Z_n)
 }}\exp \big( \ii \frac{2 \pi}{n}  \int_{M^{d-1}=\prt M^d}(  
 a_{\Z_n}(\dd  b_{\Z_n}- {A_{\Z_n}} (\beta_{(n,n)} {A_{\Z_n}})^{\frac{d-3}{2}})-b_{\Z_n}\beta_{(n,n)} {A_{\Z_n}}
  ) \big)\cr
  &\mapsto&
  \exp \big( \ii \frac{2 \pi}{n}  \int_{M^d} ( 
   ({A_{\Z_n}}+\dd \lambda_{0,\Z_n}) (\beta_{(n,n)} {A_{\Z_n}})^{\frac{d-1}{2}} 
 ) \big) \cdot 
 \cr
&& \cdot\sum_{
 \substack{
 a_{\Z_n} \in C^1(\B\Z_n,\Z_n) \\
 b_{\Z_n} \in C^{d-3}(\B\Z_n,\Z_n)
 }}\exp \big( \ii \frac{2 \pi}{n}  \int_{M^{d-1}=\prt M^d}(  
 (a_{\Z_n}+\dd\mu_{0,\Z_n})(\dd  b_{\Z_n}- {A_{\Z_n}} (\beta_{(n,n)} {A_{\Z_n}})^{\frac{d-3}{2}})\cr
 &&-(b_{\Z_n}+\lambda_{0,\Z_n}(\beta_{(n,n)}A_{\Z_n})^{\frac{d-3}{2}}  + \dd \mu_{d-4,\Z_n})\beta_{(n,n)} {A_{\Z_n}}
  ) \big).
 \eea
Since by the Stokes theorem, we have
\bea
\int_{M^{d-1}=\prt M^d}(\dd\mu_{0,\Z_n})(\dd  b_{\Z_n}- {A_{\Z_n}} (\beta_{(n,n)} {A_{\Z_n}})^{\frac{d-3}{2}})=0,
\eea
\bea
\int_{M^{d-1}=\prt M^d}(\dd \mu_{d-4,\Z_n})\beta_{(n,n)} {A_{\Z_n}}=0,
\eea
and 
\bea
\int_{M^{d-1}=\prt M^d}\lambda_{0,\Z_n}(\beta_{(n,n)}A_{\Z_n})^{\frac{d-1}{2}}=\int_{M^d}(\dd\lambda_{0,\Z_n})(\beta_{(n,n)}A_{\Z_n})^{\frac{d-1}{2}},
\eea
\eqref{eq:dd-d-1d} is gauge-invariant under \eqref{eq:gauge-transformation-general}.

\section{3+1d Nonperturbative Global Anomaly in 
$\Spin \times {\Z_{n}}$ for integer $n$ with $2 \nmid n$ and $3 \nmid n$}

\label{sec:anomaly-nmid23}

In this appendix, we explore the
{3+1d nonperturbative global anomaly for a Weyl fermion in 
$\Spin \times {\Z_{n}}$ for integer $n$ with $2 \nmid n$ and $3 \nmid n$}, compared with the fact that
\bea
 && \Omega_5^{\rm Spin}(\B\Z_{n}  )
\cong \tilde{\Omega}_5^{\rm SO}(\B\Z_{n}  )
\cong 
\Z_n \oplus 
\Z_n, \quad 2\nmid n, \quad 3\nmid n.
\eea
Here $\tilde{\Omega}_5^{\rm SO}(\B G):=\Omega_5^{\rm SO}(\B G)/\Omega_5^{\rm SO}$ is the reduced bordism group, modding out the $\Omega_5^{\rm SO}=\Omega_5^{\rm SO}(pt)$.

The perturbative local anomaly of U(1) charge $q=1$ left-handed Weyl fermion 
of $\Spin \times \U(1)$ symmetry
in 3+1d or 4d
is captured by a 5d invertible field theory (iTFT) with the anomaly index $k=1$:
\begin{equation}
   \exp( \ii k  \int_{M^5} A \frac{c_1^2}{6}-A \frac{p_1}{24}).
\end{equation} 
Now we redefine the U(1) gauge field $A$ as a $\Z_n$ gauge field $ A_{\Z_n} \in\H^1(\B\Z_n,\Z_n)= \Z_n$
with the following replacement:
\bea
A &\mapsto& \frac{2 \pi}{n}   A_{\Z_n}.\cr 
c_1 = \frac{\dd A}{2 \pi} &\mapsto& 
\frac{\dd   A_{\Z_n}}{n} \equiv \beta_{(n,n) } A_{\Z_n}. 
\eea
The $\beta_{(n,m)}: \H^*(-,\Z_m) \mapsto 
\H^{*+1}(-,\Z_n) $ 
is the Bockstein 
 homomorphism associated with the extension
 $\Z_n \stackrel{\cdot m}{\to} \Z_{nm} \to \Z_m$. Thus we get the 5d topological invariant of the $\Spin \times \Z_n$ as:
\bea\label{eq:Spin-Zn}
   &&\exp \big( \ii 2 \pi k  \int_{M^5} 
   ( \frac{1}{6n} 
   {A_{\Z_n}} (\beta_{(n,n)} {A_{\Z_n}}) (\beta_{(n,n)} {A_{\Z_n}})
   -\frac{1}{ 24n} {A_{\Z_n}} p_1 ) \big).
\eea
Since the anomaly of 4d Weyl fermion with symmetry $\Spin\times\Z_n$ does not contain 2-torsion and 3-torsion for $2\nmid n$ and $3\nmid n$ \cite{Kapustin1406.7329, 2016arXiv160406527F, GarciaEtxebarriaMontero2018ajm1808.00009, 1808.02881, GuoJW1812.11959, 2506.19710}, we can regard $6=2\cdot3$ and $24=2^3\cdot3$ as invertible in $\Z_n$ for $2\nmid n$ and $3\nmid n$.

In fact, for $2\nmid n$ and $3\nmid n$, by the Chinese remainder theorem, there exists an integer $x_n$ such that $x_n=1\mod n$ and $24|x_n$. Then we can rewrite \eqref{eq:Spin-Zn} as 
\bea \label{eq:D5}
\exp \big( \ii \frac{2 \pi k }{n} \int_{M^5} 
   ( \frac{x_n}{6} 
   {A_{\Z_n}} (\beta_{(n,n)} {A_{\Z_n}}) (\beta_{(n,n)} {A_{\Z_n}})
   -\frac{x_n}{ 24} {A_{\Z_n}} p_1 ) \big).
\eea
Since $\gcd(n,x_n)=1$,
$\gcd(n,\frac{x_n}{6})=1$, and 
$\gcd(n,\frac{x_n}{24})=1$, 
the first term 
and 
the second term
individually in \eq{eq:D5}, as well as the combined two terms
in \eq{eq:D5},
all generate a $\Z_n$ class.
We expect 
\eq{eq:D5}
as a schematic way to write one of the two cobordism invariant generators of 
$\Omega_5^{\rm Spin}(\B\Z_{n}  )
\cong 
\Z_n \oplus 
\Z_n,$ with $2\nmid n$ and $3\nmid n.$

\section{3+1d Nonperturbative Global Anomaly in 
$\Spin \times \Z_{3^r} = \Spin \times_{\Z_2^\rF} {\Z_{2 \cdot 3^r}^\rF}$}

\label{sec:Z3r}

In this appendix, we explore the
{3+1d nonperturbative global anomaly for a Weyl fermion in 
$\Spin \times \Z_{3^r} = \Spin \times_{\Z_2^\rF} {\Z_{2 \cdot 3^r}^\rF}$}, compared with the fact that
\bea\label{eq:E1}
&&\Omega_5^{{\rm Spin} \times \Z_{3^r}}
\cong
\Omega_5^{{\rm Spin} \times_{\Z_2^\rF} {\Z_{2 \cdot  3^r 
}}} \cong
 \tilde{\Omega}_5^{\rm SO}(\B\Z_{3^r }  )
 \cr
&&=
\Z_{3^{r+1}}\oplus \Z_{3^{r-1}} 
.  
\eea
Here $\tilde{\Omega}_5^{\rm SO}(\B G):=\Omega_5^{\rm SO}(\B G)/\Omega_5^{\rm SO}$ is the reduced bordism group, modding out the $\Omega_5^{\rm SO}=\Omega_5^{\rm SO}(pt)$.

We also prove that the $k=3^r$ anomaly of the 4d Weyl fermion with $\Spin\times\Z_{3^r}$ symmetry can be trivialized by a $\Z_3$ extension.

The perturbative local anomaly of U(1) charge $q=1$ left-handed Weyl fermion 
of $\Spin \times \U(1)$ symmetry
in 3+1d or 4d
is captured by a 5d invertible field theory (iTFT) with the anomaly index $k=1$:
\begin{equation}
   \exp( \ii k  \int_{M^5} A \frac{c_1^2}{6}-A \frac{p_1}{24}).
\end{equation} 
Now we redefine the U(1) gauge field $A$ as a $\Z_{3^r}$ gauge field $ A_{\Z_{3^r}} \in\H^1(\B\Z_{3^r},\Z_{3^r})= \Z_{3^r}$
with the following replacement:
\bea
A &\mapsto& \frac{2 \pi}{3^r}   A_{\Z_{3^r}}.\cr 
c_1 = \frac{\dd A}{2 \pi} &\mapsto& 
\frac{\dd   A_{\Z_{3^r}}}{3^r} \equiv \beta_{(3^r,3^r) } A_{\Z_{3^r}}. 
\eea
The $\beta_{(n,m)}: \H^*(-,\Z_m) \mapsto 
\H^{*+1}(-,\Z_n) $ 
is the Bockstein 
 homomorphism associated with the extension
 $\Z_n \stackrel{\cdot m}{\to} \Z_{nm} \to \Z_m$. Thus we get the 5d topological invariant of the $\Spin \times \Z_{3^r}$ as:
\bea\label{eq:Spin-Z3r}
   &&\exp \big( \ii 2 \pi k  \int_{M^5} 
   ( \frac{1}{2\cdot 3^{r+1}} 
   {A_{\Z_{3^r}}} (\beta_{(3^r,3^r)} {A_{\Z_{3^r}}}) (\beta_{(3^r,3^r)} {A_{\Z_{3^r}}})
   -\frac{1}{ 8\cdot 3^{r+1}} {A_{\Z_{3^r}}} p_1 ) \big).
\eea
Since the anomaly of 4d Weyl fermion with symmetry $\Spin\times\Z_{3^r}$ contains only 3-torsion \cite{Kapustin1406.7329, 2016arXiv160406527F, GarciaEtxebarriaMontero2018ajm1808.00009, 1808.02881, GuoJW1812.11959, 2506.19710}, we can regard $2$ and $8$ as invertible in $\Z_{3^{r+1}}$.

In fact, by the Chinese remainder theorem, there exists an integer $y_r$ such that $y_r=1\mod 3^{r+1}$ and $8|y_r$. So $\gcd(3^{r+1},y_r)=1$, $\gcd(3^{r+1},\frac{y_r}{2})=1$, and $\gcd(3^{r+1},\frac{y_r}{8})=1$. Then we can rewrite \eqref{eq:Spin-Z3r} as 
\bea\label{eq:E5}
\exp \big( \ii \frac{2 \pi k }{3^{r+1}} \int_{M^5} 
   ( \frac{y_r}{2} 
   {A_{\Z_{3^r}}} (\beta_{(3^r,3^r)} {A_{\Z_{3^r}}}) (\beta_{(3^r,3^r)} {A_{\Z_{3^r}}})
   -\frac{y_r}{ 8} {A_{\Z_{3^r}}} p_1 ) \big).
\eea

Since $A_{\Z_{3^r}}=A_{\Z_3}\mod3$ (see footnote \ref{footnote-mod3}) and $A_{\Z_3}p_1=0\mod3$ \cite{tomonaga1964mod,tomonaga1965pontryagin} (see Appendix \ref{sec:proof} for the proof), we have $A_{\Z_{3^r}}p_1=0\mod3$ and we can rewrite \eq{eq:E5} as
\bea\label{eq:E6}
\exp \big( \ii \frac{2 \pi k }{3^{r+1}} \int_{M^5} 
   ( \frac{y_r}{2} 
   {A_{\Z_{3^r}}} (\beta_{(3^r,3^r)} {A_{\Z_{3^r}}}) (\beta_{(3^r,3^r)} {A_{\Z_{3^r}}})
   -3\cdot\frac{y_r}{ 8} \frac{{A_{\Z_{3^r}}} p_1}{3} ) \big).
\eea
The $k=3^r$ anomaly of the 4d Weyl fermion with $\Spin\times\Z_{3^r}$ symmetry is 
\bea\label{eq:E7}
\exp(\ii \frac{2\pi}{3}\int_{M^5}\frac{y_r}{2}{A_{\Z_{3^r}}} (\beta_{(3^r,3^r)} {A_{\Z_{3^r}}}) (\beta_{(3^r,3^r)} {A_{\Z_{3^r}}})).
\eea
Since ${A_{\Z_{3^r}}} (\beta_{(3^r,3^r)} {A_{\Z_{3^r}}}) (\beta_{(3^r,3^r)} {A_{\Z_{3^r}}})={A_{\Z_{3}}} (\beta_{(3,3)} {A_{\Z_{3}}}) (\beta_{(3,3)} {A_{\Z_{3}}})\mod3$\footnote{
Physically, if $l\mid n$, when the anomaly is evaluated modulo $l$, only the
$\mathbb Z_l$ reduction of the $\mathbb Z_{n}$ background gauge field is
relevant. Let $A_{\mathbb Z_{n}}\in \H^1(\B\mathbb Z_{n},\Z_{n})$ be the
background field for a $\mathbb Z_{n}$ symmetry. Under the natural quotient
$\mathbb Z_{n}\to\mathbb Z_l$, it reduces to the $\mathbb Z_l$ background
field $A_{\mathbb Z_l}$. The associated Bockstein homomorphism is compatible with
this reduction, i.e. the diagram
\bea
\xymatrix{
\H^*(-,\Z_{n}) \ar[d]^{\mathrm{mod}\,l}
\ar[r]^{\beta_{(n,n)}} &
\H^{*+1}(-,\Z_{n}) \ar[d]^{\mathrm{mod}\,l} \\
\H^*(-,\Z_l) \ar[r]^{\beta_{(l,l)}} &
\H^{*+1}(-,\Z_l)
}
\eea
commutes. Consequently, after reduction modulo $l$, every occurrence of
$A_{\mathbb Z_{n}}$ and $\beta_{(n,n)}$ may be replaced by
$A_{\mathbb Z_l}$ and $\beta_{(l,l)}$, respectively. Therefore
\bea
A_{\mathbb Z_{n}}
\,(\beta_{(n,n)}A_{\mathbb Z_{n}})^\ell
\equiv
A_{\mathbb Z_l}
\,(\beta_{(l,l)}A_{\mathbb Z_l})^\ell
\pmod l,
\eea
for arbitrary integer power $\ell$.
In particular, this applies to $n=3^r$, $l=3$ and
$\ell=2$ in our case.
\label{footnote-mod3}}, this term generates a $\Z_3$ class. Therefore, the combined two terms in \eq{eq:E6} generate a $\Z_{3^{r+1}}$ class.
We expect 
\eq{eq:E6}
as a schematic way to write the first one of the two cobordism invariant generators of 
$\Omega_5^{\rm Spin}(\B\Z_{3^r}  )
\cong 
\Z_{3^{r+1}} \oplus 
\Z_{3^{r-1}}$.

\subsection{Theorem}

By the results in Appendix \ref{sec:trivialized}, the $k=3^r$ anomaly \eq{eq:E7} 
of the 4d Weyl fermion  with $\Spin\times\Z_{3^r}$ symmetry, namely 
$k=3^r \in 
\Z_{3^{r+1}} 
\subset 
 \Omega_5^{\rm Spin}(\B\Z_{3^r }  )
\cong
\Z_{3^{r+1}}\oplus \Z_{3^{r-1}}$,
can be trivialized by a $\Z_3$ extension. The $\Z_3$ extension is minimal because \eq{eq:E7} is a nontrivial class. We have proved the following theorem:

\begin{theorem}
    The anomaly of $3^r$ copies of Weyl fermions with charge $q=1$ in $\TP_5(\Spin\times\Z_{3^r})$ becomes trivialized by the $\Z_3$ extension in $\TP_5(\Spin\times\Z_{3^{r+1}})$, for any positive integer $r \geq 1$:
    \bea
    1\to K=\Z_3\to G_{\rm{Tot}}=\Spin\times\Z_{3^{r+1}}\to G=\Spin\times\Z_{3^r}\to1,
    \eea
    and the $K=\Z_3$ extension is minimal.
\end{theorem}

\section{3+1d Nonperturbative Global Anomaly in 
$ \Spin \times_{\Z_2^\rF} {\Z_{2 \cdot 2^p}^\rF}$}

\label{sec:Z2p}

In this appendix, we explore the
{3+1d nonperturbative global anomaly for a Weyl fermion in 
$ \Spin \times_{\Z_2^\rF} {\Z_{2 \cdot 2^p}^\rF}$}, compared with the fact that
\bea
\Omega_5^{{\rm Spin} \times_{\Z_2^\rF} {\Z_{2 \cdot  2^p 
}}} 
=
\Z_{2^{p+3}}\oplus \Z_{2^{p-1}} 
.  
\eea

The perturbative local anomaly of U(1) charge $q=1$ left-handed Weyl fermion 
of $\Spin^c=\Spin \times_{\Z_2^{\rF}} \U(1)$ symmetry
in 3+1d or 4d
is captured by a 5d invertible field theory (iTFT) with the anomaly index $k=1$:
\begin{equation}
   \exp( \ii k  \int_{M^5} A' \frac{(2c_1)^2}{48}-A' \frac{p_1}{48}).
\end{equation} 
Here, $c_1'=2c_1=\frac{\dd A'}{2\pi}$ is the first Chern class of the $\U(1)/\Z_2^{\rF}$ bundle.
Now we redefine the U(1) gauge field $A'$ as a $\Z_{2^p}$ gauge field $ A_{\Z_{2^p}} \in\H^1(\B\Z_{2^p},\Z_{2^p})= \Z_{2^p}$
with the following replacement:
\bea
A' &\mapsto& \frac{2 \pi}{2^p}   A_{\Z_{2^p}}.\cr 
c_1' = \frac{\dd A'}{2 \pi} &\mapsto& 
\frac{\dd   A_{\Z_{2^p}}}{2^p} \equiv \beta_{(2^p,2^p) } A_{\Z_{2^p}}. 
\eea
The $\beta_{(n,m)}: \H^*(-,\Z_m) \mapsto 
\H^{*+1}(-,\Z_n) $ 
is the Bockstein 
 homomorphism associated with the extension
 $\Z_n \stackrel{\cdot m}{\to} \Z_{nm} \to \Z_m$. Thus we get the 5d topological invariant of the $\Spin \times_{\Z_2^{\rF}} \Z_{2\cdot 2^p}$ as:
\bea\label{eq:Spin-Z2p}
   &&\exp \big( \ii 2 \pi k  \int_{M^5} 
   ( \frac{1}{3\cdot 2^{p+4}} 
   {A_{\Z_{2^p}}} (\beta_{(2^p,2^p)} {A_{\Z_{2^p}}}) (\beta_{(2^p,2^p)} {A_{\Z_{2^p}}})
   -\frac{1}{ 3\cdot 2^{p+4}} {A_{\Z_{2^p}}} p_1 ) \big).
\eea
Since the anomaly of 4d Weyl fermion with symmetry $\Spin\times_{\Z_2^{\rF}}\Z_{2\cdot 2^p}$ contains only 2-torsion \cite{Kapustin1406.7329, 2016arXiv160406527F, GarciaEtxebarriaMontero2018ajm1808.00009, 1808.02881, GuoJW1812.11959, 2506.19710}, we can regard $3$ as invertible in $\Z_{2^{p+3}}$.

The bundle constraint for the $\Spin \times_{\Z_2^{\rF}} \Z_{2\cdot 2^p}$ structure is
\bea
w_2=\beta_{(2^p,2^p)} {A_{\Z_{2^p}}}\mod2
\eea
where $w_2$ is the second Stiefel-Whitney class of the tangent bundle $TM$. 
Therefore, 
\bea
{A_{\Z_{2^p}}} (\beta_{(2^p,2^p)} {A_{\Z_{2^p}}}) (\beta_{(2^p,2^p)} {A_{\Z_{2^p}}}={A_{\Z_{2^p}}} w_2^2={A_{\Z_{2^p}}} p_1\mod2.
\eea
So we can rewrite \eqref{eq:Spin-Z2p} as
\bea\label{eq:Spin-Z2p-rewrite}
   &&\exp \big( \ii 2 \pi k \frac{1}{3\cdot 2^{p+3}} \int_{M^5} 
   \frac{1}{2}(  
   {A_{\Z_{2^p}}} (\beta_{(2^p,2^p)} {A_{\Z_{2^p}}}) (\beta_{(2^p,2^p)} {A_{\Z_{2^p}}})
   - {A_{\Z_{2^p}}} p_1 ) \big).
\eea
We expect 
\eq{eq:Spin-Z2p-rewrite}
as a schematic way to write the first one of the two cobordism invariant generators of 
$\Omega_5^{{\rm Spin} \times_{\Z_2^\rF} {\Z_{2 \cdot  2^p 
}}} 
=
\Z_{2^{p+3}}\oplus \Z_{2^{p-1}} $.

\section{Proof of $A_{\Z_3}p_1=0\mod3$}

\label{sec:proof}

In this appendix, we prove a general fact which implies $A_{\Z_3}p_1=0\mod3$ \cite{tomonaga1964mod,tomonaga1965pontryagin} and explain why 3 is special. Since $\H^1(\Z_n,\U(1))\cong \H^1(\Z_n,\Z_n)$, we can regard $A_{\Z_n}$ as the generator of $\H^1(\Z_n,\Z_n)$. The cup product $A_{\Z_n}p_1$ is a mod $n$ cohomology class on a manifold $M$ when $A_{\Z_n}$ is pulled back to $M$.

\subsection{A general fact about the mod $q$ Steenrod reduced power $P_q^r$
}\label{sec:proof-general}

Let $P_q^r:\H^i(-,\Z_q)\to \H^{i+2(q-1)r}(-,\Z_q)$ be the mod $q$ Steenrod reduced power where $q$ is an odd prime.

On an oriented $n$-manifold $M^n$, by the Poincar\'e duality, there exists $s_q^r\in \H^{2(q-1)r}(M^n,\Z_q)$ such that 
\bea
P_q^r(x)=s_q^r\smile x,\quad \forall x\in \H^{n-2(q-1)r}(M^n,\Z_q).
\eea
Let $P_q:=\sum_{r=0}^{\infty}P_q^r$ be the mod $q$ total Steenrod reduced power, and $s_q:=\sum_{r=0}^{\infty}s_q^r$.
We prove the following theorem \cite[Theorem 4.3]{hirzebruch1953steenrod} as its proof is hard to find in the literature.

\begin{theorem}\label{theorem-p1}
On an oriented $n$-manifold $M^n$, we have
\bea
P_q(s_q)=\sum_{j=0}^{\infty}b_{q,j}\mod q
\eea
where $\sum_{j=0}^{\infty}b_{q,j}=\prod_i(1+x_i^{q-1})$ and 
$\pm x_i$ are the Chern roots of the complexified tangent bundle $TM\otimes_{\R}\C$. In particular, for $q=3$, 
$b_{3,j}=p_j$ is the $j$-th Pontryagin class of $M^n$.
Equivalently, we have
\bea
\sum_{i+r=j}P_q^i s_q^r=b_{q,j}\mod q.
\eea
\end{theorem} 

We mimic the proof for $\Sq(v)=w$ given in \cite{milnor1974characteristic} to prove the above theorem. Here, $\Sq$ is the total Steenrod square, $v$ is the total mod 2 Wu class, and $w$ is the total Stiefel-Whitney class.
\begin{proof}
    
Let $b_i\in \H^*(M)$ be a basis, and $b_i^{\sharp}\in \H^*(M)$ the dual basis such that $\langle b_i\smile b_j^{\sharp},[M]\rangle =\delta_{ij}$ where $[M]$ is the fundamental class of $M$ and $\langle,\rangle$ is the pairing between cohomology and homology classes.

Then for all $x\in \H^*(M)$, $x=\sum b_i\langle x\smile b_i^{\sharp},[M]\rangle$.
Apply this to $x=s_q$, then 
\bea
s_q=\sum b_i\langle s_q\smile b_i^{\sharp},[M]\rangle=\sum b_i\langle P_q(b_i^{\sharp}),[M]\rangle.
\eea
Therefore,
\bea
P_q(s_q)=\sum P_q(b_i)\langle P_q(b_i^{\sharp}),[M]\rangle.
\eea

Since each Chern root $x_i$ is a degree-2 cohomology class, \bea
P_q(x_i)=x_i+x_i^q=(1+x_i^{q-1})x_i\mod q.
\eea
So, by Cartan's formula, 
\bea
P_q(\prod_i x_i)=\prod_i (1+x_i^{q-1})\prod_i x_i\mod q.
\eea
Therefore, 
\bea
P_q(e)=(\sum_{j=0}^{\infty}b_{q,j})\smile e\mod q
\eea
where $e$ is the Euler class of $M$.
Equivalently,
\bea
P_q^j(e)=b_{q,j}\smile e\mod q.
\eea

Note that $e=\Delta^*(U)$ where $\Delta:M\to M\times M$ is the diagonal map and the diagonal cohomology class
\bea
U=\sum (-1)^{\dim b_i} b_i\times b_i^{\sharp}
\eea
such that 
\bea
U/[M]=\sum (-1)^{\dim b_i} b_i\langle b_i^{\sharp},[M]\rangle=1.
\eea
This is obtained by applying $x=\sum b_i\langle x\smile b_i^{\sharp},[M]\rangle$ to $x=1$ and noting that the only nonvanishing $\langle b_i^{\sharp},[M]\rangle$ occurs when $b_i$ has degree 0, so the sign $(-1)^{\dim b_i}$ disappears.
Here, the slant product is defined as $(a\times b)/[M]:=a\langle b,[M]\rangle$.

By Cartan's formula\footnote{Cartan's formula usually applies to the cup product. The cross product $a\times b$ is defined as $\pi_1^*a\smile \pi_2^*b$ where $\pi_i$ is the projection from $M\times M$ onto its $i$-th factor. Then 
\bea
P_q(a\times b)=P_q(\pi_1^*a\smile \pi_2^*b)=P_q(\pi_1^*a)\smile P_q(\pi_2^*b)=\pi_1^*P_q(a)\smile \pi_2^*P_q(b)=P_q(a)\times P_q(b).
\eea
}, we have 
\bea
P_q(U)=\sum (-1)^{\dim b_i} P_q(b_i)\times P_q(b_i^{\sharp})
\eea
and 
\bea
P_q(U)/[M]=\sum (-1)^{\dim b_i} P_q(b_i)\times P_q(b_i^{\sharp})/[M]=\sum (-1)^{\dim b_i} P_q(b_i)\langle P_q(b_i^{\sharp}),[M]\rangle.
\eea
Since the Steenrod reduced power $P_q$ increases the degree by an even integer, the only nonvanishing $\langle P_q(b_i^{\sharp}),[M]\rangle$ occurs when $b_i^{\sharp}$ has even codimension, \ie $b_i$ has even degree. So the sign $(-1)^{\dim b_i}$ disappears and
\bea
P_q(U)/[M]=\sum P_q(b_i)\langle P_q(b_i^{\sharp}),[M]\rangle.
\eea

On the other hand, since $P_q^j(e)=b_{q,j}\smile e\mod q$ and $e=\Delta^*(U)$, by \cite[Theorem 11.3 and Lemma 11.5]{milnor1974characteristic}, we have
\bea
P_q^j(U)=(b_{q,j}\times 1)\smile U\mod q,
\eea
hence
\bea
P_q(U)=((\sum_{j=0}^{\infty}b_{q,j})\times 1)\smile U\mod q.
\eea
Therefore, 
\bea
P_q(U)/[M]=(((\sum_{j=0}^{\infty}b_{q,j})\times 1)\smile U)/[M]=(\sum_{j=0}^{\infty}b_{q,j})\smile(U/[M])=(\sum_{j=0}^{\infty}b_{q,j})\smile 1=\sum_{j=0}^{\infty}b_{q,j}\mod q.
\eea
So we finally prove that
\bea
P_q(s_q)=\sum P_q(b_i)\langle P_q(b_i^{\sharp}),[M]\rangle=\sum_{j=0}^{\infty}b_{q,j}\mod q.
\eea
\end{proof}

\subsection{$A_{\Z_n}p_1=0\mod n$ on oriented manifolds if and only if $n=3$}\label{sec:proof-3}

\begin{proposition}
Let $A_{\mathbb Z_n}\in \H^1(\B\mathbb Z_n,\mathbb Z_n)$ be the canonical generator.
Then
\bea
A_{\mathbb Z_n}\,p_1=0 \pmod n
\eea
as a universal cohomology operation on oriented manifolds if and only if
$n\mid 3$.

In particular, among all nontrivial cyclic groups $\mathbb Z_n$, the only case
for which
\bea
A_{\mathbb Z_n}\,p_1=0 \pmod n
\eea
holds identically on all oriented manifolds is $n=3$.
\end{proposition}

\begin{proof}
We first show that $n=3$ satisfies the identity.

By Theorem \ref{theorem-p1}, on every oriented manifold one has
\bea
p_1=s_3^1 \pmod 3,
\eea
where $s_3^1$ is the first mod-$3$ Wu class.
Hence
\bea
A_{\mathbb Z_3}\,p_1
=
A_{\mathbb Z_3}\,s_3^1.
\eea
Using the defining property of Wu classes,
\bea
x\,s_3^1=P_3^1(x)
\eea
for every cohomology class $x$.
Taking $x=A_{\mathbb Z_3}$ gives
\bea
A_{\mathbb Z_3}\,p_1
=
P_3^1(A_{\mathbb Z_3}).
\eea
Since $\deg(A_{\mathbb Z_3})=1<2=\deg(P_3^1)$,
the unstable condition for Steenrod reduced powers implies
\bea
P_3^1(A_{\mathbb Z_3})=0.
\eea
Therefore
\bea
A_{\mathbb Z_3}\,p_1=0 \pmod 3.
\eea

It remains to show that no other nontrivial value of $n$ has this property.

Consider the oriented manifold
\bea
M=S^1\times \mathbb{CP}^2.
\eea
Let
\bea
A\in \H^1(S^1,\mathbb Z_n)
\eea
be the generator and denote its pullback to $M$ again by $A$.
Let $h\in \H^2(\mathbb{CP}^2,\mathbb Z)$ be the hyperplane class.
Since the total Chern class is
\bea
c(T\mathbb{CP}^2)=(1+h)^3,
\eea
we have
\bea
c_1=3h,
\qquad
c_2=3h^2.
\eea
Consequently,
\bea
p_1(T\mathbb{CP}^2)
=
c_1^2-2c_2
=
9h^2-6h^2
=
3h^2.
\eea
Hence
\bea
p_1(TM)=3\,\mathrm{pr}_2^*(h^2).
\eea
Here $\mathrm{pr}_2$ is the projection onto the second factor.

Evaluating $A\,p_1$ on the fundamental class of $M$ yields
\bea
\bigl\langle A\,p_1,[M]\bigr\rangle
&=&
\bigl\langle A,[S^1]\bigr\rangle
\bigl\langle p_1,[\mathbb{CP}^2]\bigr\rangle \\
&=&
1\cdot 3 \\
&=&
3
\pmod n.
\eea
Therefore $A_{\mathbb Z_n}p_1$ cannot vanish identically unless
\bea
3\equiv 0 \pmod n,
\eea
that is,
\bea
n\mid 3.
\eea

The only nontrivial positive integer satisfying this condition is $n=3$.
Therefore
\bea
A_{\mathbb Z_n}p_1=0 \pmod n
\eea
holds universally for oriented manifolds if and only if $n=3$.
\end{proof}

\subsection{$A_{\Z_2}p_1=0\mod 2$ on spin manifolds}

\begin{remark}
The proposition is specific to oriented manifolds $M$ (namely $w_1(TM)=0$).

If one restricts to spin manifolds $M$ (namely $w_1(TM)=w_2(TM)=0$, which are also oriented manifolds), the conclusion changes slightly.

Indeed, for every spin manifold $M$,
\bea
p_1(TM)=w_2(TM)^2=0\mod2.
\eea
Hence
\bea
A_{\mathbb Z_2}\,p_1(TM)
=
0
\pmod 2.
\eea
Thus
\bea
A_{\mathbb Z_2}\,p_1=0 \pmod 2
\eea
holds identically on all spin manifolds.

Combining this observation with the proposition, we find that
\bea
A_{\mathbb Z_n}\,p_1=0 \pmod n
\eea
holds universally on all spin manifolds for
$n=2$ and $n=3$.
For oriented manifolds, however, the only nontrivial case is
$n=3$.

For odd prime $q>3$, $\deg(s_q^1)=2(q-1)>4=\deg(p_1)$, so there is no similar result for odd prime $q>3$.
\end{remark}


\bibliography{BSM-TQDM-2512.bib}
\bibliographystyle{JHEP-YM}
\end{document}